\documentclass[11pt]{article}

\usepackage{authblk}

\usepackage[T1]{fontenc}
\usepackage[a4paper,margin=30mm]{geometry}
\usepackage{amsmath,amssymb,amsthm,mathtools}
\usepackage{tikz}
\usepackage{tikz-cd}
\usepackage[most]{tcolorbox}
\usepackage[hidelinks]{hyperref}
\usepackage{quiver}
\usepackage{cleveref}
\usetikzlibrary{patterns,arrows.meta}

\usepackage[colorinlistoftodos,textsize=tiny,linecolor=gray!40]{todonotes}

\newtcolorbox{constructionbox}[1][]{
  enhanced,
  breakable,
  colback=cyan!3,
  colframe=cyan!45!black,
  boxrule=.6pt,
  arc=1.5pt,
  left=7pt,
  right=7pt,
  top=6pt,
  bottom=6pt,
  fonttitle=\bfseries,
  title={#1}
}

\newcommand{\Z}{\mathbb Z}
\newcommand{\F}{\mathbb F}

\newcommand{\cX}{\mathcal X}
\newcommand{\Po}{\mathfrak{P}}

\newcommand{\op}[1]{\operatorname{#1}}

\newtheorem{definition}{Definition}[section]
\newtheorem{example}[definition]{Example}
\newtheorem{remark}[definition]{Remark}
\newtheorem{theorem}[definition]{Theorem}
\newtheorem{proposition}[definition]{Proposition}
\newtheorem{corollary}[definition]{Corollary}
\newtheorem{lemma}[definition]{Lemma}

\newcommand{\ket}[1]{\lvert#1\rangle}

\title{Cups and Gates II:\\ Operads, twists, and addressable
intra-code gates}
\author[1]{Nikolas P. Breuckmann}
\author[2]{Margarita Davydova}
\author[3]{Jens Niklas Eberhardt}
\author[3]{Lucas~Ferdinand Hallwirth}
\author[4]{Nathanan Tantivasadakarn}

\affil[1]{School of Mathematics, University of Bristol, Fry Building,
Woodland Road, Bristol BS8 1UG, United Kingdom}
\affil[2]{Walter Burke Institute for Theoretical Physics, Department of Physics,
and Institute for Quantum Information and Matter, California Institute of Technology,
1200 E California Blvd, Pasadena, CA 91125, USA}
\affil[3]{Institut für Mathematik, Johannes Gutenberg University of Mainz,
Staudingerweg 9, Mainz, 55128, Germany}
\affil[4]{C. N. Yang Institute for Theoretical Physics, Stony Brook University,
Stony Brook, NY 11794, USA}
\date{}

\begin{document}

\maketitle

\begin{abstract}
In Part I, we introduced a general framework for constructing logical gates in quantum CSS codes from cohomology invariants, focusing on \emph{cup products} and their generalizations from topological to qLDPC codes.
In particular, we produced a family of $C^{k-1}Z$ logical gates that we called copy-cup on $k$ copies of the same quantum code implemented by finite-depth circuits.  

Here, we vastly extend this framework by introducing \emph{twists} incorporating code automorphisms, \emph{higher} twisted cohomology operations, and \emph{folding} operations in relative cohomology. The framework includes operations lifting coefficients from $\Z/2$ to $\Z/2^k$, enabling all diagonal logical gates such as single-qubit rotations $\op{R}(2\pi/2^k)$.
Composing these operations with maps that select cohomology classes yields \emph{addressable} logical gates. 
Examples include addressable $T$ gates on three- and four-tori and gates in qLDPC codes beyond topological ones.

We introduce cohomology operations, such as the higher twisted Pontryagin powers, that to our knowledge have not previously appeared in the mathematics literature.
These operations arise from $E_\infty$-structures on cochains, with explicit constructions using cyclic resolutions and the surjection operad.
We expect these results to be of independent interest in homotopy theory. 
Our formalism applies much more broadly to chain complexes carrying suitable operadic structures, and we expect this operadic perspective to uncover further connections between homotopy theory and
fault-tolerant quantum gates.

\end{abstract}

\tableofcontents

\newpage

\newtheorem*{thmA}{Theorem A}
\newtheorem*{thmB}{Theorem B}
\newtheorem*{thmC}{Theorem C}
\newtheorem*{thmD}{Theorem D}

\section{Introduction}
\label{sec:introduction}

Fault-tolerant quantum computing requires the construction of \emph{fault-tolerant logical gates}, which are implemented by a physical circuit that limits the spread of errors.
Logical gates implemented by constant-depth circuits consisting of bounded-size gates are arguably the most efficient way of achieving this. 
Finding a systematic framework for constructing such circuits is the subject of this paper and its first part, Cups and Gates I~\cite{CupsAndGates}.

Early constructions of fault-tolerant logical gates were tied to specific codes, such as the transversal Clifford and $T$ gates in the family of color codes~\cite{Bombin_2007,KubicaYoshidaPastawski15} or families of toric codes in $d$ dimensions~\cite{VasmerBrown3D,4DTCnonClifford}. 
The theorem of Bravyi and K\"onig~\cite{BravyiKoenig2013} bounds the level of the Clifford hierarchy reachable by constant-depth \emph{local} circuits in $D$-dimensional topological codes.
However, it is far less well understood what becomes possible with the long-range connectivity of more general qLDPC codes~\cite{burton2020limitations,FuZhengLiLiu25}. 
With the advent of high-rate qLDPC codes~\cite{BreuckmannEberhardt2021} the question has become pressing, since few of the long-established techniques apply to them directly.\footnote{Other complementary techniques include gate teleportation and magic-state distillation~\cite{GottesmanChuang99,BravyiKitaev05}, code switching~\cite{PaetznickReichardt2013,AndersonDuclosCianciPoulin2014,BombinGaugeCC}, lattice and code surgery~\cite{HorsmanSurgery,CohenSurgery,CowtanSurgery}, pieceable fault tolerance~\cite{YoderPartitioning}, automorphism and fold-transversal gates~\cite{GrasslRoetteler2013,Moussa2016,foldtransversalgates,Quintavalle_2023,EberhardtSteffan2024,SayginelEtAl2024}, and computer search~\cite{WebsterSearch}.}

\paragraph{Cohomology invariants are logical gates.}
In Part I~\cite{CupsAndGates}, we built a framework to construct fault-tolerant logical gates based on \emph{cohomology invariants}, which we recall now.
In a CSS code, physical computational basis states $\ket{a}$ are labeled by cochains $a\in A^q$, while logical states $\ket{[a]}$ arise from the superposition of all representatives of a cohomology class $[a]\in H^q(A)$. A cohomology invariant is a function $\Phi\colon A^q\to\Z/2^s$ whose value on a cocycle depends only on its cohomology class. 
Hence, the diagonal unitary
\[
 U_\Phi\ket a=\exp\bigl(2\pi i\,\Phi(a)/2^s\bigr)\ket a
\]
preserves the code space and acts on the logical basis state $\ket{[a]}$ by applying the corresponding phase. 
The unitary $U_\Phi$ is realized by a circuit that is the product of (multi-controlled) phase rotations. 
The cohomology operations $\Phi$ that we describe here are all \emph{local} themselves in the sense that they are a sum of terms each of which depends on a bounded number of bits, with every bit occurring in a bounded number of terms. 
This in turn implies that they yield a constant-depth physical circuit. 
Cohomology thus connects a global feature, the logical action, to a local one, the circuit, and all computations of logical actions can be carried out at the level of cohomology. 
The resulting circuits are often quite unintuitive, and it is unlikely that many of them could have been found by hand or by brute-force search.

Common sources of cohomology invariants are \emph{cohomology operations} such as the \emph{cup product}, together with an integral. For degree-one cocycles $x_1,\ldots,x_r$ with $\F_2$ coefficients on a closed triangulated $r$-manifold $X$, the integral
\[
 \int_X x_1\cup x_2\cup\cdots\cup x_r
\]
is a cohomology invariant that counts the parity of the multilinear (degree-$r$) intersection of the dual cycles $\{x_i^*\}$.
In topological codes, this immediately gives finite-depth logical gates, such as the $CZ$ gate between two copies of the two-dimensional toric code and the $CCZ$ gate between three copies of the three-dimensional toric code~\cite{Yoshida15,Yoshida16,Yoshida17,KubicaYoshidaPastawski15,Barkeshli23,Barkeshli24_2,Chen2023higher,Wang2023,Zhu2023}.

In Part I~\cite{CupsAndGates} of this work, and related concurrent papers \cite{lin2024transversal,golowich2024quantum}, the notion of cup product was generalized to qLDPC codes. To achieve this, we first defined cup products for classical codes which we then showed to induce cup products on their tensor and balanced products. We identified a structure weaker than a differential graded algebra that is sufficient for the new operation to be a cohomology invariant.\footnote{The condition for invariance that we found was that the product need not be associative, and the Leibniz rule only needs to hold after integration. This condition can be checked locally on the parity-check matrices, in a way reminiscent of triorthogonality~\cite{BravyiHaah2012}, and does not require an underlying manifold.}
As a result, we constructed the \emph{copy-cup} gates: 
given $r$ copies of a qLDPC quantum code equipped with an $r$-fold cup product, the phase $(-1)^{\int x_1\cup\cdots\cup x_r}$ is implemented by a constant-depth circuit of physical $C^{r-1}Z$ gates between the blocks. When the invariant is nontrivial, it implements a logical gate consisting of the $C^{r-1}Z$ acting on various combinations of logical qubits, which is at the $r$-th level of Clifford hierarchy, up to lower-level corrections. 

On the constructive side, cup products beyond manifolds were used to find qLDPC codes with constant-depth non-Clifford gates in finite-length tricycle codes and balanced product codes~\cite{MenonMagicTricycles2026,JacobMcLauchlanBrowne2025,LiPreskillXu25,TiewBreuckmann2026}, as well as in families of almost-good and good qLDPC codes which are also locally testable~\cite{LiLiLiu26almostgood,LiLiLiu26good}, and codes with constant-weight stabilizers whose transversal $C^{r-1}Z$ gates surpass the previously known trade-off between length, distance and number of logical gates~\cite{GolowichTamoZhu2026}.
On the conceptual side, Poincar\'e duality and cup and cap products have been established for sheaf codes~\cite{LiLiLiuNguyen25}, together with a general theory of cohomological invariants of hypergraph-product and sheaf codes~\cite{LiShaoWeiLiLiu26}, and no-go theorems constrain which gates product codes can support~\cite{FuZhengLiLiu25}. 

\paragraph{What was missing, and the new ideas.}
Despite this progress, some of the most basic desiderata remained out of reach. 
We want non-Clifford gates \emph{within a single code block}, and ideally, we also want to \emph{address} individual logical qubits. 
The desired framework needs to capture rotation gates, such as $T=\op{diag}(1,e^{2\pi i/8})$. 
Since the phase associated to an integral of products of $\Z/2$-cocycles is a sign, ordinary cup products are not sufficient for producing the eighth roots of unity that $T$ requires. 
Finally, we want constructions which apply beyond topological codes. 
This paper addresses all of these points by introducing novel cohomology invariants, which we describe now.

One of the ways to obtain new operations is to incorporate symmetries of the cochain complex\footnote{For example, in the case of topological codes, the symmetries can be simply spatial symmetries. More generally, these are automorphisms of the cochain complex of the code.} into cohomology operations.
Additionally, we can \emph{lift coefficients}, say from $\Z/2$ to~$\Z/8$, which forces correction terms that use the higher homotopies expressing the fact that the cup product is only commutative up to homotopy.
These homotopies can be organized in an $E_\infty$-structure and the language of operads which allows us to systematically construct \emph{higher twisted cohomology operations} generalizing (higher) Pontryagin squares.

We introduce further ideas in order to make the resulting gates addressable. 
One of them is that we select the logical inputs by composing the cohomology operation with cochain maps that act non-trivially on different cohomology classes including removing some of them from the expression while retaining cohomology invariance.

Lastly, another idea is to use symmetries of the cochain complex to ``fold'' it in a nontrivial way and consider relative cohomology operations on the ``folded'' complex. 

Together, these ideas produce many addressable classes of new operations. In particular, we obtain individually addressable $T$ gates on toric codes in three and four dimensions, as well as addressable non-Clifford gates on further manifolds, and on qLDPC codes which are not associated with manifolds.
To our knowledge, many of the operations here have not been studied in the mathematics literature before, and we expect them to be of independent interest in homotopy theory. 

\subsection{Main contributions}
\label{sec:main-contributions}

We now explain the ideas of the paper in more detail, in the order in which they build on each other.

\paragraph{Twisting by a symmetry.}
Assume that the cochain complex has automorphism group~$G$, so the action of elements $g\in G$ descends to cohomology. Then, the \emph{twisted cup product}
\begin{equation}\label{eq:intro-twisted-power}
 \int_X x\cup g_2x\cup\cdots\cup g_rx,
 \qquad g_i\in G,
\end{equation}
is a cohomology invariant of the single class $x$. Intuitively, this yields non-trivial outcomes for classes $x$ that intersect with elements in their $G$-orbit.
Nothing in this argument requires $g$ to be invertible. In fact, any endomorphism $F$ of the code complex can be substituted into a multilinear invariant, $\int F_1(x)\cup\cdots\cup F_r(x)$, and the result is again a cohomology invariant. 
This gives single-block gates of $CCZ$ type on the three-torus, on products of a surface with a circle, and on the tricycle and Golowich--Tamo--Zhu codes discussed in \Cref{sec:non-clifford-fold-transversal-gates}.

\paragraph{Lifting coefficients and $E_\infty$-structure.}
Twisted cup products over $\Z/2$ can only produce gates of $C^{r-1}Z$ type. For other phase gates, such as the $T$ gate, we need invariants valued in $\Z/2^k$. The naive attempt is to read the bits of a cocycle $x\in Z^q(X;\Z/2)$ as the integers $0$ and $1$, which defines a non-linear inclusion map $\iota$ with values in $C^q(X;\Z)$. Denoting $\widetilde x= \iota x$, we can form $\int\widetilde x\cup g_2\widetilde x\cup\cdots$, say, modulo $8$. The problem is that $\widetilde x$ need not be a cocycle. Its coboundary vanishes modulo $2$, but need not vanish over $\Z$. As a consequence, such a naive expression is not a cohomology operation.

The \emph{Pontryagin square} is the simplest example of repairing this in the untwisted case with only two inputs. The expression
\[
 \Po^{(2)}(x)=\widetilde x\cup\widetilde x+\widetilde x\cup_1 d\widetilde x
 \pmod 4
\]
is a well-defined operation $H^q(X;\Z/2)\to H^{2q}(X;\Z/4)$ for even $q$. 
The \emph{correction term} $\widetilde x\cup_1d\widetilde x$ involves Steenrod's cup-$1$ product $\cup_1 $ and vanishes when $\widetilde x$ is already closed modulo four. 
It makes the expression a cocycle whose cohomology class is independent of the chosen lift and cocycle representative (\Cref{sec:integral-lifts}).
The Pontryagin square is the prototype of the general form of cohomology operations in this paper, whose leading term is typically an $r$-fold cup power, plus correction terms built from coboundaries and higher cup products, such that the sum is a cohomology invariant with finer coefficients.

The cup-$1$ product supplies the homotopy between the cup product and its graded opposite. The cancellation modulo four uses both its cochain identity and the divisibility of $d\widetilde x$ by two. The same mechanism in higher arity requires homotopies between the $r$-fold cup product and its permutations, homotopies between these homotopies, and so on. Such coherent systems of higher homotopies form an \emph{$E_\infty$-structure} on cochains~\cite[Section~3]{May1972}. To make these operations explicit, we use the surjection operad $\cX$, constructed by McClure--Smith and Berger--Fresse~\cite{McClureSmith2003,BergerFresse2004}. Its elements are words whose action is defined by cutting simplices into overlapping faces. The word $(1,2)$ gives the cup product, $(1,2,1)$ gives its commutativity homotopy, and $(1,2,\ldots,r)$ gives the $r$-fold cup product. The word $(1,2,3,4,1)$ gives a homotopy between the fourfold cup product and its cyclic permutation. The differential of the operad records the identities relating these operations.

Since we want to twist by symmetries, say by a subgroup $G$ of the symmetric group $\Sigma_r$, it is particularly useful to organize the corresponding homotopies using a free resolution~$W$ of~$\Z$ over~$\Z[G]$. 
A $G$-equivariant chain map $W\to\cX(r)$~\cite{KMM2021} turns each generator~$e_{i,j}$ into an $r$-input cochain operation~$\mu^G_{i,j}$ lowering the total degree by~$i$. 
Here $\mu^G_{0,0}$ is the $r$-fold cup product, the $\mu^G_{1,j}$ give homotopies between its permutations, and the $\mu^G_{2,j}$ give homotopies between these homotopies. Their identities are encoded by the differential of~$W$.
For the simple case $G=C_2$ this recovers Steenrod's higher cup products. We explain the construction in \Cref{sec:cohomology-operations}.

\paragraph{Twisting and lifting at the same time.}
These operations can be combined to obtain \emph{twisted higher cohomology operations} which are the mathematical heart of this paper. 

Let $G$ be a finite abelian group acting on the cochain complex $C^*(X;\Z/\theta)$ and permuting the $r=|G|$ inputs of a cochain operation. Let $\omega\colon G\to\{\pm1\}$ record whether $g$ preserves or reverses the orientation, and let $M=M(\theta)[-q]$ be Steenrod's universal complex, which is freely generated by $u$ and $v$ with $du=\theta v$. A cocycle $x$ with $\Z/\theta$ coefficients together with its lift to a ring $R$, say  $\widetilde x$, can be packaged as a chain map $f_x\colon M\to C^*(X;R)$ with $f_x(u)=\widetilde x$ and $f_x(v)=\beta_\theta(x)=(d\widetilde x)/\theta$. The whole construction is contained in the following diagram of chain maps:
\begin{equation*}\label{eq:intro-evaluation-diagram}
\begin{tikzcd}[column sep=3.4em]
 \operatorname{ev}_x: W\otimes_{\Z[G]}\bigl(R_\omega\otimes M^{\otimes r}\bigr)
   \arrow[r,"\mathrm{id}\otimes f_x^{\otimes r}"]
 & W\otimes_{\Z[G]}\bigl(R_\omega\otimes C^*(X;R)^{\otimes r}\bigr)
   \arrow[r,"E"]
 & R.
\end{tikzcd}
\end{equation*}
The object on the left is \emph{universal}: it depends only on the group $G$, the resolution $W$, the coefficients $\theta$ and $R$, the degree $q$ and the character $\omega$, but \emph{not} on the specific cochain complex. 
Only finitely many terms contribute in each total degree, so its homology, the universal twisted homology $H_n(G;R_\omega\otimes M^{\otimes r})$, can be computed by finite linear algebra.
A cycle representative $z=\sum e_{i,j}\otimes m_{i,j}$ of a homology class can be thought of as a recipe for constructing a cohomology invariant: 
$m_{i,j}$ is a linear combination of tensor words in $u$ and $v$, and $e_{i,j}$ says which generalized cup-$i$ product to apply to it.

The first map in the diagram, $\mathrm{id}\otimes f_x^{\otimes r}$, substitutes $\widetilde x$ and $\beta_\theta(x)$ for $u$ and $v$, and the second map $E$ applies the operation $\mu^G_{i,j}$ to the inputs \emph{translated along the $G$-orbit}, $(a_1,\dots,a_r)\mapsto(g_1a_1,\dots,g_ra_r)$, and integrates over $X$.  
Since all maps are chain maps, a cycle $z$ on the left defines a cohomology invariant
\[
 \Phi_{[z]}\colon H^q(X;\Z/\theta)\longrightarrow R,
 \qquad [x]\longmapsto\operatorname{ev}_x(z).
\]
To obtain a physical circuit, we choose a representative satisfying a divisibility condition that makes its formula defined on arbitrary cochains, as explained in \Cref{sec:lifts-to-cochain-operations}.

For $G=C_2$ acting trivially on $X$, $\theta=2$ and $R=\Z/4$, the cycle $z=e_0\otimes u\otimes u+2e_1\otimes u\otimes v$ recovers the (integrated) Pontryagin square.

For $G=C_4$, we construct an operation lifting coefficients from $\Z/2$ to $\Z/8$, and for $G=C_3$ one lifting coefficients from $\Z/3$ to $\Z/9$. Combined with the cochain maps described below, these give individually addressable qubit $T$ gates on the four-dimensional toric code and qutrit $T$ gates on the three-dimensional toric code, respectively (\Cref{sec:the-four-torus,sec:qutrits-and-the-three-torus}).

\paragraph{Addressability by cochain maps.}
The logical gates arising from these twisted operations typically act on all logical qubits at once. To isolate an operation on a single cohomology class, we use the following simple observation. 
If $\Phi$ is a cohomology invariant on a complex $D$ and $F\colon C\to D$ is a cochain map, then $\Phi\circ F$ is a cohomology invariant on $C$, with logical action $[a]\mapsto\Phi([Fa])$. 
The map $F$ is substituted into the phase, not implemented as a unitary, so it need neither be invertible nor commute with the symmetry defining $\Phi$. 
Whenever $F$ removes support on a cohomology class, the gate acts as the identity on the corresponding logical qubits, which gives addressability.

Concrete bounded-weight maps come from \emph{circle folds}. Folding a circle onto a semicircle, $s\colon S^1\to S^1$, satisfies $s^2=s$, so $s^*$ and $1-s^*$ are complementary cochain projections. On cohomology, $s^*$ retains $H^0$ and kills $H^1$, while $1-s^*$ does the reverse. On products of circles, these maps can be combined with coordinate permutations and sums to select individual logical directions or parities (\Cref{sec:circle-fold-projections}).

To summarize, the toolbox for addressing gates includes:
\begin{itemize}
 \item \emph{Pullbacks} along cellular maps, and cochain maps between code complexes in general, transport a cohomology invariant from one complex to another: $\Phi\mapsto\Phi\circ F$.
 \item \emph{Folds and cochain projections}, such as the circle fold $s^*$ and its complement $1-s^*$, select part of the logical information, making a gate act on chosen logical qubits or parities.
 \item \emph{Relative cochains} $C^*(X,Y)$, vanishing on a subcomplex $Y$, incorporate boundary conditions, generalizing the smooth and rough boundaries of surface codes. Folds naturally produce relative cochains, for example, $1-s^*$ vanishes on the fixed semicircle.
\end{itemize}

\paragraph{Relative cohomology and folds.}
Symmetries enter in another way through folding and relative cohomology. The starting point is a logical $T$ gate on three copies of a surface code on a tetrahedron with coupled boundary conditions~\cite{HsinKobayashiZhu2024,Bauer_2025,Davydova_2026}. Its $\Z/8$-valued phase combines a triple cup product in the bulk with correction terms on a face and an edge.

Our observation is that all three inputs can be supplied by a single surface code on a cube with relative boundary conditions on two opposite faces. Coordinate permutations fold three tetrahedral regions onto one tetrahedron, and suitable combinations of the folded cochains satisfy the required boundary conditions. This gives a constant-depth $T$ gate on the single logical qubit of the three-dimensional surface code (\Cref{sec:relative-t-cube}). Circle folds and pullbacks then transport this operation to individually addressable $T$ gates on the three-dimensional toric code (\Cref{sec:relative-t-three-torus}).

The construction extends to suitable manifold codes and, through bounded-weight cochain maps, to our tricycle and three-dimensional tile-code examples (\Cref{sec:relative-t-further-examples,sec:relative-t-beyond-manifolds,sec:addressable-tt-example}). Higher-dimensional versions give higher rotation gates whenever the corresponding simplex operation is available.

\paragraph{Examples.}
\Cref{tab:intro-examples} collects the resulting gates on topological and non-topological codes.

\begin{table}[p]
\centering
\small
\setlength{\tabcolsep}{5pt}
\renewcommand{\arraystretch}{1.25}
\makeatletter
\let\intro@parboxrestore\@arrayparboxrestore
\def\@arrayparboxrestore{\intro@parboxrestore\raggedright\let\\\tabularnewline}
\makeatother
\resizebox{\textwidth}{!}{%
\begin{tabular}{p{0.25\textwidth}p{0.28\textwidth}p{0.38\textwidth}p{0.25\textwidth}}
\hline
\textbf{Complex / code} & \textbf{Logical gate (up to Clifford frame)} &
\textbf{Cohomology operation and its twist} & \textbf{Comments}\\
\hline
\multicolumn{4}{l}{\emph{Ordinary $\mathrm{mod} \ 2$ cup products twisted by automorphisms}}\\
\hline
$3$D toric code, $T^3$ 
 & $\overline{CCZ}_{123}$ 
 &$C_3$ (coordinate permutations)-twisted operation $\int x\cup\sigma x\cup\sigma^2x$
 & \Cref{sec:folded-three-torus}; known~\cite{SahayLongKhemani26}\\
$\Sigma\times S^1$, where $\Sigma$ is a 2D surface, e.g. the Bolza surface
 & $\overline{CCZ}_{u,\gamma_1+\gamma_3,\gamma_2+\gamma_4}$
 & $\int(1+s^*)x\cup s^*x\cup g\,s^*x$; automorphism $g$ of $\Sigma$; circle fold $s$
 & \Cref{sec:bolza-surface}; new\\
Magic-tricycle codes, e.g. $[[108,15,(12,6)]]$ of~\cite{MenonMagicTricycles2026}
 & a product of three $\overline{CCZ}$ gates in one block
 & $\int x\cup F(x)\cup F'(x)$; weight-six cochain endomorphisms $F,F'$ 
 & \Cref{sec:folded-nontopological-codes}; new\\
Golowich--Tamo--Zhu,  $[[N, N^{1-\varepsilon}, N^{(1-\varepsilon)/3}]]$
 & $\Omega(N^{1-\varepsilon})$ addressable $\overline{CCZ}$ gates in one block
 & triple cup product $\int x\cup\sigma x\cup\sigma^2x$; $C_3$ cyclically permuting the three tensor factors
 & \Cref{sec:folded-nontopological-codes}; new \\
\hline
\multicolumn{4}{l}{\emph{Higher twisted Pontryagin operations}}\\
\hline
4D toric code with qubits on edges, $T^4$
 & individually addressable $\overline T_i$, $\overline{CS}_{ij}$, $\overline{CCZ}_{ijk}$
 & $C_4$ (coordinate permutations with sign reversal)-twisted Pontryagin power $\Po^{(4)}_\rho$, with circle folds $F_i$
 & \Cref{sec:the-four-torus}; new\\
Qutrit 3D toric code, $T^3$
 & addressable qutrit $\overline T=\op{diag}(1,\zeta_9,\zeta_9^{-1})$
 & $C_3$ (coord. permutation)-twisted Pontryagin cube $\Po^{(3)}_\rho$; with circle folds $F_i$
 & \Cref{sec:qutrits-and-the-three-torus}; new\\
$\Sigma\times\Sigma$, e.g.\ the Bolza square
 & products of $\overline T^{\pm1}$, $\overline{CS}$, $\overline{CZ}$, $\overline{CCZ}$, $\overline{CCCZ}$ 
 & $\Po^{(4)}_\rho$ with $\rho(p,q)=(\tau q,p)$, and $\Po^{(4)}_{\sigma,\tau}$ with the Klein four-group generated by $\tau\times\tau$ and the swap
 & \Cref{sec:squares-of-surfaces}; new\\
\hline
\multicolumn{4}{l}{\emph{Relative cohomology and folds onto boundaries}}\\
\hline
Three surface codes on a 3D tetrahedron
 & $\overline T$ 
 & relative cohomology version of the triple cup product, $\Phi_{\mathcal D}$ 
 & \Cref{sec:relative-t-tetrahedron}; known~\cite{HsinKobayashiZhu2024,Bauer_2025,Davydova_2026}\\
Surface code on a 3D cube $Q$
 & $\overline T$
 & $\Phi_{\mathcal D}$ folded by the coordinate permutations of the cube
 & \Cref{sec:relative-t-cube}; new\\
3D toric code, $T^3$
 & addressable $\overline T_i$; $\overline{CS}_{ij}$, $\overline{CCZ}_{ijk}$ 
 & $\Phi_{\mathcal D}$ with circle fold, quotient map, fold onto the tetrahedron
 & \Cref{sec:relative-t-three-torus}; new\\
$d$-torus $T^d$, qubits on edges
 &addressable $\overline{\op{R}}_i(2\pi/2^d)$
 & $\Phi_{\mathcal D_d}$ folded by coordinate permutations and a circle fold
 & \Cref{sec:relative-t-further-examples}; new\\
 any closed oriented $M^n$, $n\geq3$, with $H_1(M;\Z/2)\neq0$
 & $\overline T$ on the circle qubit
 & generalization of $\Phi_{\mathcal D}$, then circle fold
 & \Cref{sec:relative-t-further-examples}; new\\
3D tile codes of~\cite{BreuckmannChoeEberhardtPereiraSteffan2026}, e.g. $[[545,16]]$
 & addressable $\overline T_i$ on all $16$ logical qubits
 & $\Phi_Q$ composed with specific cochain maps
 & \Cref{sec:relative-t-beyond-manifolds}; new\\
Magic-tricycle codes of \cite{MenonMagicTricycles2026}, e.g. $[[108,15,(12,6)]]$
 & addressable $\overline T_i$ on all $15$ logical qubits
 & $\Phi_{T^3}$ composed with specific cochain maps
 & \Cref{sec:addressable-tt-example}; new\\
\hline
\end{tabular}}
\caption{The examples of this paper. 
All gates act within a single code block and are implemented by
constant-depth circuits of controlled phase rotations.}
\label{tab:intro-examples}
\end{table}

In our examples we did nothing to optimize the depth of the circuits, and there is a lot of room to do so. First, every operation lives in an equivalence class: cohomologous universal cocycles $z$, different resolutions $W$, and different comparison maps $W\to\cX(r)$ all give the same logical gate but have different cochain-level expressions. Searching within an equivalence class for a representative of small depth is a concrete optimization problem. Second, coefficients may be lifted before or after applying the cochain maps which make a gate addressable, and these maps are themselves far from unique. Third, the symmetry~$G$, the automorphisms, and the folds can all be chosen differently, and for a given target gate set some choices will lead to much shallower circuits than others. We leave such optimizations to future work.

\subsection{Comparison with the literature}
\label{sec:comparison-with-literature}

\paragraph{Topological codes and the symmetries picture.}
For topological codes, constant-depth logical gates have long been understood through symmetry-protected topological phases and the cohomology operations describing their responses, in particular cup products and their higher versions~\cite{Yoshida15,Yoshida16,Yoshida17,Barkeshli23,Barkeshli24_1,Barkeshli24_2,Chen2023higher,Chen23,Wang2023,Zhu2023,HsinKobayashiZhu2024}. The twisted operations of this paper differ in that a symmetry of the space enters the operation itself; a few examples of this have been known as gates on folded triangles and on tetrahedra~\cite{HsinKobayashiZhu2024}, but this paper goes far beyond that.

\paragraph{Fold-transversal gates.}
Fold-transversal Clifford gates~\cite{Moussa2016,foldtransversalgates,Quintavalle_2023,EberhardtSteffan2024,SayginelEtAl2024} use an automorphism of the code together with the intersection form, and more generally, code automorphisms give Clifford gates~\cite{GrasslRoetteler2013}. Our results show the non-Clifford extensions that are somewhat similar in spirit: automorphisms are substituted into a multilinear cohomology invariant; they also enter inside the homotopies of the $E_\infty$-structure, which is what allows the coefficients to be lifted.

\paragraph{Constructions of qLDPC codes with non-Clifford gates.}
Cup products on sheaf codes and products of algebraic codes~\cite{lin2024transversal,golowich2024quantum}, on homological product codes~\cite{Zhu25topological,Zhu25systolic}, on tricycle and related product codes~\cite{MenonMagicTricycles2026,JacobMcLauchlanBrowne2025,TiewBreuckmann2026,LiPreskillXu25}, on almost-good and good locally testable codes~\cite{LiLiLiu26almostgood,LiLiLiu26good} and on the codes of Golowich, Tamo and Zhu~\cite{GolowichTamoZhu2026} all produce multi-block $C^{r-1}Z$ gates of copy-cup type. 

Our new contribution shows how to fold such gates into a single block and also how to make them addressable. 

\paragraph{Homological frameworks for logical gates.}
Several recent works develop general homological descriptions of logical gates on qLDPC codes. Li, Li, Liu and Nguyen~\cite{LiLiLiuNguyen25} establish Poincar\'e duality and cup and cap products for sheaf codes. Li, Shao, Wei, Li and Liu~\cite{LiShaoWeiLiLiu26} develop a theory of (co)homological invariants of hypergraph-product and sheaf codes. These multiplicative structures provide starting points for our universal construction, provided they extend to generalized cup-$i$ products with a compatible symmetry and integral.

Sahay, Long and Khemani~\cite{SahayLongKhemani26} develop a concurrent framework for logical gates and addressability based on general chain maps between code complexes, without restricting to maps arising from cup products and their higher homotopies. It organizes diagonal gates of $C^{\ell-1}Z$ type between $\ell$ code blocks by the homology of the tensor product of the $\ell$ code complexes; in the language of \Cref{sec:logical-gates-from-invariants} these are exactly the $\Z/2$-multilinear cohomology invariants, and the classification is an instance of the K\"unneth theorem. 
Every such logical action is realized by some cycle, in general a dense one.
Gates are found by sparsifying it with boundaries, when possible, and reading off the circuit. 
A multi-block pattern applied within one block yields the symmetrized action. Their examples are the full Clifford group on the two-dimensional toric code, addressable $CCZ$ gates between and within blocks of the three-dimensional toric code, and addressable Clifford and $CCZ$ gates on fracton codes. They obtain diagonal Clifford gates such as $S$ from symmetric chain maps on a suitable lift of the whole code complex to $\Z/4$ coefficients, giving explicit quadratic, non-multilinear phase functions, and point to quadratic homological maps. A full treatment of intra-block non-Clifford gates is left to future work.

Our results and discussion are complementary. Their multilinear diagonal construction produces products of $C^kZ$ gates (including intra-block ones), and they also construct $S$ gates using $\Z/4$-valued quadratic phase functions. 
Our construction here focuses on intra-code gates and supplies explicit coefficient-lifting operations for single-block $T$, $CS$ and higher phase gates.
The universal lifts use the higher homotopies of the $E_\infty$-structure, organized by a resolution of the symmetry group, to construct the correction terms needed for these lifted expressions to descend to cohomology, while the relative construction also uses boundary terms. Sparsity follows automatically from bounded incidence and bounded-weight cochain maps.
Their framework also treats non-diagonal gates and measurement gadgets. We expect the two approaches to combine well to yield much stronger results in the future.

\paragraph{Mathematics.}
Steenrod constructed the cup-$i$ products~\cite{Steenrod1947} and the cyclic reduced powers~\cite{Steenrod1953} and gave the universal description of the Pontryagin square~\cite{Steenrod1957}, Thomas the higher Pontryagin powers~\cite{Thomas1957}, and May gave a general algebraic framework for Steenrod operations~\cite{May1970}. The surjection operad is due to McClure and Smith~\cite{McClureSmith2003} and Berger and Fresse~\cite{BergerFresse2004}, and Kaufmann and Medina-Mardones~\cite{KMM2021} gave the explicit cochain-level formulas for Steenrod operations which we build on. The twisted operations constructed here, which evaluate the higher homotopies along the orbit of a group acting on the space, appear to be novel.

\subsection{Outlook}
\label{sec:outlook}

The universal construction of \Cref{sec:constructing-twisted-cohomology-operations} uses only three things about the space $X$: an action of $G$ on its cochains, an $E_\infty$-type family of generalized cup-$i$ products, and an integral compatible with $G$. Nothing requires $X$ to be a manifold, and the most important open direction is to systematically study these operations on qLDPC codes beyond manifolds. Tensor and balanced products of classical codes with cup products~\cite{CupsAndGates} and the multiplicative structures of~\cite{LiLiLiuNguyen25,LiShaoWeiLiLiu26} are natural candidates for constructing concrete operations.

Another question on the coding theory side is how far addressability reaches. Which gates at a given level of the Clifford hierarchy can be individually addressed, and which logical bases can be reached by these methods? For codes associated with $d$-dimensional spaces, what constraints does the geometry impose on constant-depth circuits with nonlocal connectivity? Such constraints would provide a nonlocal analogue of the Bravyi--K\"onig bound. Finally, we have only begun to exploit $E_\infty$-structures and operads for fault-tolerant quantum computation; we expect this operadic perspective to uncover further connections between homotopy theory and logical gates.

The cup product is a relatively common tool in condensed-matter and high-energy physics. It is used to write down topological actions and exactly solvable lattice models, such as Dijkgraaf--Witten gauge theories~\cite{DijkgraafWitten}, and to describe symmetry-protected topological (SPT) phases~\cite{Chen13} and anomalies of ordinary and higher-form symmetries~\cite{KapustinThorngren14anomalies,kapustin2015higher,Gaiottoetal2015,Barkeshli24_1}. Steenrod's higher cup products~\cite{Steenrod1947} similarly enter lattice descriptions of fermionic phases and bosonization~\cite{GaiottoKapustin16,ChenKapustinRadicevic18,Chen2023higher}.
Cup and higher-cup products have also recently been used to construct
finite-depth non-Clifford logical gates in topological and qLDPC
codes~\cite{HsinKobayashiZhu25,KobayashiZhuHsin26,SahayLongKhemani26},
while closely related group-cohomological constructions and protocols
using non-Abelian topological codes yield phase gates throughout the
Clifford hierarchy and beyond~\cite{HuangWarmanSchaferNamekiChen25,
WarmanSchaferNameki26}; the new invariants developed here may provide
further logical gates in these settings.
More generally, it would be interesting to understand the role of newly found invariants in physics.

On the mathematical side, the universal twisted homology groups $H_n(G;R_\omega\otimes M^{\otimes r})$ and their liftable parts are computable by finite linear algebra, and we have only evaluated them in the cases needed for our examples. Computing them for all finite abelian groups, degrees and coefficients would give a complete classification of twisted cohomology operations of this type, and we expect that each class can be found constructively. The operations themselves seem interesting in homotopy theory in their own right, as cohomology operations for spaces with a group action, and we expect them to admit an interpretation as topological (albeit geometrically non-local) responses of systems with a spatial symmetry.

\subsection{Structure of the paper}
\label{sec:structure}

The remainder of the paper is structured as follows. \Cref{sec:cochain-complexes} contains preliminary material such as cochain complexes, their relation to CSS codes, gates from cohomology operations, as well as some mathematical preliminaries needed later. \Cref{sec:cohomology-operations} introduces the cup-$i$ products and Steenrod squares, the surjection operad and the $E_\infty$-structure on cochains, and derives the generalized cup-$i$ products and the lifting of coefficients.
\Cref{sec:constructing-twisted-cohomology-operations} contains the main mathematical construction: universal twisted elements and their evaluation on a space with symmetry.
\Cref{sec:selecting-cohomology-operations} explains addressability by composing with cochain maps.

The remaining sections are devoted to examples: \Cref{sec:non-clifford-fold-transversal-gates} twists ordinary cup products into single-block gates and shows some addressable examples of them; \Cref{sec:non-clifford-gates-via-twisted-cohomology-operations} shows several applications of lifted-coefficient higher twisted Pontryagin powers; and \Cref{sec:addressable-relative-cup-gates} discusses the relative cohomology operations building on the example of a phase gate on a 3-dimensional tetrahedron that gives a $T$ gate on a cube and a torus; this is further generalized to higher dimensions and to some non-manifold examples.

\paragraph{Acknowledgements.}
We gratefully acknowledge discussions with Andreas Bauer, Shankar Balasubramanian, Louis Golowich, Vedika Khemani, David Lin, David Long, Rahul Sahay.
LLM assistance was used for extensive reference searches and working out some specific examples. 
M.D. was supported by the Walter Burke Institute for Theoretical Physics at Caltech and by the Institute for Quantum Information and Matter, an NSF Physics Frontiers Center (NSF Grant PHY-2317110). 
N.P.B.~acknowledges support through the Quantum Research Pod, CIQC and the Simons Institute for the Theory of Computing.
This work was done in part while N.P.B.~was visiting the Simons Institute for the Theory of Computing, supported by NSF QLCI Award OMA-2016245.
N.P.B., M.D. and N.T.\ acknowledge the support of the Microsoft Quantum Pioneers Program.
When working on the first part \emph{Cups and Gates I}, we spent many sleepless nights chasing signs in correction terms for higher twisted Pontryagin squares and ultimately gave up. All coauthors are incredibly grateful to J.N.E.'s master's student and one of the paper's authors, L.F.H., who tamed the beast called surjection operad to fix all these problems so elegantly.

\section{Cochain complexes, CSS codes, and logical gates}
\label{sec:cochain-complexes}
We explain how based cochain complexes define CSS codes and how
cohomology invariants give logical gates and circuits. We then recall
cochains of spaces, tensor products, hom complexes, and resolutions.

\subsection{Cochain complexes}
\label{sec:based-cochain-complexes}

\begin{definition}
A \emph{cochain complex} over a commutative ring $R$ is a sequence of
$R$-modules
\[
 \cdots\longrightarrow A^{q-1}\xrightarrow{d}A^q
 \xrightarrow{d}A^{q+1}\longrightarrow\cdots
\]
with $R$-linear differentials satisfying $d^2=0$.
\end{definition}
For quantum codes we work with \emph{based} complexes: each $A^q$ is a
finite-rank free module with a chosen basis.
The cocycles, coboundaries, and cohomology are
\[
 Z^q(A)=\ker d,\qquad
 B^q(A)=\operatorname{im}d,\qquad
 H^q(A)=Z^q(A)/B^q(A).
\]
A \emph{cochain map} $F\colon A\to A'$ is a degree-preserving linear map
satisfying $dF=Fd$.  It induces a map on cohomology.  Two cochain maps
$F,G$ are \emph{homotopic} if $$F-G=dh+hd$$ for a map $h$ of degree $-1$,
in which case their induced maps on cohomology agree.

It is sometimes useful to work instead with chain complexes, indexed by
$A_n$, whose differential lowers degree.
This is only a change of convention: setting $A^n=A_{-n}$ translates
between chain and cochain notation.

Our formulas will often use the coefficientwise lift
\[
 \iota\colon\F_2^n\longrightarrow(\Z/2^s)^n
\]
which reads each bit as $0$ or $1$.  When a
construction also requires a differential on the lifted coefficients, we
will specify a compatible complex over $\Z$ or $\Z/2^s$.

\subsection{CSS codes}
\label{sec:css-codes}

Let $A$ be a based cochain complex over $\F_2$.  Three consecutive terms
\[
 A^{q-1}\xrightarrow{d}A^q\xrightarrow{d}A^{q+1}
\]
define a CSS code. There is one physical qubit for each basis element $e$ of $A^q$.  Every
$a\in A^q$ has a unique expansion $a=\sum_e a(e)\,e$ with coefficients
$a(e)\in\F_2$, and we write $\ket a=\bigotimes_e\ket{a(e)}$ for the
corresponding computational basis state, so that $a(e)$ is the value of
qubit $e$ in the state $\ket a$.

The Pauli operators $X_e$ and $Z_e$ on qubit $e$ act on these states by
\begin{equation}
    X_e\ket a=\ket{a+e},\qquad Z_e\ket a=(-1)^{a(e)}\ket a.
\end{equation}
We use the same symbols for basis elements and their duals.
More generally, for $c\in A^q$ and $c^*\in(A^q)^*$, we write the support of the operator as its superscript:
\begin{equation}
X^c=\prod_e X_e^{c(e)},\qquad Z^{c^*}=\prod_e Z_e^{c^*(e)},
\end{equation}
so that $X_e=X^e$, $Z_e=Z^e$, and
\begin{equation}
 X^c\ket a=\ket{a+c},\qquad Z^{c^*}\ket a=(-1)^{\langle c^*,a\rangle}\ket a,
\end{equation}
where $\langle c^*,a\rangle\in\F_2$ is the standard
chain-cochain pairing.  In particular, $X^c$ and $Z^{c^*}$ commute if
$\langle c^*,c\rangle=0$ and anticommute otherwise.

For basis elements $v\in A^{q-1}$ and $p\in(A^{q+1})^*$,
the checks are $X^{dv}$ and $Z^{d^Tp}$. They commute since $d^2=0$.

The code space is the common $+1$ eigenspace of all the checks.  The
$Z$-checks are diagonal in the computational basis:
$$Z^{d^Tp}\ket a=(-1)^{\langle p,da\rangle}\ket a,$$ and $\langle p,da\rangle$
is the coefficient of $p$ in $da$.  Therefore, a state is fixed by every
$Z$-check exactly when it is a superposition of basis states $\ket a$ with
$da=0$, that is, with $a$ a cocycle.  The $X$-checks permute these basis
states, $X^{dv}\ket a=\ket{a+dv}$, so such a superposition is fixed by every
$X$-check exactly when its amplitudes are constant on each coset
$a+B^q(A)$ of the space of coboundaries
$B^q(A)=\operatorname{im}(d\colon A^{q-1}\to A^q)$.  These cosets are
precisely the elements of the cohomology group
$H^q(A)=\ker(d\colon A^q\to A^{q+1})/B^q(A)$, so the code space has the
orthonormal basis
\[
 \ket{[a]}=\frac{1}{\sqrt{|B^q(A)|}}\sum_{b\in B^q(A)}\ket{a+b},
 \qquad [a]\in H^q(A),
\]
which does not depend on the choice of cocycle $a$ representing the class
$[a]$.

The cohomology classes in $H^q(A)$ thus label the logical basis states, and
choosing a basis of $H^q(A)$ identifies these states with the computational
basis states of $k=\dim_{\F_2}H^q(A)$ logical qubits.  The code therefore
encodes $k$ logical qubits into $n=\dim_{\F_2}A^q$ physical qubits.

\subsection{Logical gates from cohomology invariants}
\label{sec:logical-gates-from-invariants}

Since cohomology classes label the logical basis states, a phase
assigned to each class defines a diagonal logical gate. We therefore
consider a function, not necessarily linear,
\[
 \overline\Phi\colon H^q(A)\longrightarrow\Z/2^s,
\]
which we call a cohomology invariant.

A physical implementation requires a function on all cochains,
\[
 \Phi\colon A^q\longrightarrow\Z/2^s,
\]
describing a unitary diagonal in the physical $Z$-basis.
For $\Phi$ to induce $\overline\Phi$, we require
$\Phi(a)=\overline\Phi([a])$ for every cocycle $a$.  In particular, the
value must be independent of the representative:
\[
 \Phi(a+b)=\Phi(a),
 \qquad a\in Z^q(A),\quad b\in B^q(A).
\]
In this case we say that $\Phi$ \emph{descends to cohomology}.

\begin{definition}\label{def:unitary-from-operation}
The diagonal unitary associated to $\Phi$ is
\[
 U_\Phi\ket a
 =\exp\!\left(\frac{2\pi i}{2^s}\Phi(a)\right)\ket a.
\]
\end{definition}

\begin{proposition}\label{prop:diagonal-logical-gate}
If $\Phi$ descends to cohomology, $U_\Phi$ preserves the code space and
has logical action
\[
 U_\Phi\ket{[a]}
 =\exp\!\left(\frac{2\pi i}{2^s}\overline\Phi([a])\right)\ket{[a]}.
\]
\end{proposition}

\begin{proof}
Every summand $\ket{a+b}$ of the code state $\ket{[a]}$ receives the same
phase.
\end{proof}

The logical action depends only on the values of $\Phi$ on cocycles.
Its values on other cochains determine the remaining physical action and
therefore matter when constructing the circuit.

\subsection{Circuits and their logical action}
\label{sec:circuits-and-logical-action}

We now express the physical and logical diagonal unitaries of the
previous section as circuits of controlled rotation gates.
We use the following lemma.
\begin{lemma}[Phase polynomials]\label{lem:phase-polynomial}
For $x = (x_1,\ldots,x_n)\in\F_2^n$ and a subset $I$ of $[n]=\{1,\ldots,n\}$, we write
\begin{equation}
  x^I = \prod_{i \in I} \iota(x_i) 
\end{equation}
where $x^I\in\Z/2^s$ and $x^\emptyset=1$. Thus $x^I=1$ if
$x_i=1$ for all $i\in I$ and $0$ otherwise. Every function
$f\colon\F_2^n\to\Z/2^s$ has a unique expansion in these monomials
\[
 f(x_1,\ldots,x_n)
 =\sum_{I\subseteq\{1,\ldots,n\}}c_I x^I
\]
where $c_I\in \Z/2^s$.   Its coefficients are
\[
 c_I=\sum_{J\subseteq I}(-1)^{|I|-|J|}f(1_J),
\]
with $1_J$ the indicator vector of $J$.
\end{lemma}

Let $\op{R}(\varphi)=\operatorname{diag}(1,e^{i\varphi})$, with $\varphi$
measured in radians. For a set
$I$ of qubits, let $\op{R}_I(\varphi)$ be the diagonal gate
\[
 \op{R}_I(\varphi)\ket x=e^{i\varphi\,x^I}\ket x ,
\]
which multiplies $\ket x$ by $e^{i\varphi}$ if all qubits in $I$ are
in the state $\ket1$ and leaves it unchanged otherwise.  On one qubit this
is $\op{R}(\varphi)$, on two qubits it is the controlled-$\op{R}(\varphi)$
gate, and in general it is $\op{R}(\varphi)$ on one qubit of $I$ controlled
by the other $|I|-1$ qubits; it does not matter which qubit is the target.
For example,
\[
 \op{R}_i(\pi)=Z_i,\quad
 \op{R}_i(\pi/2)=S_i,\quad
 \op{R}_i(\pi/4)=T_i,\quad
 \op{R}_{ij}(\pi)=CZ_{ij},\quad
 \op{R}_{ijk}(\pi)=CCZ_{ijk},
\]
and $\op{R}_\emptyset(\varphi)=e^{i\varphi}$ is a global phase.

Applying \Cref{lem:phase-polynomial} to $\Phi$, the circuit
\[
 C_\Phi=\prod_{I\subseteq[n]}\op{R}_I\!\left(\frac{2\pi c_I}{2^s}\right)
\]
implements $U_\Phi$.

\begin{example}\label{ex:lifting-creates-controls}
Let $\iota\colon\F_2\to\Z/4$ be the $0/1$-lift.  For $x,y\in\F_2$,
the phase
\[
 \Phi(x,y)=\iota(x)+\iota(y)\pmod4
\]
gives $C_\Phi=S_1S_2$, where
$S=\op{R}(\pi/2)=\operatorname{diag}(1,i)$.
In contrast, lifting the sum $x+y$ taken in $\F_2$ gives
\[
 \iota(x+y)=\iota(x)+\iota(y)-2\iota(x)\iota(y)\pmod4,
\]
whose circuit is $S_1S_2CZ_{12}$.  The extra $CZ$ gate comes from the
nonadditivity of the lift.
\end{example}

When $\Phi$ descends to cohomology, the same construction applied to
$\overline\Phi$ gives its logical circuit $C_{\overline\Phi}$.  Choose
cocycles $\gamma_1,\ldots,\gamma_k$ representing a logical basis and
expand
\[
 \overline\Phi\!\left(\sum_{j=1}^k\epsilon_j[\gamma_j]\right)
 =\Phi\!\left(\sum_{j=1}^k\epsilon_j\gamma_j\right).
\]
The resulting circuit on the $k$ logical qubits describes the action of
$C_\Phi$ on the code space.

\begin{remark}[Locality and depth]
To obtain bounded depth, we need a cochain formula with bounded
interaction size and participation.  Suppose
\[
 \Phi(a)=\sum_\alpha\Phi_\alpha(a),
\]
where each summand depends on at most $w$ physical bits and each bit
occurs in at most $\nu$ summands.  The corresponding bounded-size diagonal
unitaries can be scheduled in at most
$
 w(\nu-1)+1$
layers since each gate conflicts with at most $w(\nu-1)$ others.  Expanding
each local unitary into controlled rotations adds only bounded overhead
for fixed $w$.  For a family of codes, uniform bounds on $w$ and $\nu$
therefore give bounded depth. No geometric locality is assumed here.
\end{remark}

\begin{remark}
The CSS construction and the criterion for a diagonal logical gate work
equally over $\F_p$, for prime $p$, with $p$-dimensional qudits.  A
function
\[
 \Phi\colon A^q\longrightarrow\Z/p^s
\]
on a complex over $\F_p$, descending to cohomology, gives a logical gate
with phases $e^{2\pi i\Phi(a)/p^s}$.  Coefficient lifts use the
representatives $0,\ldots,p-1$.
\end{remark}

\begin{example}[A qutrit phase]
For one qutrit, let $t\in\F_3$ and consider
\[
 \Phi(t)=\iota(t)^3\pmod9,
\]
where $\iota(t)\in\{0,1,2\}$.  Its values are $0,1,-1$ modulo nine,
so
\[
 U_\Phi=T=\operatorname{diag}(1,\zeta_9,\zeta_9^{-1}),
 \qquad \zeta_9=e^{2\pi i/9}.
\]
This is the qutrit $T$ gate used in
\Cref{sec:qutrits-and-the-three-torus}.
\end{example}

\subsection{Cochain complexes associated to spaces}
\label{sec:cochains-of-spaces}
Many of the CSS codes we will consider here are \emph{topological}, that is, they come from the cochain complex associated to a topological space. We use simplicial complexes as a convenient model for these spaces.
\begin{definition}
A \emph{finite locally ordered simplicial complex} $X$ is a collection of
nonempty subsets of a finite vertex set, closed under taking nonempty
subsets, together with a total order on the vertices of every simplex such
that every face carries the restricted order.  A simplex with $p+1$ vertices
is a $p$-simplex, written $[v_0,\ldots,v_p]$ with $v_0<\cdots<v_p$.  We
write $X_p$ for the set of $p$-simplices.
\end{definition}

A total order on all vertices induces a local order.  A non-trivial
automorphism never preserves a total order of all vertices, but it may
preserve a local one, as needed in
\Cref{sec:constructing-twisted-cohomology-operations}.

The integral chain groups of $X$ and their boundary are
\[
 C_p(X;\Z)=\bigoplus_{\sigma\in X_p}\Z\sigma\text{ where }
 d[v_0,\ldots,v_p]
 =\sum_{j=0}^p(-1)^j[v_0,\ldots,\widehat v_j,\ldots,v_p]
\]
with $d=0$ in degree zero.  The cochain complex with coefficients in $R$
is
\[
 C^p(X;R)=\operatorname{Hom}_{\Z}(C_p(X;\Z),R)\text{ where }
 (da)(\sigma)=a(d\sigma).
\]
By abuse of notation, we will sometimes denote both a simplex and its dual basis cochain by
$[v_0,\ldots,v_p]$, so that $X_p$ labels the bases of both $C_p(X;\Z)$ and $C^p(X;R)$.

Applying the CSS construction to $C^\bullet(X;\F_2)$ gives a code with
qubits on the $q$-simplices.  For $q=1$, edges give qubits, vertices give
$X$-checks, and triangles give $Z$-checks.  Relative cochains
$C^\bullet(X,Y;R)$, consisting of cochains vanishing on a subcomplex
$Y$, give the corresponding construction with relative boundary
conditions.

\begin{remark}
More generally, one can use cubical or cellular chain complexes, with
one basis element per cell.  The passage from a complex and a cohomology
invariant to a code and a logical gate is unchanged.
\end{remark}

\subsection{Tensor products and hom complexes}
\label{sec:hom-complex}

In this section, we will use homological grading, as this is customary for the operads that we will use below, so a cochain complex is reindexed by
$A_n=A^{-n}$.

The \emph{tensor product} of chain complexes over $R$ is
\[
 (A\otimes_R B)_n=\bigoplus_{i+j=n}A_i\otimes_R B_j,
 \qquad
 d(a\otimes b)=da\otimes b+(-1)^{|a|}a\otimes db.
\]
Here $|a|$ denotes the degree of a homogeneous element $a$.  Interchanging
the factors introduces the Koszul sign
$a\otimes b\mapsto(-1)^{|a||b|}b\otimes a$.

If $A$ has a right action of a group $G$ by chain maps and $B$ has a left
action, their \emph{balanced product} is
\[
 A\otimes_{R[G]}B
 =\frac{A\otimes_R B}
 {\langle(a\cdot g)\otimes b-a\otimes(g\cdot b)\rangle_R}.
\]
It inherits the tensor-product grading and differential.  This construction
underlies balanced product quantum codes~\cite{balancedproductcodes}.  For abelian
$G$, a left action can also be read as a right action by $a\cdot g=ga$ and the balancing relation then becomes $ga\otimes b=a\otimes gb$, as in
the universal complexes used later.

The hom complex records chain maps and their homotopies.

\begin{definition}\label{def:hom-complex}
The \emph{hom complex} of two complexes $A$ and $B$ over $R$ consists of
graded linear maps:
\[
 \operatorname{hom}(A,B)_n
 =\prod_j\operatorname{Hom}_R(A_j,B_{j+n}),
 \qquad
 d_{\mathrm{hom}}f=d_B\circ f-(-1)^n f\circ d_A.
\]
In cohomological notation we write
$\operatorname{hom}^q(A,B)=\operatorname{hom}(A,B)_{-q}$; these are the
maps which raise the cohomological degree by $q$.
\end{definition}
A cycle of degree $0$
is a map with $d\circ f=f\circ d$, that is, a chain map.  More generally the
cycles of degree $n$ are the \emph{chain maps of degree $n$}, which satisfy
$d\circ f=(-1)^nf\circ d$.  Two such maps $f$ and $g$ are \emph{homotopic} if
\[
 f-g=d_{\mathrm{hom}}h=d\circ h+(-1)^nh\circ d
 \qquad\text{for some } h\in\operatorname{hom}(A,B)_{n+1}.
\]
Homotopic chain maps induce the same
map in homology.
Hence $H_n(\operatorname{hom}(A,B))$ is the set of chain maps of degree $n$
up to homotopy, and in particular $H_0(\operatorname{hom}(A,B))$ consists of
chain maps up to homotopy.

\begin{example}[Cocycles are chain maps]\label{ex:cocycles-as-chain-maps}
Let $R$ also denote the complex with $R$ in degree $0$ and zero otherwise.
A map $f\colon R\to A$ of degree $n$ is determined by $a=f(1)\in A_n$, and
$d_{\mathrm{hom}}f=d\circ f$ corresponds to $da$, so
$\operatorname{hom}(R,A)\cong A$.  Thus the cocycles in $A^q$ are the
chain maps in $\operatorname{hom}^q(R,A)$, and two cocycles are cohomologous
if and only if the corresponding maps are homotopic.
\end{example}

One important property of the tensor product and the hom complex is
that they are \emph{adjoint} to each other.

\begin{proposition}[Hom-tensor adjunction]\label{prop:hom-tensor-adjunction}
Let $A$ be a complex of right $R[G]$-modules, $B$ a complex of left
$R[G]$-modules and $C$ a complex of $R$-modules.  Let $G$ act on
$\operatorname{hom}(B,C)$ from the right by
$(\varphi\cdot g)(b)=\varphi(g\cdot b)$, and write
$\operatorname{hom}_{R[G]}$ for the subcomplex of $G$-equivariant maps.
Then
\[
 \operatorname{hom}_{R[G]}\bigl(A,\operatorname{hom}(B,C)\bigr)
 \xrightarrow{\ \cong\ }
 \operatorname{hom}\bigl(A\otimes_{R[G]}B,\,C\bigr),
 \qquad
 f\longmapsto\bigl(a\otimes b\mapsto f(a)(b)\bigr),
\]
is an isomorphism of chain complexes.  In particular, $G$-equivariant chain
maps $A\to\operatorname{hom}(B,C)$ are the same as chain maps
$A\otimes_{R[G]}B\to C$.  For the trivial group this is the ordinary
isomorphism
$\operatorname{hom}(A,\operatorname{hom}(B,C))\cong\operatorname{hom}(A\otimes B,C)$.
\end{proposition}

For abelian $G$, with left actions read as right actions as above, the
proposition identifies the $G$-equivariant chain maps
$\mu\colon W\to\operatorname{hom}(A^{\otimes r},A)$ of
\Cref{sec:resolutions-and-generalized-cup-i} with chain maps
$W\otimes_{R[G]}A^{\otimes r}\to A$, such as Steenrod's
map in \Cref{eq:steenrods-map}.

The multilinear cochain operations of \Cref{sec:cohomology-operations} are
elements of hom complexes, namely of the endomorphism operad
$\operatorname{hom}(A^{\otimes r},A)$ of
\Cref{ex:endomorphism-operad}.  Its homotopies and higher homotopies,
such as the generalized cup-$i$ products, will play the role of correction
terms which make a coefficient lift possible in the universal twisted
construction of \Cref{sec:constructing-twisted-cohomology-operations}.

\subsection{Resolutions}
\label{sec:resolutions}

Let $G$ be a finite group.  In \Cref{sec:cohomology-operations} the
homotopies between an operation and its permutations by $G$ are organized by
a free resolution of $G$.  We write $\Z[G]$ for the group ring.  A $\Z[G]$-module is an
abelian group with a $G$-action, and $\Z$ denotes the trivial module.
Resolutions are graded homologically, with $W_i$ in degree $i$.

\begin{definition}[Free resolution]\label{def:permutation-group-resolution}
A \emph{free resolution} of $\Z$ over $\Z[G]$, or simply of $G$, is a chain
complex $W$ of free $\Z[G]$-modules with $W_i=0$ for $i<0$, together with a
$\Z[G]$-linear augmentation $\varepsilon\colon W_0\to\Z$, such that
\[
 \cdots\longrightarrow W_2\longrightarrow W_1\longrightarrow W_0
 \xrightarrow{\ \varepsilon\ }\Z\longrightarrow0
\]
is exact.
\end{definition}
The following theorem says that free resolutions are unique up to
$G$-equivariant homotopy equivalence.
\begin{theorem}\label{thm:comparison}
Let $W$ be a chain complex of free $\Z[G]$-modules with $W_i=0$ for $i<0$
and a $\Z[G]$-linear augmentation $\varepsilon_W\colon W_0\to\Z$ with
$\varepsilon_W\circ d=0$.  Let $V$ be a complex of $\Z[G]$-modules with
$V_i=0$ for $i<0$ and a $\Z[G]$-linear augmentation $\varepsilon_V$ for which
$\cdots\to V_1\to V_0\to\Z\to0$ is exact.  Then there is a $G$-equivariant
chain map $\psi\colon W\to V$ lifting the identity of $\Z$, that is
$\varepsilon_V\circ\psi=\varepsilon_W$, and any two such maps are
$G$-equivariantly homotopic.
\end{theorem}

Note that only the source has to be free and only the target has to be
exact.

\begin{corollary}\label{cor:resolutions-homotopy-equivalent}
Any two free resolutions of $\Z$ over $\Z[G]$ are $G$-equivariantly homotopy
equivalent, by maps lifting the identity of $\Z$ which are unique up to
equivariant homotopy.
\end{corollary}

For a complex $N$ of $\Z[G]$-modules, an equivariant homotopy equivalence
$W\to W'$ induces a homotopy equivalence of the balanced tensor products
$W\otimes_{\Z[G]}N\to W'\otimes_{\Z[G]}N$ of \Cref{sec:hom-complex}.
So by \Cref{cor:resolutions-homotopy-equivalent} the homology of
$W\otimes_{\Z[G]}N$ does not depend on the choice of $W$.  In
\Cref{sec:resolutions-and-generalized-cup-i}, the comparison theorem gives
an equivariant chain map $\psi\colon W\to\cX(r)$ and its uniqueness up to
equivariant homotopy.

\section{Cohomology Operations}
\label{sec:cohomology-operations}
For topological CSS codes, that is, codes arising from the cochain complexes of spaces, \emph{cohomology operations} can be used to obtain cohomology
invariants used in \Cref{sec:logical-gates-from-invariants} and hence logical gates.

There are two sources of interesting cochain operations. The first is the
failure of the cup product to be graded commutative on cochains, although
it is graded commutative on cohomology. The homotopies repairing this
failure, the homotopies between these homotopies, and so on, are new
cochain operations. Suitable combinations of them descend to cohomology.
We use operads and free resolutions to organize these homotopies.

The second is the lifting of coefficients. A cocycle with
$\Z/\theta$ coefficients can be lifted to a cochain with finer
coefficients such as $\Z/\theta^k$ or $\Z$. The lift need not be a
cocycle, but its coboundary is divisible by $\theta$. This can be used
to build operations whose output has finer coefficients than their input.

Steenrod's construction of the Pontryagin square combines both
ingredients and serves as our worked example.
We recall these classical constructions in preparation for the twisted
cohomology operations of
\Cref{sec:constructing-twisted-cohomology-operations}.

Throughout this section, $X$ denotes a locally ordered simplicial
complex.

\subsection{Cochain and cohomology operations}
\label{sec:cochain-and-cohomology-operations}
\begin{definition}\label{def:cochain-operation}
Fix input degrees $q_1,\ldots,q_r$ and an output degree $m$.  A
\emph{cochain operation} is a family of maps
\[
 \Phi_X\colon
 C^{q_1}(X;R_1)\times\cdots\times C^{q_r}(X;R_r)
 \longrightarrow C^m(X;R),
\]
not necessarily linear, which is natural under order-preserving simplicial
maps.  It is a \emph{cohomology operation} if it descends to cohomology,
that is, if it induces a well-defined map
\[
 H^{q_1}(X;R_1)\times\cdots\times H^{q_r}(X;R_r)
 \longrightarrow H^m(X;R).
\]
\end{definition}
A cochain operation is an explicit formula on representatives, whereas a
cohomology operation remembers only the resulting cohomology class.

The first and prototypical example of a cochain operation is the cup product. 
\begin{example}\label{ex:cup-product}
 For
$a\in C^p(X;R)$ and $b\in C^q(X;R)$ it is defined on an ordered simplex by
\[
 (a\cup b)[v_0,\ldots,v_{p+q}]
 =a[v_0,\ldots,v_p]\,b[v_p,\ldots,v_{p+q}],
\]
that is, the simplex is cut at the vertex $v_p$ into a front face for $a$
and a back face for $b$.  It satisfies the Leibniz rule
\[
 d(a\cup b)=da\cup b+(-1)^p a\cup db,
\]
so the product of two cocycles is again a cocycle.  Moreover, if $a$ and
$b$ are cocycles and $u\in C^{p-1}(X;R)$, $w\in C^{q-1}(X;R)$, then
\[
 (a+du)\cup b-a\cup b=du\cup b=d(u\cup b),
 \qquad
 a\cup(b+dw)-a\cup b=a\cup dw=(-1)^p\,d(a\cup w).
\]
Thus the cup product descends to a cohomology operation
\[
 \cup\colon H^p(X;R)\times H^q(X;R)\longrightarrow H^{p+q}(X;R).
\]
\end{example}

The cup product makes cohomology a graded-commutative ring,
with $[a]\cup[b]=(-1)^{pq}[b]\cup[a]$.  But on cochains it is only
associative and not graded commutative.  This is where the interesting
structure comes from.  Many of the cohomology operations in this paper will be
derived from $r$-fold cup products and the failure of the cup product to be
graded commutative. The following example describes this failure for the cup product of two cochains.

\begin{example}[The cup-$i$ products]\label{ex:higher-cup-products}
Let $a\in C^p(X;R)$ and $b\in C^q(X;R)$ be cocycles.  Since the cup product
is graded commutative on cohomology, $a\cup b$ and $(-1)^{pq}b\cup a$ differ
by a coboundary.  Steenrod~\cite{Steenrod1947} constructed a cochain operation
\[
 \cup_1\colon C^p(X;R)\otimes C^q(X;R)\longrightarrow C^{p+q-1}(X;R)
\]
defined for arbitrary cochains and satisfying
\[
 a\cup b-(-1)^{pq}\,b\cup a=(-1)^{p+q-1}\,d(a\cup_1 b) + (-1)^{p+q}da \cup_1 b +(-1)^q a \cup_1 db.
\]
  In other words, $\cup_1$ is a homotopy between
the cup product and the cup product with its factors swapped.  Now $\cup_1$
is again not commutative, but only commutative up to a further homotopy
$\cup_2$, which in turn is commutative only up to $\cup_3$, and so on.
Steenrod constructed this whole tower of \emph{cup-$i$ products}
\[
 \cup_i\colon C^p(X;R)\otimes C^q(X;R)\longrightarrow C^{p+q-i}(X;R),
 \qquad \cup_0=\cup,
\]
and showed that they satisfy \emph{Steenrod's formula}
\begin{equation}\label{eq:steenrods-formula}
\begin{aligned}
 d(a\cup_i b)={}&
 (-1)^{p+q-i}a\cup_{i-1}b
 +(-1)^{pq+p+q}b\cup_{i-1}a\\
 &+da\cup_i b+(-1)^p a\cup_i db,
\end{aligned}
\end{equation}
with $\cup_{-1}=0$.  For $i=0$ this is the Leibniz rule, and for $i\geq1$
it says that $\cup_i$ is a homotopy between $\cup_{i-1}$ and
$\cup_{i-1}$ with swapped arguments.

The cup-$i$ products are given by explicit formulas on ordered simplices,
so they are defined for simplicial as well as for singular cochains.  Over
$\F_2$, where no signs have to be tracked, put $n=p+q-i$.  Then
\begin{equation}\label{eq:cup-i-formula}
\begin{aligned}
 (a\cup_i b)[v_0,\ldots,v_n]
 =\sum_{0\leq j_0<j_1<\cdots<j_i\leq n}
 &a[v_0,\ldots,v_{j_0},v_{j_1},\ldots,v_{j_2},\ldots]\\
 \cdot\,&b[v_{j_0},\ldots,v_{j_1},v_{j_2},\ldots,v_{j_3},\ldots].
\end{aligned}
\end{equation}
The cuts $j_0<\cdots<j_i$ divide the ordered vertex list into $i+2$
consecutive blocks sharing their endpoints. The cochain $a$ is evaluated
on the union of the blocks with even index and $b$ on the union of the
blocks with odd index, starting with index zero. Only those cuts contribute
for which these faces have dimensions $p$ and $q$. For $i=0$ this is the cup
product.  For $i=1$ the second cut is forced to be $j_1=j_0+q$, and over the
integers the formula reads
\begin{equation}\label{eq:cup-one-formula}
 (a\cup_1 b)[v_0,\ldots,v_{p+q-1}]
 =\sum_{j=0}^{p-1}(-1)^{(p-j)(q+1)}\,
 a[v_0,\ldots,v_j,v_{j+q},\ldots,v_{p+q-1}]\,
 b[v_j,\ldots,v_{j+q}].
\end{equation}
The two faces share exactly the vertices $v_j$ and $v_{j+q}$.
\end{example}

From these homotopies one obtains new cohomology operations, for example, the Steenrod squares.

\begin{example}[Steenrod squares]\label{ex:steenrod-squares}
For $x\in Z^q(X;\Z/2)$ and $0\leq i\leq q$, the two $\cup_{i-1}$-terms of Steenrod's
formula~\eqref{eq:steenrods-formula} with $a=b=x$ cancel modulo $2$, so
$x\cup_ix$ is a cocycle, and the same formula gives
$(x+dy)\cup_i(x+dy)-x\cup_ix=d\bigl(dy\cup_{i+1}x+y\cup_idy+y\cup_{i-1}y\bigr)$.
Hence the \emph{Steenrod squares}~\cite{Steenrod1947}
\[
 Sq^{q-i}\colon H^q(X;\Z/2)\longrightarrow H^{2q-i}(X;\Z/2),
 \qquad
 [x]\longmapsto[x\cup_i x],
\]
are cohomology operations.  Only the top square $Sq^q[x]=[x\cup x]$ comes
from the cup product itself; all others exist because $\cup$ fails to be
commutative on cochains.
\end{example}

We also want to be able to lift coefficients, which is another source of
cohomology operations.  For example, operations from $\Z/2$ to $\Z/8$ will
later yield our construction of the $T$ gate.  We show two examples of
cohomology operations doing that.

\begin{example}[The integral Bockstein]\label{ex:integral-bockstein}
Let $\theta\geq2$ and fix the set-theoretic section
\[
 \iota_\theta\colon\Z/\theta\longrightarrow\Z,
 \qquad
 \iota_\theta(\overline k)=k\quad(0\leq k<\theta).
\]
Using the basis of simplices, we apply it coefficientwise to cochains.  This
is not additive, but it lifts: for $x\in C^q(X;\Z/\theta)$ the integral
cochain $\widetilde x=\iota_\theta x$ reduces to $x$ modulo $\theta$.  If $x$
is a cocycle, then $d\widetilde x$ reduces to $dx=0$, so every coefficient
of $d\widetilde x$ is divisible by $\theta$ and we may define
\begin{equation}\label{eq:integral-bockstein}
 \beta_\theta: Z^q(X;\mathbb Z/\theta) \longrightarrow Z^{q+1}(X;\mathbb Z), 
 \qquad 
 \beta_\theta(x)=\frac{d\widetilde x}{\theta}\in C^{q+1}(X;\Z).
\end{equation}
It is a cocycle, because $\theta\,d\beta_\theta(x)=dd\widetilde x=0$ and
integral cochains have no torsion.  If $y=x+dc$ is a cohomologous cocycle,
then $\widetilde y-\widetilde x$ and $d\widetilde c$ agree modulo $\theta$,
so $\widetilde y-\widetilde x=d\widetilde c+\theta b$ for an integral
cochain $b$, and therefore $\beta_\theta(y)-\beta_\theta(x)=db$.  Thus
\[
 \beta_\theta\colon H^q(X;\Z/\theta)\longrightarrow H^{q+1}(X;\Z)
\]
is a cohomology operation, the \emph{integral Bockstein}.  Note that the
division by $\theta$ is only possible because $x$ is a cocycle, and that the
lift $\widetilde x$ itself does not descend to cohomology: it need not be a
cocycle.
\end{example}

\begin{example}[The Pontryagin square on cochains]
\label{ex:pontryagin-square-formula}
Let $q$ be even, let $x\in Z^q(X;\Z/2)$ and let $a=\iota_2x$ be its
$0/1$-lift, regarded as a cochain with $\Z/4$ coefficients.  Steenrod's
formula for the \emph{Pontryagin square}~\cite{Steenrod1957} is
\begin{equation}\label{eq:pontryagin-square-cochain-formula}
 \mathfrak P(x)=a\cup a+a\cup_1da\in C^{2q}(X;\Z/4).
\end{equation}
It refines the cup square: since $da$ is even, $a\cup_1 da$ vanishes modulo
$2$, so $\mathfrak P(x)$ reduces to $x\cup x$.  To see that it is a cocycle
modulo $4$, write $da=2b$ with $b=\beta_2(x)$ as in
\Cref{ex:integral-bockstein}.  Since $q$ is even, the Leibniz rule
gives
\[
 d(a\cup a)=da\cup a+a\cup da=2\,b\cup a+2\,a\cup b,
\]
and Steenrod's formula~\eqref{eq:steenrods-formula} for $i=1$, applied to
$a$ and $da$, gives
\[
 d(a\cup_1da)=a\cup da-da\cup a+da\cup_1da
 =2\,a\cup b-2\,b\cup a+4\,b\cup_1b.
\]
Adding up,
\[
 d\mathfrak P(x)=4\,a\cup b+4\,b\cup_1 b\equiv0\pmod 4.
\]
The correction term $a\cup_1 da$ is exactly what cancels the failure of
$a\cup a$ to be a cocycle modulo $4$, and it uses both sources of cochain
operations at once: a cup-$i$ product and a lifted coefficient.  That
$\mathfrak P(x)$ depends only on the class of $x$ can also be checked by
hand, but it is more instructive to deduce it from the general machinery;
we do this in \Cref{ex:pontryagin-square}.
\end{example}

One could try to construct all kinds of cohomology operations from these
examples in the same way.  For example, there are also the higher
Pontryagin powers~\cite{Thomas1957}.  But it is difficult to handle such
explicit formulas and difficult to generalise them.  Furthermore, they do
not suffice for the operations constructed later which yield gates at
higher levels of the Clifford hierarchy.  What is needed is a systematic way to produce all
homotopies between permuted cup products together with their boundaries.
This is the subject of the next section.

\subsection{Constructing \texorpdfstring{$r$}{r}-cochain operations from operads}
\label{sec:higher-cochain-operations}

The cup-$i$ products describe homotopies involving two inputs.
Composing them with cup products gives homotopies between products
with more factors, but their binary identities do not specify the
higher homotopies relating these compositions.
The $E_\infty$ structure on cochains supplies compatible homotopies
for permutations of any number of inputs. We describe this structure
using the surjection operad.
We use the tensor products and hom complexes of
\Cref{sec:hom-complex}.

\begin{definition}[Differential graded operad]
\label{def:dg-operad}
A \emph{differential graded operad} $\mathcal O$ is a set of chain
complexes $\mathcal O(r)$, $r\geq1$, each carrying an action of the
symmetric group $\Sigma_r$.  Furthermore there are substitution operations
\[
 u(u_1,\ldots,u_r)\in
 \mathcal O(s_1+\cdots+s_r),
 \qquad
 u\in\mathcal O(r),\quad u_j\in\mathcal O(s_j),
\]
which are associative, equivariant and compatible with the differential,
and there is a unit operation $1\in\mathcal O(1)$.
\end{definition}

An operad can be thought of as a blueprint for the rules of some algebraic
operations.  The chain complex $\mathcal O(r)$ consists of the $r$-input operations, the
$\Sigma_r$-action says how to relabel their inputs, substitution says how to
plug the output of one into an input of another, and the differential
records boundaries.  Nothing here refers to cochains yet.  The following
example and definition will connect the blueprint to actual operations.

\begin{example}[The endomorphism operad]
\label{ex:endomorphism-operad}
For a cochain complex $A$, its \emph{endomorphism operad} is
\[
 \operatorname{End}_A(r)=\operatorname{hom}(A^{\otimes r},A).
\]
We use the homological grading of \Cref{sec:hom-complex}.
Thus an element $\nu\in\operatorname{End}_A(r)_s$ consists of
multilinear maps
\[
 \nu\colon A^{q_1}\otimes\cdots\otimes A^{q_r}
 \longrightarrow A^{q_1+\cdots+q_r-s}.
\]
For $\nu\in\operatorname{End}_A(r)_s$ the
differential is
\begin{equation}\label{eq:endomorphism-differential}
 d_{\mathrm{End}}\nu
 =d\circ\nu-(-1)^s\nu\circ d_{A^{\otimes r}}.
\end{equation}
Here $d_{A^{\otimes r}}$ is the tensor-product differential
of \Cref{sec:hom-complex}.
Thus $d_{\mathrm{End}}\nu=\nu'-\nu''$ says precisely that $\nu$ is a
homotopy between the operations~$\nu'$ and~$\nu''$.  The symmetric group
acts by permuting the arguments with the Koszul sign,
\[
 (\pi\nu)(a_1,\ldots,a_r)=\pm\,\nu(a_{\pi(1)},\ldots,a_{\pi(r)}).
\]
In particular, for the cycle $\rho=(1\,2\,\ldots\,r)$,
\begin{equation}\label{eq:cyclic-action-on-endomorphisms}
 (\rho\nu)(a_1,\ldots,a_r)
 =(-1)^{|a_1|(|a_2|+\cdots+|a_r|)}
 \nu(a_2,\ldots,a_r,a_1).
\end{equation}
This action commutes with $d_{\mathrm{End}}$.  Substitution is the
composition of multilinear maps,
\[
 \nu(\nu_1,\ldots,\nu_r)(a_1,\ldots,a_{s_1+\cdots+s_r})
 =\pm\,\nu\bigl(\nu_1(a_1,\ldots,a_{s_1}),\ldots,
 \nu_r(a_{s_1+\cdots+s_{r-1}+1},\ldots,a_{s_1+\cdots+s_r})\bigr),
\]
where the sign is the Koszul sign of moving each $\nu_j$ past the arguments
$a_1,\ldots,a_{s_1+\cdots+s_{j-1}}$.  It is convenient to substitute into a
single input only, $\nu\circ_j\nu'=\nu(1,\ldots,1,\nu',1,\ldots,1)$ with~$\nu'$ in position $j$.  For instance, associativity of the cup product is
the identity $\cup\circ_1\cup=\cup\circ_2\cup$, and the three-input
operation $(a,b,c)\mapsto a\cup_1(b\cup c)$ is $\cup_1\circ_2\cup$.
The cup product, the cup-$i$ products, and the multilinear
operations used in our twisted constructions are elements of
$\operatorname{End}_{C^*(X;R)}$.
\end{example}

\begin{definition}[Algebra over an operad]
\label{def:algebra-over-operad}
An \emph{algebra over $\mathcal O$} is a cochain complex $A$ together with a
morphism of differential graded operads
\[
 \mathcal O\longrightarrow\operatorname{End}_A.
\]
It turns every abstract element of $\mathcal O(r)$ into an actual
$r$-input operation on $A$, compatibly with differentials, permutations,
and substitution.
\end{definition}

Which operad one chooses decides which structure is being imposed.
\begin{example}
  Let $\mathcal O$ be the operad with $\mathcal O(r) = R$ with the trivial $\Sigma_r$-action in degree 0 and nothing
  in higher degrees. An algebra $A$ over $\mathcal O$ picks an operation 
  \[
  \mu: R \to \operatorname{End}_A(r), \quad 1 \mapsto (\mu_r: A^{\otimes r} \to A)
  \]
  for each arity $r$. Substitution imposes the relation $\mu_r \circ_k \mu_s = \mu_{r+s-1}$.
  It is straightforward that this defines a strictly commutative differential graded algebra. 
\end{example}

The operad used here is the \emph{surjection operad} $\cX$ of
McClure--Smith~\cite{McClureSmith2003} and
Berger--Fresse~\cite{BergerFresse2004}.  We display a surjection
$u\colon\{1,\ldots,r+i\}\twoheadrightarrow\{1,\ldots,r\}$ as a word
\[
 u=(u_1,\ldots,u_{r+i}),\qquad
 u_j\in\{1,\ldots,r\},
\]
and call it non-degenerate when two adjacent labels are never equal.  Then
$\cX(r)_i$ is the free abelian group generated by the non-degenerate
surjections of length $r+i$; over a ring $R$ we tensor with $R$.  The
differential is the signed sum of the words obtained by deleting one entry.
Delete at first each 1 with alternating
signs, then each 2 and so on. The last term obtained by deleting $i$ has the same sign as
the first term obtained by deleting $i+1$. Terms with adjacent equal labels
or a missing label are zero. For example, the differential of the word $(121232)$ is
\begin{align*}
  \partial (121232) &= (21232) - (12232) - (11232) + (12132) - (12123) \\
  &= (21232) + (12132) - (12123),
\end{align*}
where deleting the only occurrence of the label $3$ gives a non-surjection
and has been omitted.
This makes $\cX(r)$ a chain complex.  The action of $\Sigma_r$ is given by
relabelling with a sign,
\[
 \pi u=\epsilon(\pi,u)\,(\pi(u_1),\ldots,\pi(u_{r+i})),
 \qquad
 \epsilon(\pi,u)=\prod_{\substack{k<l\\ \pi(k)>\pi(l)}}(-1)^{(n_k-1)(n_l-1)},
\]
where $n_k$ is the number of occurrences of the label $k$ in $u$, and thus permutes the
inputs.
Substitution inserts one word into another: if the label $j$
occurs $k$ times in $u$, then $u\circ_j u'$ is the signed sum over all ways
of cutting $u'$ into $k$ consecutive pieces which overlap in one label, the
pieces replacing the occurrences of $j$ and the labels being shifted
accordingly; for details see~\cite{McClureSmith2003,BergerFresse2004}.  For
example,
\[
 (1,2,1)\circ_2(1,2,3)=(1,\boldsymbol{2,3,4},1),
 \qquad
 (1,2,1)\circ_1(1,2)=\pm(\boldsymbol{1},3,\boldsymbol{1,2})\pm(\boldsymbol{1,2},3,\boldsymbol{2}).
\]

Cochains of a space are an algebra over the surjection
operad~\cite{McClureSmith2003,BergerFresse2004}.  A surjection in
$\cX(r)_i$ represents a cochain operation of arity $r$ and cohomological
degree $-i$, that is, an element of
$\operatorname{hom}(C^*(X;R)^{\otimes r},C^*(X;R))_i$.  What makes this
description useful is that each surjection is a compact and completely finite
recipe for evaluating such an operation on a single simplex.  For example, the element
$(1,3,1,2,1)\in\cX(3)_2$ will split the simplex into five overlapping
pieces, give the first, third and fifth to the first cochain, the fourth to
the second, and the second to the third.  We now make this precise.

Let $a_j\in C^{q_j}(X;R)$ and let $[v_0,\ldots,v_n]$ be an ordered simplex,
where $n=q_1+\cdots+q_r-i$.  For every weakly increasing sequence of cuts
\[
 0=p_0\leq p_1\leq\cdots\leq p_{r+i}=n,
\]
form the overlapping intervals
\[
 I_j=[v_{p_{j-1}},\ldots,v_{p_j}].
\]
For each label $k$, collect the intervals $I_j$ for which $u_j=k$ and
concatenate their vertices in the original order; call the resulting face
$F_k$.  Degenerate terms are set to zero.  The operation associated with
$u$ is
\begin{equation}\label{eq:interval-cut-formula}
 u(a_1,\ldots,a_r)[v_0,\ldots,v_n]
 =\sum_{\text{compatible cuts}}
 \mathord{\pm}\prod_{k=1}^r a_k(F_k).
\end{equation}
Here ``compatible'' means $\dim F_k=q_k$ for every $k$.  Over $\Z$, the
interval-cut and Koszul signs are those of
\cite{McClureSmith2003,KMM2021}; over $\F_2$ they disappear.  We use the
degree convention in which $(1,2)$ is exactly the cup product of
\Cref{ex:cup-product}.

\begin{example}[The cup-$i$ products in the surjection operad]
\label{ex:cup-i-surjections}
The degree-zero word $(1,2)$ has one cut, so it acts by
\[
 (1,2)(a,b)=a\cup b.
\]
Here, as before, simplices also denote their dual basis cochains.
For $v_0<v_1<v_2$, its two intervals are $[v_0,v_1]$ and $[v_1,v_2]$, and hence
\[
 [v_0,v_1]\cup[v_1,v_2]=[v_0,v_1,v_2].
\]
Visually,
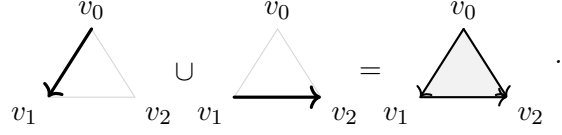
\begin{center}
\begin{tikzpicture}[x=0.82cm,y=0.82cm,baseline=(current bounding box.center),
                    line cap=round,line join=round]
 \begin{scope}
  \coordinate (a) at (0,1.1); \coordinate (b) at (-0.7,0);
  \coordinate (c) at (0.7,0);
  \draw[gray!32] (a)--(c)--(b);
  \draw[->,very thick] (a)--(b);
  \node[above] at (a) {$v_0$}; \node[below left] at (b) {$v_1$};
  \node[below right] at (c) {$v_2$};
 \end{scope}
 \node at (1.5,0.4) {$\cup$};
 \begin{scope}[shift={(3,0)}]
  \coordinate (a) at (0,1.1); \coordinate (b) at (-0.7,0);
  \coordinate (c) at (0.7,0);
  \draw[gray!32] (a)--(b) (a)--(c);
  \draw[->,very thick] (b)--(c);
  \node[above] at (a) {$v_0$}; \node[below left] at (b) {$v_1$};
  \node[below right] at (c) {$v_2$};
 \end{scope}
 \node at (4.5,0.4) {$=$};
 \begin{scope}[shift={(6,0)}]
  \coordinate (a) at (0,1.1); \coordinate (b) at (-0.7,0);
  \coordinate (c) at (0.7,0);
  \fill[gray!10] (a)--(b)--(c)--cycle;
  \draw[->,thick] (a)--(b); \draw[->,thick] (a)--(c);
  \draw[->,thick] (b)--(c);
  \node[above] at (a) {$v_0$}; \node[below left] at (b) {$v_1$};
  \node[below right] at (c) {$v_2$};
 \end{scope}
\end{tikzpicture}.
\end{center}

The degree-one word $(1,2,1)$ gives the first cup-$i$ product, $\cup_1$.  Its cuts
\[
 [v_0,\ldots,v_j]\mid[v_j,\ldots,v_{j+q}]\mid[v_{j+q},\ldots,v_n]
\]
are exactly the ones in~\Cref{eq:cup-one-formula}, and for
$a\in C^p(X;R)$ and $b\in C^q(X;R)$
\[
 (1,2,1)(a,b)=\pm\,a\cup_1 b,
\]
with a sign which only depends on $p$ and $q$.
On the triangle, the cut $[v_0]\mid[v_0,v_1,v_2]\mid[v_2]$ gives
\[
 [v_0,v_2]\cup_1[v_0,v_1,v_2]=-[v_0,v_1,v_2].
\]
Visually,
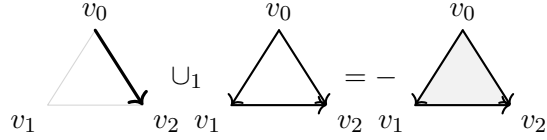
\begin{center}
\begin{tikzpicture}[x=0.9cm,y=0.9cm,baseline=(current bounding box.center),
                    line cap=round,line join=round]
 \begin{scope}
  \coordinate (a) at (0,1.1); \coordinate (b) at (-0.7,0);
  \coordinate (c) at (0.7,0);
  \draw[gray!32] (a)--(b)--(c);
  \draw[->,very thick] (a)--(c);
  \node[above] at (a) {$v_0$}; \node[below left] at (b) {$v_1$};
  \node[below right] at (c) {$v_2$};
 \end{scope}
 \node at (1.35,0.4) {$\cup_1$};
 \begin{scope}[shift={(2.7,0)}]
  \coordinate (a) at (0,1.1); \coordinate (b) at (-0.7,0);
  \coordinate (c) at (0.7,0);
  \draw[->,thick] (a)--(b); \draw[->,thick] (a)--(c);
  \draw[->,thick] (b)--(c);
  \node[above] at (a) {$v_0$}; \node[below left] at (b) {$v_1$};
  \node[below right] at (c) {$v_2$};
 \end{scope}
 \node at (4.05,0.4) {$=-$};
 \begin{scope}[shift={(5.4,0)}]
  \coordinate (a) at (0,1.1); \coordinate (b) at (-0.7,0);
  \coordinate (c) at (0.7,0);
  \fill[gray!10] (a)--(b)--(c)--cycle;
  \draw[->,thick] (a)--(b); \draw[->,thick] (a)--(c);
  \draw[->,thick] (b)--(c);
  \node[above] at (a) {$v_0$}; \node[below left] at (b) {$v_1$};
  \node[below right] at (c) {$v_2$};
 \end{scope}
\end{tikzpicture}
\end{center}
Over $\F_2$ the minus sign is, of course, invisible.

More generally, let $\alpha_i=(1,2,1,2,\ldots)\in\cX(2)_i$ over $\F_2$ be the
alternating word with $i+2$ entries.  Its interval cuts are exactly the cuts
$j_0<\cdots<j_i$ of~\Cref{eq:cup-i-formula}, so over $\F_2$
\[
 a\cup_i b=\alpha_i(a,b),
 \qquad
 \cup=(1,2),\quad \cup_1=(1,2,1),\quad \cup_2=(1,2,1,2),\ldots,
\]
and over the integers this holds up to signs depending on the degrees.

In arity two nothing else exists: $\cX(2)_i$ has the basis $\alpha_i$ and
$\tau\alpha_i=\pm(2,1,2,\ldots)$, where $\tau$ is the transposition.  Hence
$\cX(2)_i=\Z[C_2]\alpha_i$ is free of rank one over the group ring of
$C_2=\langle\tau\rangle$. The rule for the operadic boundary $\partial$ gives
\[
 \partial(1,2,1)=(2,1)-(1,2),
 \qquad
 \partial(1,2,1,2)=(2,1,2)+(1,2,1).
\] 
More generally,
\[
 \partial\alpha_{2k+1}=(\tau-1)\alpha_{2k},
 \qquad
 \partial\alpha_{2k+2}=(1+\tau)\alpha_{2k+1}.
\]
These identities exhibit $\cX(2)$ as the periodic free resolution
\begin{equation}\label{eq:w-two}
 W(2)=\bigl\{\Z[C_2]e_0\xleftarrow{\ \tau-1\ }\Z[C_2]e_1
 \xleftarrow{\ 1+\tau\ }\Z[C_2]e_2\xleftarrow{\ \tau-1\ }\cdots\bigr\}
\end{equation}
of $\Z$ over $\Z[C_2]$, with $e_i\mapsto\alpha_i$.  Composing with the
action on cochains, the whole tower of cup-$i$ products is organised in a
single $C_2$-equivariant chain map, which is how Steenrod~\cite{Steenrod1953}
phrased it:
\begin{equation}\label{eq:steenrods-map}
\begin{aligned}
 \mu^{(2)}\colon W(2)\otimes_{\Z[C_2]}C^*(X)^{\otimes2}
 &\longrightarrow C^*(X),\\
 e_i\otimes a\otimes b&\longmapsto\alpha_i(a,b)=\pm\,a\cup_i b.
\end{aligned}
\end{equation}
Here $\tau$ acts on $C^*(X)^{\otimes2}$ by swapping the factors with the
Koszul sign, and $W(2)$ indexes the homotopies. Since this action is a chain map the operadic boundary recovers 
Steenrod's formula in \Cref{eq:steenrods-formula}.
\end{example}

The same mechanism works in higher arity, and this is where the surjection
operad goes beyond Steenrod's cup-$i$ products.  The word $(1,2,\ldots,r)$
acts as the $r$-fold cup product.  The word
$(1,2,3,4,1)=(1,2,1)\circ_2(1,2,3)$ acts by
\[
 (1,2,3,4,1)(a,b,c,d)=\pm\,a\cup_1(b\cup c\cup d),
\]
and its boundary is
\[
 \partial(1,2,3,4,1)=(2,3,4,1)-(1,2,3,4).
\]
Thus it is a homotopy between the fourfold cup product $a\cup b\cup c\cup d$
and the rotated product $b\cup c\cup d\cup a$.  In $\cX(4)$ there are many
more such homotopies, and homotopies between them, for every permutation of
the four inputs.  The next section explains how to pick a small and
coherent family among them.

\subsection{Resolutions and generalized cup-\texorpdfstring{$i$}{i} products}
\label{sec:resolutions-and-generalized-cup-i}
In this section we construct higher arity analogues of the cup-$i$ products.
The surjection operad contains far too many operations to work with
directly.  A subgroup $G\leq\Sigma_r$ of the permutations of the inputs,
together with a small free resolution of $G$, singles out a family of
operations which is still coherent, but explicit and finite in every degree.
These operations will be the building blocks for cohomology operations on
spaces with a $G$-symmetry.  We start with the property of the surjection
operad which makes this possible.

\begin{definition}[$E_\infty$ operad]
\label{def:e-infinity-operad}
A differential graded operad $\mathcal E$ is an \emph{$E_\infty$ operad} if
for every $r$ it holds that
\begin{enumerate}
 \item[(i)] the action of $\Sigma_r$ on $\mathcal E(r)$ is free, and
 \item[(ii)] the augmentation $\mathcal E(r)\to R$ is a quasi-isomorphism.
\end{enumerate}
Equivalently, each $\mathcal E(r)$ is a free resolution of $R$ over
$R[\Sigma_r]$.  An \emph{$E_\infty$ structure} on a cochain complex $A$ is a
morphism $\mathcal E\to\operatorname{End}_A$, that is, the structure of an
algebra over $\mathcal E$.
\end{definition}

Two operations of the same arity and degree with the same boundary differ
by a cycle. In positive degrees this cycle is a boundary by condition (ii).
In degree zero, the same holds when their augmentations agree. Condition (i) is what makes
the choice coherent. Since $\Sigma_r$ acts freely, a choice made on one representative of each orbit 
propagates equivariantly to all permutations.  For $r=2$ we have already seen this in
\Cref{ex:cup-i-surjections}, where $\cX(2)$ turned out to be the
resolution $W(2)$.  In general, $\Sigma_r$ permutes the non-degenerate
surjections freely, since a surjection uses every label, and each augmentation
$\cX(r)\to\Z$ is a quasi-isomorphism~\cite{McClureSmith2003,BergerFresse2004}.  Hence the surjection
operad is an $E_\infty$ operad, and simplicial cochains carry an
$E_\infty$ structure.

Freeness is inherited by subgroups, so for every subgroup $G\leq\Sigma_r$
the complex $\cX(r)$ is also a free resolution of $\Z$ over $\Z[G]$.  It is,
however, a very large one.  We now explain how to use a small resolution to pick a coherent family 
of surjections of $\cX(r)$ that describes the permutation by elements of $G$.

Let $G\leq\Sigma_r$ and let $W$ be a free $\Z[G]$-resolution of $\Z$
as in \Cref{sec:resolutions}. We always choose a $\Z[G]$-basis
$\{e_{i,j}\}_{j\in J_i}$ of $W_i$ and assume that $W_0=\Z[G]e_0$ is free of
rank one with $\varepsilon(e_0)=1$, where $e_0=e_{0,0}$.  The differential is then recorded by
coefficients $\gamma^{i,j}_{j'}\in\Z[G]$,
\[
 de_{i,j}=\sum_{j'\in J_{i-1}}\gamma^{i,j}_{j'}\,e_{i-1,j'}.
\]
\begin{example}[The cyclic resolution]
\label{ex:cyclic-resolution}
For $r\geq1$, let $C_r=\langle\rho\mid\rho^r=1\rangle$, embedded in
$\Sigma_r$ by $\rho\mapsto(1\,2\,\ldots\,r)$, and set
\[
 T=\rho-1,
 \qquad
 N=1+\rho+\cdots+\rho^{r-1}
 \quad\text{in }\Z[C_r].
\]
Steenrod's periodic resolution $W(r)$~\cite{Steenrod1953} is the
augmented chain complex
\begin{equation}\label{eq:cyclic-resolution}
 \cdots\xrightarrow{N}\Z[C_r]e_3
 \xrightarrow{T}\Z[C_r]e_2
 \xrightarrow{N}\Z[C_r]e_1
 \xrightarrow{T}\Z[C_r]e_0
 \xrightarrow{\varepsilon}\Z\longrightarrow0.
\end{equation}
Here $\varepsilon$ sends every element of $C_r$ to $1$, and
\begin{equation}\label{eq:cyclic-resolution-differential}
 de_{2k+1}=T e_{2k},
 \qquad
 de_{2k+2}=N e_{2k+1}.
\end{equation}
Exactness follows from $TN=NT=0$ together with
$\ker T=N\Z[C_r]$ and $\ker N=T\Z[C_r]$.  In particular,
$W(r)_i=\Z[C_r]e_i$ has a single basis element in every degree, and for
$r=2$ we recover the resolution $W(2)$ of~\Cref{eq:w-two}.
\end{example}

Let $A=C^*(X;R)$.  In \Cref{sec:higher-cochain-operations} we saw that
Steenrod's cup-$i$ products are a $C_2$-equivariant chain map
$W(2)\to\operatorname{End}_A(2)$, $e_i\mapsto\pm\cup_i$.  We now generalise
this by constructing, for every subgroup $G\leq\Sigma_r$ and every free
resolution $W$ of $G$, a $G$-equivariant chain map
\begin{equation}\label{eq:w-action}
 \mu\colon W\longrightarrow\operatorname{End}_A(r).
\end{equation}
The map $\mu$ can be conveniently constructed through the factorization
\begin{equation}\label{eq:w-surjection-factorization}
 W\xrightarrow{\ \psi\ }\cX(r)
 \longrightarrow\operatorname{End}_A(r).
\end{equation}
The second arrow is the action of the surjection operad described in
\Cref{sec:higher-cochain-operations}. For the first arrow, recall that
$\cX(r)$ restricted to $G$ is a free $\Z[G]$-resolution of $\Z$.
\Cref{thm:comparison} therefore gives a $G$-equivariant chain map
$\psi\colon W\to\cX(r)$ lifting the identity of $\Z$.
We normalize $\psi$ by
\[
 \psi(e_0)=(1,2,\ldots,r).
\]

The map $\psi$ can be written down explicitly.  Let $h$ be a contraction of
$\cX(r)$ onto $\Z$, that is a map of degree $+1$ with
\[
 dh+hd=\mathrm{id}-\eta\varepsilon,
 \qquad \eta(1)=(1,2,\ldots,r).
\]
The contraction need not be equivariant.  We define $\psi$ on the basis and extend $\Z[G]$-linearly. 
We do this recursively by
\begin{equation}\label{eq:general-comparison-recursion}
 \psi(e_{i,j})= h \psi(d e_{i,j}) = h\Bigl(\sum_{j'\in J_{i-1}}\gamma^{i,j}_{j'}\,
 \psi(e_{i-1,j'})\Bigr).
\end{equation}
The augmentation is preserved because $\varepsilon\psi(e_0)=1$.
For $i\geq1$, assume that $\psi$ commutes with the differential in
lower degrees. Then
\[
 d\psi(de_{i,j})=0,\qquad \varepsilon\psi(de_{i,j})=0,
\]
where the second equality uses $\varepsilon d=0$ when $i=1$. Hence
\begin{align*}
 d\psi(e_{i,j})
 &=dh\psi(de_{i,j})\\
 &=(\mathrm{id}-\eta\varepsilon-hd)\psi(de_{i,j})\\
 &=\psi(de_{i,j}).
\end{align*}
Kaufmann and Medina-Mardones~\cite[Section~6.2]{KMM2021} give an explicit
contraction $h$ of the surjection operad, so that
\Cref{eq:general-comparison-recursion} yields explicit formulas.

The uniqueness part of \Cref{thm:comparison} makes the construction
well posed. For a fixed resolution $W$, two choices of $\psi$, for
instance from two different contractions $h$, are $G$-equivariantly
homotopic. Two different resolutions $W$ and $W'$ are related by a
comparison map, unique up to equivariant homotopy, and composing it with
$\psi'$ gives a map homotopic to $\psi$. In both cases the cohomology
operations of \Cref{sec:constructing-twisted-cohomology-operations}
agree after the universal classes are transported along the comparison
map. The cochain formulas, and hence the circuits, may depend on the
choice.

\begin{definition}[Generalized cup-$i$ products]
\label{def:generalized-cup-i}
Let $G\leq\Sigma_r$, let $W$ be a free resolution of $G$ with basis
$\{e_{i,j}\}$, and let $\mu\colon W\to\operatorname{End}_A(r)$ be constructed
as above.  The \emph{cup-$(G,i,j)$ product} is the $r$-input operation of
degree $-i$
\[
 \mu^G_{i,j}=\mu(e_{i,j})\colon
 A^{q_1}\otimes\cdots\otimes A^{q_r}
 \longrightarrow A^{q_1+\cdots+q_r-i}.
\]
The index $i$ is the homotopy level and $j$ indexes the basis elements of
the $i$-th degree of the resolution.  We use the same symbol for the element
$\psi(e_{i,j})$ of the surjection operad and for the cochain operation it
defines.  Since $\mu$ is a chain map,
\[
 d_{\mathrm{End}}\mu^G_{i,j}
 =\sum_{j'\in J_{i-1}}\gamma^{i,j}_{j'}\,\mu^G_{i-1,j'},
\]
where $G$ acts by permuting the inputs as
in~\Cref{eq:cyclic-action-on-endomorphisms}.
\end{definition}

The cup-$(G,0,0)$ product is always the $r$-fold cup product
$\mu^G_{0,0}=(1,2,\ldots,r)$, and $\mu^G_{1,j}$ is a homotopy between it and
the permuted products prescribed by $de_{1,j}$.  The higher products are
homotopies between these homotopies, as prescribed by the resolution.

\begin{example}[The cup-$(r,i)$ products]
\label{ex:cup-r-i-products}
For $G=C_r$ and the cyclic resolution $W(r)$ of
\Cref{ex:cyclic-resolution} there is one basis element in each
degree, and we write
\[
 \mu^{(r)}_i=\mu^{C_r}_i=\mu(e_i)
\]
for the \emph{cup-$(r,i)$ product}.  The differential of $W(r)$ gives
\begin{align}
 d_{\mathrm{End}}\mu_0^{C_r}&=0,\notag\\
 d_{\mathrm{End}}\mu_{2k+1}^{C_r}
   &=(\rho-1)\mu_{2k}^{C_r},\notag\\
 d_{\mathrm{End}}\mu_{2k+2}^{C_r}
   &=N\mu_{2k+1}^{C_r}.
 \label{eq:first-mu-boundaries}
\end{align}
Here $\rho$ cyclically rotates the inputs with the Koszul sign.  Thus
$\mu_1^{C_r}$ is a homotopy between the $r$-fold cup product and its
rotation, the norm $N\mu_1^{C_r}$ of these homotopies is a cycle, and
$\mu_2^{C_r}$ is a chosen filler of it.  For the periodic resolution the
recursion in \Cref{eq:general-comparison-recursion} becomes the recursion
\begin{align}
 \psi(e_{2k+1})&=hT\psi(e_{2k}),\notag\\
 \psi(e_{2k+2})&=hN\psi(e_{2k+1}),
 \label{eq:surjection-comparison-recursion}
\end{align}
see~\cite[Definition~6.5 and Theorem~6.6]{KMM2021}.  For $r=2$ the
cup-$(2,i)$ products are Steenrod's cup-$i$ products up to sign, and for a
prime $r=p$ they give cochain representatives of the Steenrod
operations~\cite{KMM2021}.
\end{example}

\begin{example}[The arity-four operations used below]
\label{ex:arity-four-operations}
In arity four, write $\mu_i=\mu_i^{(4)}$. The explicit map $\psi$ of
Kaufmann and Medina-Mardones~\cite{KMM2021} gives
\begin{align}
 \mu_0={}&(1,2,3,4),\label{eq:arity-four-mu-zero}\\
 \mu_1={}&(1,2,3,4,1),\label{eq:arity-four-mu-one}\\
 \mu_2={}&(1,2,3,4,1,2)+(1,3,4,1,2,3)
 +(1,4,1,2,3,4)\notag\\
 &+(1,2,3,4,2,3)+(1,2,4,2,3,4)
 +(1,2,3,4,3,4).
 \label{eq:arity-four-mu-two}
\end{align}
Consequently
\[
 \mu_0(a,b,c,d)=a\cup b\cup c\cup d,
 \qquad
 \mu_1(a,b,c,d)=\pm\,a\cup_1(b\cup c\cup d).
\]
The operation $\mu_2$ is a genuine four-input operation that cannot be reduced to binary cup-$i$ products or iterations thereof.  If
$\rho=(1\,2\,3\,4)$ and $N=1+\rho+\rho^2+\rho^3$, then
\[
 d\mu_0=0,
 \qquad
 d\mu_1=(\rho-1)\mu_0,
 \qquad
 d\mu_2=N\mu_1,
\]
as one checks directly with the rule for $d$.
\end{example}

\begin{example}[Evaluating $\mu_2$ on a simplex]
Denote the six summands of~\Cref{eq:arity-four-mu-two}, in their
displayed order, by $u_1,\ldots,u_6$. They all induce a cochain operation
\[
u_i\in\operatorname{hom}(C^*(X)^{\otimes 4},C^*(X))_2.
\]
Let $X$ be the tetrahedron with vertices $v_0<v_1<v_2<v_3$ and work over $\F_2$.
Consider in degree $5$ the operation $u_2:(C^*(X)^{\otimes 4})^5 \to C^3(X)$.
Let $a,b,c,d$ be cochains of degrees $(2,1,1,1)$.
We want to evaluate $u_2(a,b,c,d)$ on $[v_0,v_1,v_2,v_3]$.

We first list the compatible cuts of $u_2=(1,3,4,1,2,3)$ in
\Cref{eq:interval-cut-formula}.  Write $\ell_j$ for the
length of the interval $I_j$.  The label $2$ occurs once, at position $5$, so
$\dim F_2=1$ forces $\ell_5=1$; likewise $\dim F_4=1$ forces $\ell_3=1$.  
Taking the remaining constraints into account it follows that there are exactly two
compatible cuts:
\[
\renewcommand{\arraystretch}{1.25}
\begin{array}{l|llll}
 \text{cut} & F_1 & F_2 & F_3 & F_4\\\hline
 {}[v_0]\mid[v_0]\mid[v_0,v_1]\mid[v_1,v_2]\mid[v_2,v_3]\mid[v_3]
  & [v_0,v_1,v_2] & [v_2,v_3] & [v_0,v_3] & [v_0,v_1]\\
 {}[v_0,v_1]\mid[v_1]\mid[v_1,v_2]\mid[v_2]\mid[v_2,v_3]\mid[v_3]
  & [v_0,v_1,v_2] & [v_2,v_3] & [v_1,v_3] & [v_1,v_2]
\end{array}
\]
For these input degrees, the following are the only two basis tensors on
which $u_2$ is nonzero. As before, simplices denote their dual basis cochains.
\begin{align*}
 u_2([v_0, v_1, v_2],[v_2, v_3],[v_0, v_3],[v_0, v_1] )=[v_0,v_1,v_2,v_3] \\ 
 u_2([v_0, v_1, v_2],[v_2, v_3],[v_1, v_3],[v_1, v_2] )=[v_0,v_1,v_2,v_3] 
\end{align*}
Both are sent to zero by the other five summands, so $\mu_2$ has the same
value on each. Writing either input tuple as $(a,b,c,d)$, we obtain
\[
 \mu_2(a,b,c,d)=[v_0,v_1,v_2,v_3].
\]
Diagrammatically, this equality reads as follows.
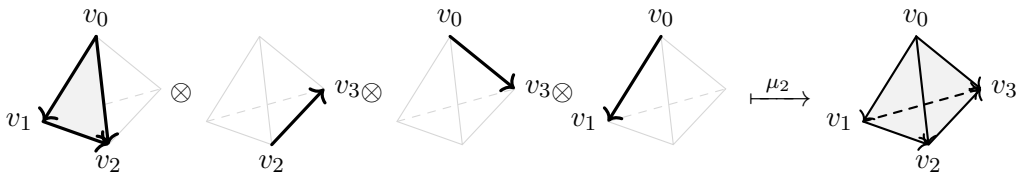
\begin{center}
\begin{tikzpicture}[x=0.9cm,y=0.9cm,baseline=(current bounding box.center),
                    line cap=round,line join=round]
 \begin{scope}[xshift=0cm]
  \coordinate (a) at (0,1.25); \coordinate (b) at (-0.78,0);
  \coordinate (c) at (0.18,-0.34); \coordinate (d) at (0.95,0.48);
  \draw[gray!32] (a)--(d) (c)--(d); \draw[gray!32,dashed] (b)--(d);
  \fill[gray!10] (a)--(b)--(c)--cycle;
  \draw[->,very thick] (a)--(b); \draw[->,very thick] (a)--(c);
  \draw[->,very thick] (b)--(c);
  \node[above] at (a) {$v_0$}; \node[left] at (b) {$v_1$};
  \node[below] at (c) {$v_2$};
 \end{scope}
 \node at (1.25,0.43) {$\otimes$};

 \begin{scope}[shift={(2.4,0)}]
  \coordinate (a) at (0,1.25); \coordinate (b) at (-0.78,0);
  \coordinate (c) at (0.18,-0.34); \coordinate (d) at (0.95,0.48);
  \draw[gray!32] (a)--(b)--(c) (a)--(c) (a)--(d);
  \draw[gray!32,dashed] (b)--(d);
  \draw[->,very thick] (c)--(d);
  \node[below] at (c) {$v_2$}; \node[right] at (d) {$v_3$};
 \end{scope}
 \node at (4.05,0.43) {$\otimes$};

 \begin{scope}[shift={(5.2,0)}]
  \coordinate (a) at (0,1.25); \coordinate (b) at (-0.78,0);
  \coordinate (c) at (0.18,-0.34); \coordinate (d) at (0.95,0.48);
  \draw[gray!32] (a)--(b)--(c) (a)--(c) (c)--(d);
  \draw[gray!32,dashed] (b)--(d);
  \draw[->,very thick] (a)--(d);
  \node[above] at (a) {$v_0$}; \node[right] at (d) {$v_3$};
 \end{scope}
 \node at (6.85,0.43) {$\otimes$};

 \begin{scope}[shift={(8.3,0)}]
  \coordinate (a) at (0,1.25); \coordinate (b) at (-0.78,0);
  \coordinate (c) at (0.18,-0.34); \coordinate (d) at (0.95,0.48);
  \draw[gray!32] (a)--(c)--(d) (a)--(d) (b)--(c);
  \draw[gray!32,dashed] (b)--(d);
  \draw[->,very thick] (a)--(b);
  \node[above] at (a) {$v_0$}; \node[left] at (b) {$v_1$};
 \end{scope}

 \node at (10.05,0.45) {$\xmapsto{\ \mu_2\ }$};
 \begin{scope}[shift={(12.05,0)}]
  \coordinate (a) at (0,1.25); \coordinate (b) at (-0.78,0);
  \coordinate (c) at (0.18,-0.34); \coordinate (d) at (0.95,0.48);
  \fill[gray!10] (a)--(b)--(c)--(d)--cycle;
  \draw[->,thick] (a)--(b); \draw[->,thick] (a)--(c);
  \draw[->,thick] (a)--(d); \draw[->,thick] (b)--(c);
  \draw[->,thick,dashed] (b)--(d); \draw[->,thick] (c)--(d);
  \node[above] at (a) {$v_0$}; \node[left] at (b) {$v_1$};
  \node[below] at (c) {$v_2$}; \node[right] at (d) {$v_3$};
 \end{scope}
\end{tikzpicture}
\end{center}

Which summand contributes depends on the inputs.  In total degree $6$, consider
\[
 \mu_2\colon(C^*(X)^{\otimes4})^6\to C^4(X)
\]
on the $4$-simplex
$X$ with vertices $v_0<\cdots<v_4$, and the cochains
\[
 a=[v_0,v_1,v_2],\quad b=[v_2,v_3],\quad c=[v_3,v_4],\quad
 d=[v_2,v_3,v_4]
\]
of degrees $(2,1,1,2)$.  Here $u_5=(1,2,4,2,3,4)$ has the two compatible
cuts
\begin{gather*}
 [v_0,v_1,v_2]\mid[v_2]\mid[v_2]\mid[v_2]\mid[v_2,v_3]\mid[v_3,v_4],\\
 {}[v_0,v_1,v_2]\mid[v_2]\mid[v_2,v_3]\mid[v_3]\mid[v_3,v_4]\mid[v_4].
\end{gather*}
The first is degenerate, since $F_2=[v_2,v_2]$, while the second has faces $F_1=[v_0,v_1,v_2]$, $F_2=[v_2,v_3]$, $F_3=[v_3,v_4]$ and
$F_4=[v_2,v_3,v_4]$, which are exactly the chosen simplices.  Similarly it is
\[
 \mu_2(a,b,c,d)=u_5(a,b,c,d)=[v_0,v_1,v_2,v_3,v_4].
\]
\end{example}

\begin{example}[The Klein four-group]
\label{ex:klein-four-operations}
Let $V_4=\langle \sigma,\tau\mid \sigma^2=\tau^2=1,\ \sigma\tau=\tau\sigma\rangle$ act on itself by left
multiplication.  Enumerating its elements as $1,\sigma,\tau,\sigma\tau$ embeds it into
$\Sigma_4$ as the regular representation,
\[
 \sigma\longmapsto(1\,2)(3\,4),
 \qquad
 \tau\longmapsto(1\,3)(2\,4),
 \qquad
 \sigma\tau\longmapsto(1\,4)(2\,3).
\]
Since $\Z[V_4]=\Z[C_2]\otimes\Z[C_2]$, the tensor product
$W=W(2)\otimes W(2)$ of two copies of the periodic resolution, with $\sigma$
acting on the first and $\tau$ on the second factor, is a free resolution of
$V_4$.  In degree $i$ it has the $i+1$ basis elements
$e_{(k,l)}=e_k\otimes e_l$ with $k+l=i$, that is $e_{i,l}$ in the notation of
\Cref{def:permutation-group-resolution}.  In low degrees the
differential is
\[
 de_{(1,0)}=(\sigma-1)e_{(0,0)},
 \qquad
 de_{(0,1)}=(\tau-1)e_{(0,0)},
 \qquad
 de_{(1,1)}=(\sigma-1)e_{(0,1)}-(\tau-1)e_{(1,0)}.
\]
Write $\mu_{(k,l)}=\mu(e_{(k,l)})$.  One possible choice of $\psi$ in
these degrees is
\begin{align*}
 \mu_{(0,0)}={}&(1,2,3,4),\\
 \mu_{(1,0)}={}&(1,2,1,4,3)+(1,2,3,4,3) = (132)\circ_1 (121) + (123) \circ_3 (121),\\
 \mu_{(0,1)}={}&(1,3,4,1,2)+(1,2,3,4,2) = ((121)\circ_1(12)) \circ_3 (12),\\
 \mu_{(1,1)}={}&(1,2,1,4,3,1)-(1,3,4,1,2,1)-(1,3,4,3,2,1)\\
 &-(1,2,3,4,3,2)+(1,2,4,3,2,1),
\end{align*}
and one checks with the rule for $d$ that
$d\mu_{(1,0)}=(\sigma-1)\mu_{(0,0)}$, $d\mu_{(0,1)}=(\tau-1)\mu_{(0,0)}$ and
$d\mu_{(1,1)}=(\sigma-1)\mu_{(0,1)}-(\tau-1)\mu_{(1,0)}$.  Over $\F_2$ the two
first-order operations are composites of cup and cup-$1$ products,
\begin{align*}
 \mu_{(1,0)}(a,b,c,d)&=(a\cup_1b)\cup d\cup c\\
 &\quad+a\cup b\cup(c\cup_1d),\\
 \mu_{(0,1)}(a,b,c,d)&=(a\cup b)\cup_1(c\cup d).
\end{align*}
The first is a homotopy for swapping the inputs within each of the pairs
$(1,2)$ and $(3,4)$. The second is a homotopy for swapping the two pairs.

Note that for $C_r$ the coherence data form
a line $\mu_0,\mu_1,\mu_2,\ldots$, whereas for $V_4$ they form a lattice
indexed by $(k,l)$.  
The operation $\mu_{(1,1)}$ is a homotopy between the two homotopies
obtained by applying the commuting involutions $\sigma$ and $\tau$
in either order.
\end{example}

\subsection{Lifting coefficients}
\label{sec:integral-lifts}

In this section, we formalise how to lift cocycles with
$\Z/\theta$-coefficients to $R$-coefficients. In
\Cref{ex:integral-bockstein}, a cocycle $x$ was lifted to an integral
cochain $\widetilde x$, and the failure of $\widetilde x$ to be a
cocycle was recorded by the integral cochain $\beta_\theta(x)$.
Steenrod puts this pair into a single map out of a universal
complex~\cite[\S2]{Steenrod1957}.

\begin{definition}\label{def:universal-complex}
The \emph{universal complex} $M(\theta)$ is the cochain complex freely
generated by $u$ in degree $0$ and $v$ in degree $1$, with
\[
 du=\theta v,
 \qquad dv=0.
\]
For a commutative ring $R$, write
\begin{equation}\label{eq:coefficient-extended-M}
 M(\theta)_R=M(\theta)\otimes_\Z R,
\end{equation}
omitting the subscript when the coefficient ring is clear.
Tensor products of complexes of $R$-modules are understood to be over
$R$.
\end{definition}

Over $\Z$, $M(\theta)$ is the cochain complex $\Z\xrightarrow{\theta}\Z$,
whose only cohomology is $\Z/\theta$ in degree one.  Its role is explained by
the hom complex of \Cref{sec:hom-complex}.  For $x \in Z^q(X,\Z/\theta)$ define a
map $f_x \in \operatorname{hom}^q(M(\theta),C^*(X;\Z))$ of degree $q$ by
\[
 f_x\colon M(\theta)\longrightarrow C^*(X;\Z),
 \qquad
 f_x(u)=\widetilde x,
 \quad
 f_x(v)=(-1)^q\beta_\theta(x).
\]
It is a cycle in the hom complex, since
\[
 (d_{\mathrm{hom}}f_x)(u)=d\widetilde x-(-1)^q\theta f_x(v)
 =d\widetilde x-\theta\beta_\theta(x)=0,
 \qquad
 (d_{\mathrm{hom}}f_x)(v)=(-1)^q d\beta_\theta(x)=0.
\]
Thus the lift together with its Bockstein defines a degree-$q$ chain map
out of $M(\theta)$, that is an element
\[
 f_x\in Z^q\bigl(\operatorname{hom}(M(\theta),C^*(X;\Z))\bigr).
\]
Conversely, the relation $df(u)=(-1)^q\theta f(v)$ says that for every such
cycle $f$ the cochain $f(u)$ reduces to a cocycle modulo $\theta$.

\begin{proposition}\label{prop:universal-complex-represents}
The assignment $f\mapsto f(u)\bmod\theta$ induces an isomorphism
\[
 H_{-q}\bigl(\operatorname{hom}(M(\theta),C^*(X;\Z))\bigr)
 \xrightarrow{\ \cong\ }H^q(X;\Z/\theta),
\]
whose inverse sends $[x]$ to $[f_x]$.
\end{proposition}

\begin{proof}
We prove that the assignment $x \mapsto f_x$ descends to cohomology. Let $y = x + dc$. As in \Cref{ex:integral-bockstein}
it follows that $\widetilde y - \widetilde x = d\widetilde c + \theta b$ and $\beta_{\theta}(y) - \beta_{\theta}(x) = db$ 
for some integral cochain $b$. Define the homotopy
$h\in\operatorname{hom}^{q-1}(M(\theta),C^*(X;\Z))$ by
$h(u)=\widetilde c$ and $h(v)=(-1)^q b$. Indeed,
\begin{align*}
 (d_{\mathrm{hom}}h)(u)
 &=d\widetilde c+\theta b=f_y(u)-f_x(u),\\
 (d_{\mathrm{hom}}h)(v)
 &=(-1)^q db=f_y(v)-f_x(v),
\end{align*}
so $f_x$ and $f_y$ are homotopic. Conversely, if $f$ and $f'$ are homotopic,
write $f-f'=d_{\operatorname{hom}}h$. Then
\[
f(u)-f'(u)=dh(u)+(-1)^q\theta h(v)\equiv dh(u)\pmod\theta,
\]
so $f(u)$ and $f'(u)$ are cohomologous $\mod \theta$. It remains to check that the assignments are inverses of each other.
By construction,
\[
 [f_x(u)\bmod\theta]=[\widetilde x\bmod\theta]=[x].
\]
For any cycle $f$, put $x=f(u)\bmod\theta$ and write
$f(u)-\widetilde x=\theta b$. Taking $h(u)=0$ and
$h(v)=(-1)^q b$ gives $f-f_x=d_{\mathrm{hom}}h$.
Hence the assignments are inverse.
\end{proof}

For an arbitrary coefficient ring $R$, composing
with the coefficient map $C^*(X;\Z)\to C^*(X;R)$ gives the $R$-lift
$f_x\in Z^q(\operatorname{hom}(M(\theta)_R,C^*(X;R)))$, and cohomologous
cocycles still give homologous, that is homotopic, maps.

For computations it is convenient to shift the degrees so that $f_x$ becomes
a degree-zero chain map.  Let $M(\theta)[-q]$ be generated by $u$ in degree $q$ and $v$ in degree $q+1$ with
$du=\theta v$.  A degree-$q$ chain map $f$ out of $M(\theta)$ corresponds to
the degree-zero chain map out of $M(\theta)[-q]$ which sends $u$ to $f(u)$ and
$v$ to $(-1)^qf(v)$.  For $f_x$ we obtain
\begin{equation}\label{eq:steenrod-cochain-representation}
 f_x\colon M(\theta)[-q]\longrightarrow C^*(X;\Z),
 \qquad
 f_x(u)=\widetilde{x},
 \quad
 f_x(v)=\beta_\theta(x).
\end{equation}
From now on we use this form.
The construction of $f_x$ above requires $x$ to be a cocycle. To extend our
formulas to arbitrary cochains, we choose universal expressions in which every
occurrence of $v$ carries a factor of $\theta$, and evaluate by
the substitutions
\[
 u\longmapsto\widetilde x,
 \qquad \theta v\longmapsto d\widetilde x.
\]
On cocycles this agrees with $f_x$, since
$\theta\beta_\theta(x)=d\widetilde x$.
Steenrod's construction of the Pontryagin square is the basic example of
this cancellation and will serve as the model for our constructions.

\begin{example}[The Pontryagin square]\label{ex:pontryagin-square}
Let $q$ be even and put $M=M(2)[-q]$.  Give $M^{\otimes2}$ the transposition
action of $C_2=\langle\tau\rangle$.
In the total cochain complex
\[
 W(2)\otimes_{\Z[C_2]}
 \bigl(M^{\otimes2}\otimes\Z/4\bigr),
 \qquad |e_i\otimes m|=|m|-i,
\]
where $W(2)$ is the resolution in \Cref{eq:w-two} and the differential is
\[
 d(e_i\otimes m)=de_i\otimes m+(-1)^i e_i\otimes dm,
\]
consider the degree-$2q$ element
\begin{equation}\label{eq:universal-pontryagin-cocycle}
 z_{\mathfrak P}
 =e_0\otimes(u\otimes u)+2e_1\otimes(u\otimes v).
\end{equation}
This is a cocycle.  Indeed, $dz_{\mathfrak P}=e_0\otimes\bigl(d(u\otimes u)
+(\tau-1)(2u\otimes v)\bigr)-2e_1\otimes d(u\otimes v)$, and since $q$ is even,
\begin{align*}
 d(u\otimes u)&=2(v\otimes u+u\otimes v),\\
 (\tau-1)(2u\otimes v)&=2(v\otimes u-u\otimes v),\\
 d(2u\otimes v)&=4v\otimes v.
\end{align*}
Thus the first two lines add to $4v\otimes u=0$, while the last line is
also zero modulo $4$.  This is Steenrod's universal cocycle for the
Pontryagin square~\cite[\S2]{Steenrod1957}.

Now let $x\in Z^q(X;\Z/2)$.  The representing chain map
$f_x\colon M\to C^*(X;\Z)$ induces an \emph{evaluation map} $\operatorname{ev}_x$,
given as the composite
\begin{equation}\label{eq:pontryagin-evaluation-triangle}
\begin{tikzcd}[row sep=3.2em,column sep=2.2em]
 W(2)\otimes_{\Z[C_2]}\bigl(M^{\otimes2}\otimes\Z/4\bigr)
   \arrow[r,"\mathrm{id}\otimes f_x^{\otimes2}"]
   \arrow[dr,"\operatorname{ev}_x"']
 & W(2)\otimes_{\Z[C_2]}\bigl(C^*(X;\Z)^{\otimes2}\otimes\Z/4\bigr)
   \arrow[d,"\mu^{(2)}"]\\
 & C^*(X;\Z/4)
\end{tikzcd}
\end{equation}
The horizontal map is defined on the balanced tensor product because
$f_x^{\otimes2}$ is $C_2$-equivariant, and the vertical map is Steenrod's
map in \Cref{eq:steenrods-map}, which sends $e_i\otimes a\otimes b$ to the
cup-$(2,i)$ product $\mu^{(2)}_i(a,b)$ of
\Cref{ex:cup-r-i-products}.  Both maps are chain maps, so
$\operatorname{ev}_x$ is a chain map.  Evaluating $z_{\mathfrak P}$ is a
term-by-term reading of~\Cref{eq:universal-pontryagin-cocycle}.  With
\[
 a=\widetilde x=\iota_2x,
 \qquad b=\beta_2(x)=\frac{da}{2},
\]
and $\mu_0^{(2)}=(1,2)=\cup$, $\mu_1^{(2)}=(1,2,1)=\pm\cup_1$, we get
\begin{align*}
 \operatorname{ev}_x(z_{\mathfrak P})
 &=\mu_0^{(2)}(a,a)+2\mu_1^{(2)}(a,b)\\
 &=a\cup a\pm 2a\cup_1 b\\
 &\equiv a\cup a+a\cup_1 da\pmod4.
\end{align*}
where we used that the sign does not matter for even numbers $\mod 4$ and that $2b = da$.
We therefore recover the cochain
formula in \Cref{eq:pontryagin-square-cochain-formula} for $\mathfrak P(x)$.

The machinery now does the rest.  Since $z_{\mathfrak P}$ is a cocycle and
$\operatorname{ev}_x$ is a chain map, $\operatorname{ev}_x(z_{\mathfrak P})$
is a cocycle modulo $4$, as we already checked by hand in
\Cref{ex:pontryagin-square-formula}.  If $x$ and $y$ are cohomologous,
then $f_x$ and $f_y$ are homotopic, hence so are
$\mathrm{id}\otimes f_x^{\otimes2}$ and $\mathrm{id}\otimes f_y^{\otimes2}$
on the balanced tensor product~\cite[\S1]{May1970}, and
$\operatorname{ev}_x(z_{\mathfrak P})$ and $\operatorname{ev}_y(z_{\mathfrak P})$
differ by a coboundary.  Replacing $z_{\mathfrak P}$ by a cohomologous
cocycle changes the output by a coboundary as well.  Hence the cochain
formula in \Cref{eq:pontryagin-square-cochain-formula} defines a cohomology
operation
\[
 \mathfrak P\colon H^q(X;\Z/2)\longrightarrow H^{2q}(X;\Z/4).
\]
It reduces modulo $2$ to $x\cup x$ and satisfies
\[
 \mathfrak P(x+y)
 =\mathfrak P(x)+\mathfrak P(y)+2(x\cup y),
\]
where the last term uses the inclusion $\Z/2\xrightarrow{\,2\,}\Z/4$.
These properties identify it as the usual Pontryagin square.
\end{example}

The Pontryagin square is the arity-two, untwisted instance of a general
construction. In \Cref{sec:constructing-twisted-cohomology-operations},
we also let $G$ act on the space $X$ and evaluate the $r$ inputs along
a $G$-orbit. We close this section by listing the data and steps of
the construction.
\begin{enumerate}
 \item the input degree $q$, the input coefficient ring $\Z/\theta$, and the output coefficient ring $R$. For the Pontryagin square, we had $\theta=2$ and $R=\Z/4$;
 \item a space $X$ with an action of a group $G$;
 \item a permutation representation $G\to\Sigma_r$;
 \item a free $\Z[G]$-resolution $W$ of $G$ with basis $\{e_{i,j}\}$ and boundary
 coefficients $\gamma^{i,j}_{j'}$ (\Cref{def:permutation-group-resolution});
 \item generalized cup-$i$ products, that is a factorisation
 $$W\to\cX(r)\to\operatorname{End}_{C^*(X;R)}(r)$$ through the surjection
 operad, computed by the recursion in \Cref{eq:general-comparison-recursion}; each basis element $e_{i,j}$ gives an $r$-input cochain operation $\mu^G_{i,j}$ lowering total degree by $i$;
 \item a universal twisted cocycle in
 $W\otimes_{\Z[G]}(R_\omega\otimes M(\theta)[-q]^{\otimes r})$, where the
 orientation character $\omega$ accounts for the action of $G$ on the
 fundamental class of $X$, $g_*[X] = \omega(g)[X]$;
 \item evaluation of that element through the generalized cup-$i$ products
 along a $G$-orbit, and a division-free representative which turns it into a
 cochain formula that descends to cohomology on cocycles.
\end{enumerate}

\newpage

\section{Twisted Cohomology Operations}
\label{sec:constructing-twisted-cohomology-operations}

Ordinary natural cohomology operations
\[
 \Phi\colon H^1(-;\Z/\theta)\longrightarrow H^m(-;R),
 \qquad m>1,
\]
vanish on classes that lift integrally. Indeed, such a class has the form
$a=f^*\bar u$ for a map $f\colon X\to S^1$, where $\bar u$ is the reduction
of the integral generator of $H^1(S^1;\Z)$. Naturality gives
\[
 \Phi_X(a)=f^*\Phi_{S^1}(\bar u)=0,
\]
since $H^m(S^1;R)=0$. In particular, these operations vanish on all
degree-one classes of tori. We instead evaluate the generalized
cup-$i$ products on the images of one cochain under a finite group
of symmetries of $X$.

We construct a threefold $C_3$ operation lifting coefficients from
$\Z/3$ to $\Z/9$, and fourfold $C_4$ and $V_4$ operations lifting
coefficients from $\Z/2$ to $\Z/8$. Their applications to qutrit
and qubit $T$ gates appear in
\Cref{sec:non-clifford-gates-via-twisted-cohomology-operations}.

Throughout the chapter we fix the following setup. $X$ is a closed connected oriented $n$-dimensional triangulated
 manifold with fundamental cycle $[X]$.
$G = \{g_1=1,g_2,\ldots,g_r\}$ is a finite abelian group acting on $X$ by simplicial
 automorphisms preserving the local vertex orders. This ensures strict
 naturality of the cochain operations. The assumption can be arranged by
 barycentric subdivision, ordering barycentres within each simplex by the
 dimensions of their corresponding faces.
 Let $\pi\colon G\to\Sigma_r$ be the regular representation, given by
 left multiplication on the elements of $G$:
 \[
  g\,g_k=g_{\pi(g)(k)}.
 \]
$W$ is a free resolution of $G\leq\Sigma_r$ with basis $\{e_{i,j}\}$
 as in \Cref{def:permutation-group-resolution}, and
 $\mu^G_{i,j}$ are the generalized cup-$i$ products of
 \Cref{def:generalized-cup-i}. Naturality reads as
 \[
 \mu_{i,j}^G(ga_1,\ldots, ga_r) = g \mu_{i,j}^G(a_1,\ldots,a_r).
 \]
 Furthermore, $\theta\geq2$ and $q\geq0$ are integers, $R$ is a quotient ring of
 $\Z$, for example $\Z/8$, and $M=M(\theta)[-q] \otimes_{\Z}R$ is the
 shifted universal complex of \Cref{sec:integral-lifts}.

\subsection{Universal twisted operations}
\label{sec:universal-twisted-operations}

\begin{definition}
Integration is evaluation on the fundamental cycle,
\[
 \int_X a=a([X]),\qquad a\in C^n(X;R),
\]
extended by zero in other degrees. Stokes' theorem gives $\int_X db=0$
for $b\in C^{n-1}(X;R)$, so integration descends to cohomology.
\end{definition}
The action on the fundamental
class is described by the \emph{orientation character}
\[
 \omega\colon G\longrightarrow\{\pm1\},
 \qquad
 g_*[X]=\omega(g)[X],
\]
or equivalently
\begin{equation}\label{eq:orientation-character-integration}
 \int_X g\cdot a=\omega(g)\int_Xa.
\end{equation}
For a commutative ring $R$, let $R_\omega$ be the free rank-one
$R$-module with $g\cdot1=\omega(g)1$.

\begin{remark}
More generally, the construction works with evaluation on any $n$-cycle
$c$ satisfying $g_*c=\omega(g)c$ as an equality of chains.
\end{remark}

The universal construction below is independent of $X$. Its classes
can be evaluated on any space satisfying the setup above whose
$G$-action has the prescribed orientation character $\omega$.

\begin{definition}[Universal twisted element]
\label{def:universal-twisted-element}
Let $G$ be a finite abelian group, $\pi\colon G\to\Sigma_r$ the regular
representation, $\omega\colon G\to\{\pm1\}$ a character and $R$ a
commutative ring, and fix $\theta$ and $q$.  The group $G$ acts on
$R_\omega\otimes_R M^{\otimes r}$ by
\begin{equation}\label{eq:graded-cyclic-tensor-action}
 g\,(1\otimes y_1\otimes\cdots\otimes y_r)
 =\omega(g)\,\epsilon\;
 1\otimes y_{\pi(g)(1)}\otimes\cdots\otimes y_{\pi(g)(r)},
\end{equation}
where $\epsilon=\pm1$ is the Koszul sign of the permutation of the graded
elements $y_1,\ldots,y_r$.  Define the \emph{universal twisted homology
group}
\begin{equation}\label{eq:universal-twisted-homology-group}
 H_n\!\left(
 G;R_\omega\otimes_R M^{\otimes r}
 \right)
 :=H^n\!\left(
 W\otimes_{\Z[G]}
 \bigl(R_\omega\otimes_R M^{\otimes r}\bigr)
 \right),
\end{equation}
where the balanced tensor product identifies $ge\otimes m$ with
$e\otimes gm$, and the grading and differential are
\[
 |e_{i,j}\otimes m|=|m|-i,
 \qquad
 d(e_{i,j}\otimes m)=de_{i,j}\otimes m+(-1)^i e_{i,j}\otimes dm.
\]
A \emph{universal twisted element} of type $(G,\pi,\omega,n,r)$ is a class $[z] \in H_n(G,R_{\omega}\otimes_R M^{\otimes r})$. 
\end{definition}

Since $W$ is free, $W\otimes_{\Z[G]}N=\bigoplus_{i,j}e_{i,j}\otimes N$ for
every $G$-module $N$, so every element can be written uniquely as
$z=\sum_{i,j}e_{i,j}\otimes m_{i,j}$. Since $M^{\otimes r}$ is bounded, only
finitely many resolution degrees contribute in total degree $n$. If $W$
is finitely generated in each degree, the group can therefore be computed
by finite linear algebra. By the
comparison theorem it does not depend on the choice of $W$ up to
isomorphism.  For the cyclic group with $\pi(\rho)=(1\,2\,\ldots\,r)$, the
generator acts by
\[
 \rho\,(1\otimes y_1\otimes\cdots\otimes y_r)
 =\omega(\rho)(-1)^{|y_1|(|y_2|+\cdots+|y_r|)}
 1\otimes y_2\otimes\cdots\otimes y_r\otimes y_1,
\]
which is the action appearing in \Cref{ex:pontryagin-square} for
$r=2$ and trivial $\omega$. We construct explicit cocycles in
\Cref{sec:folded-ccz,sec:universal-c4-cocycle,sec:universal-v4-cocycle}.

To evaluate these classes, define the $R$-linear map
\[
  E\colon W\otimes_{\Z[G]}\bigl(R_\omega\otimes_R C^*(X;R)^{\otimes r}\bigr) \to R
\]
by 
\begin{equation}\label{eq:twisted-multiplication}
  E\bigl(e_{i,j}\otimes(1\otimes a_1\otimes\cdots\otimes a_r)\bigr) = \int_X \mu^G_{i,j}(g_1a_1,g_2a_2,\ldots,g_ra_r)
\end{equation}
for cochains $a_1,\ldots,a_r\in C^*(X;R)$. The integral vanishes unless
the input has total degree $n$.
Before integration, the twisted cochain expression need not
respect the balancing relation. The following lemma shows
that its integral does. Write $D_G=g_1\otimes\cdots\otimes g_r$.

\begin{lemma}\label{lem:evaluation-balanced}
Let $R_{\omega} \otimes C^*(X;R)^{\otimes r}$ carry the action in \Cref{eq:graded-cyclic-tensor-action}
with $M$ replaced by $C^*(X;R)$. Then
\Cref{eq:twisted-multiplication} defines a map on the balanced
tensor product which vanishes on boundaries.
\end{lemma}

\begin{proof}
Write $A=(a_1,\ldots,a_r)$. By equivariance of $\mu$, the identity
$gg_k=g_{\pi(g)(k)}$, naturality, and
\Cref{eq:orientation-character-integration}, we obtain
\begin{align*}
  E(ge_{i,j}\otimes A)
 &=\int_X\bigl(\pi(g)\mu^G_{i,j}\bigr)(g_1a_1,\ldots,g_ra_r)\\
 &=\epsilon\int_X\mu^G_{i,j}\bigl(g_{\pi(g)(1)}a_{\pi(g)(1)},\ldots,g_{\pi(g)(r)}a_{\pi(g)(r)}\bigr)\\
 &=\epsilon \int_Xg\,\mu^G_{i,j}\bigl(g_1a_{\pi(g)(1)},\ldots,g_ra_{\pi(g)(r)}\bigr)\\
 &=\omega(g)\epsilon E\bigl(e_{i,j}\otimes(a_{\pi(g)(1)},\ldots,a_{\pi(g)(r)})\bigr) \\
 &=E(e_{i,j}\otimes gA)
\end{align*}
So the orientation character produced by the integral is exactly canceled by the factor
$R_{\omega}$ and $E$ is balanced. It remains to show that boundaries are mapped to $0$.
We only have to evaluate on boundaries of $e_{i,j}\otimes A$ with degree $n-1$; all other cases are trivial as the 
integral maps them to zero.
Since each $g_k$ commutes with $d$, so does $D_G$. Since $\mu$ is a chain map,
\Cref{eq:endomorphism-differential} gives
\[
 d\bigl(\mu^G_{i,j}(D_GA)\bigr)
 =\mu(de_{i,j})(D_GA)+(-1)^i\mu^G_{i,j}(D_G\,dA).
\]
Evaluating on the boundary thus yields
\[
 E\bigl(d(e_{i,j}\otimes A)\bigr)
 =\int_X\mu(de_{i,j})(D_GA)+(-1)^i\int_X\mu^G_{i,j}(D_G\,dA)
 =\int_Xd\bigl(\mu^G_{i,j}(D_GA)\bigr)=0
\]
by Stokes' theorem. Thus $E$ vanishes on boundaries.

\end{proof}

\begin{definition}[Evaluation map]\label{def:evaluation-map}
Let $x\in Z^q(X;\Z/\theta)$ and let $f_x\colon M\to C^*(X;R)$ be its $R$-lift in
\Cref{eq:steenrod-cochain-representation}.  Its tensor power
$f_x^{\otimes r}$ is $G$-equivariant, and the \emph{evaluation map} is the
composite
\[
\begin{tikzcd}[column sep=3.4em]
 W\otimes_{\Z[G]}\bigl(R_\omega\otimes M^{\otimes r}\bigr)
   \arrow[r,"\mathrm{id}\otimes f_x^{\otimes r}"]
   \arrow[rr,bend right=18,"\operatorname{ev}_x"']
 & W\otimes_{\Z[G]}\bigl(R_\omega\otimes C^*(X;R)^{\otimes r}\bigr)
   \arrow[r,"E"]
 & R.
\end{tikzcd}
\]
Explicitly,
\begin{equation}\label{eq:twisted-cyclic-evaluation}
 \operatorname{ev}_x
 \bigl(e_{i,j}\otimes(1\otimes y_1\otimes\cdots\otimes y_r)\bigr)
 =\int_X\mu^G_{i,j}\bigl(
 g_1f_x(y_1),g_2f_x(y_2),\ldots,
 g_rf_x(y_r)\bigr).
\end{equation}
\end{definition}

\begin{proposition}\label{prop:twisted-cyclic-evaluation}
Let $z$ be a cocycle of total degree $n$ in
$W\otimes_{\Z[G]}\bigl(R_\omega\otimes M^{\otimes r}\bigr)$, and let
$x\in Z^q(X;\Z/\theta)$. Then the value $\operatorname{ev}_x(z)$ depends
only on $[z]$ and $[x]$. It therefore defines
\begin{equation}\label{eq:integrated-twisted-operation}
 \begin{aligned}
 H_n\!\left(
 G;R_\omega\otimes_R M^{\otimes r}
 \right)\times H^q(X;\Z/\theta)&\longrightarrow R,\\
 ([z],[x])&\longmapsto\operatorname{ev}_x(z).
 \end{aligned}
\end{equation}
\end{proposition}

\begin{proof}
By \Cref{lem:evaluation-balanced}, $E$ vanishes on boundaries.
Since $\mathrm{id}\otimes f_x^{\otimes r}$ is a chain map,
$\operatorname{ev}_x$ also vanishes on boundaries. Its value on $z$
therefore depends only on $[z]$. Finally, if $x$ and
$y$ are cohomologous, then $f_x$ and $f_y$ are homotopic by
\Cref{prop:universal-complex-represents}.  Since $W$ is
$\Z[G]$-free, the maps $\mathrm{id}\otimes f_x^{\otimes r}$ and
$\mathrm{id}\otimes f_y^{\otimes r}$ are then chain homotopic on the balanced
tensor product~\cite[\S1]{May1970}, say by $H$, and
$\operatorname{ev}_x(z)-\operatorname{ev}_y(z)=E(dHz+Hdz)=E(dHz)=0$ for a
cocycle $z$.
\end{proof}

For fixed $[z]$, we write
\[
 \Phi_{[z]}\colon H^q(X;\Z/\theta)\longrightarrow R,
 \qquad
 \Phi_{[z]}([x])=\operatorname{ev}_x(z)
\]
for the cohomology invariant defined by $[z]$.

\subsection{Lifts to cochain operations}
\label{sec:lifts-to-cochain-operations}

The operation $\Phi_{[z]}$ is defined on cohomology classes.
A physical circuit needs a formula on every cochain. We obtain
one by choosing a universal representative in which every
occurrence of $v$ is accompanied by a factor of $\theta$, so that
$\theta v$ can be replaced by $d\iota x$.

Throughout this subsection, assume $n=rq$ and that reduction
modulo $\theta$ defines a map $R\to\Z/\theta$, as for
$R=\Z/\theta^k$.

In total degree $rq$, the component $m_{i,j}$ of an element
$z=\sum e_{i,j}\otimes m_{i,j}$ has degree $rq+i$, so every tensor monomial in
$m_{i,j}$ contains exactly $i$ copies of $v$. Thus for a cochain lift we
require $m_{i,j}$ to be divisible by $\theta^i$. Let $\Theta$ multiply
the component of resolution degree $i$ by $\theta^i$.

\begin{definition}\label{def:cochain-divisible}
A cycle $z$ is \emph{cochain-divisible} if it admits a factorization
\[
 z=\Theta\bar z
   =\sum_{i,j}e_{i,j}\otimes\theta^i\bar m_{i,j}
 \in W\otimes_{\Z[G]}\bigl(R_\omega\otimes M^{\otimes r}\bigr).
\]
A universal class is \emph{liftable} if it has a cochain-divisible
representative. We retain the choice of
$\bar z=\sum_{i,j}e_{i,j}\otimes\bar m_{i,j}$ as part of the data;
$\bar z$ itself need not be a cycle.
We write
\[
 H_{n}^{\mathrm{coch}}\!\left(
 G;R_\omega\otimes_R M^{\otimes r}
 \right)
\]
for the submodule of liftable classes.
\end{definition}

For $a\in C^q(X;R)$, let $F_a$ be the graded $R$-linear map
on the universal complex obtained by substituting
\[
 u\longmapsto a,\qquad v\longmapsto da
\]
in every tensor factor, leaving the resolution and orientation
factors unchanged. It is not generally a chain map.

\begin{definition}\label{def:lifted-evaluation-map}
Given a cycle $z=\Theta\bar z$ and a choice of $\bar z$, define
\[
 \widehat\Phi_{\bar z}\colon C^q(X;R)\longrightarrow R,
 \qquad
 \widehat\Phi_{\bar z}(a)=E(F_a(\bar z)).
\]
For $x\in C^q(X;\Z/\theta)$, the physical phase is
$\widehat\Phi_{\bar z}(\iota x)$, where $\iota$ is the
coefficientwise $\{0,\ldots,\theta-1\}$-lift to $R$.
\end{definition}

The following theorem identifies its action on cocycles.

\begin{theorem}\label{thm:cochain-lift}
Let $z=\Theta\bar z$ be a cochain-divisible representative of $[z]$.
\begin{enumerate}
 \item If the reduction of $a\in C^q(X;R)$ modulo $\theta$ is a cocycle
 $x\in Z^q(X;\Z/\theta)$, then
 \[
  \widehat\Phi_{\bar z}(a)=\Phi_{[z]}([x]).
 \]
 In particular the value depends only on the class $[x]$, and any cochain
 reducing to $x$ may be used in place of $\iota x$.
 \item Consequently
 \[
  \widehat\Phi_{\bar z}(\iota(x+d\lambda))=\widehat\Phi_{\bar z}(\iota x)
 \]
 for every $x\in Z^q(X;\Z/\theta)$ and $\lambda\in C^{q-1}(X;\Z/\theta)$, and
 $x\mapsto\widehat\Phi_{\bar z}(\iota x)$ induces the cohomology operation
 $\Phi_{[z]}\colon H^q(X;\Z/\theta)\to R$.
\end{enumerate}
\end{theorem}

\begin{proof}
We only prove (1), since (2) follows by applying (1) to $\iota(x+d\lambda)$ and
$\iota x$, which reduce to the cohomologous cocycles $x+d\lambda$ and $x$.

Let $f_x$ be the $R$-lift in
\Cref{eq:steenrod-cochain-representation} of $x$.  Both $a$ and $\widetilde x$
reduce to $x$ modulo $\theta$, so we can write $a-\widetilde x=\theta c$ with
$c\in C^q(X;R)$.  Define
\[
 f_a\colon M(\theta)[-q]\longrightarrow C^*(X;R),
 \qquad
 f_a(u)=a,
 \quad
 f_a(v)=\beta_\theta(x)+dc.
\]
Then $h: u \mapsto 0, v\mapsto c$ is a homotopy between $f_x$ and $f_a$, so $f_a=f_x+d_{\hom}h$.  
In particular 
\[
 \theta f_a(v)=d\widetilde x+\theta\,dc=da.
\]

Write $z=\Theta\bar z=\sum_{i,j}e_{i,j}\otimes\theta^i\bar m_{i,j}$. Every monomial in $\bar m_{i,j}$ contains exactly
$i$ copies of $v$ and $F_a$ replaces each of them by $da = \theta f_a(v)$ and each $u$ by $a=f_a(u)$. 
Hence
\[
 F_a(\bar z)
 =\sum_{i,j}e_{i,j}\otimes\theta^i f_a^{\otimes r}(\bar m_{i,j})
 =(\mathrm{id}\otimes f_a^{\otimes r})(z).
\]
Thus $\widehat\Phi_{\bar z}(a)=E\bigl((\mathrm{id}\otimes f_a^{\otimes r})(z)\bigr)$
is the evaluation of $z$ computed with the lift $f_a$ in place of $f_x$.  By
the same argument as in the proof of
\Cref{prop:twisted-cyclic-evaluation}, the maps
$\operatorname{id}\otimes f_a^{\otimes r}$ and
$\operatorname{id}\otimes f_x^{\otimes r}$ are homotopic, and consequently
\[
\widehat \Phi_{\bar z}(a) = E((\operatorname{id}\otimes f_x^{\otimes r})(z)) = \Phi_{[z]}([x])
\]
which proves the claim.
\end{proof}

Cohomologous universal cocycles define the same cohomology
operation and hence the same logical gate. Their cochain lifts
may nevertheless differ on noncocycles, so the resulting physical
gates can act differently outside the code space. The circuit
size and depth also depend on the chosen representative.

\subsection{Evaluation on liftable classes}
\label{sec:action-on-liftable-classes}

The correction terms disappear when the input class lifts to the target
coefficient group.  What remains is the twisted cup power along the orbit,
which is the form in which the operations are evaluated in
\Cref{sec:non-clifford-gates-via-twisted-cohomology-operations}.

\begin{proposition}\label{cor:liftable-classes}
Let $R=\Z/\theta^k$. Suppose $[z]$ has type $(G,\pi,\omega,rq,r)$
and a cochain-divisible representative with
$e_0$-component $e_0\otimes c_0\,u^{\otimes r}$.
If $[x]\in H^q(X;\Z/\theta)$ has a cocycle lift $c\in Z^q(X;R)$, then
\begin{equation}\label{eq:liftable-action}
 \Phi_{[z]}([x])
 =c_0\int_X(g_1c)\cup(g_2c)\cup\cdots\cup(g_rc).
\end{equation}
\end{proposition}

\begin{proof}
By \Cref{thm:cochain-lift}, we may evaluate the
cochain formula using the cocycle lift $c$. Every term
of positive resolution degree contains a coboundary
of an input and therefore vanishes, since $dc=0$.
The remaining term is the stated cup product.
\end{proof}

This formula determines the logical action. The physical circuit still
uses the full cochain formula, including its correction terms.

\subsection{The universal $C_3$ cocycle over $\Z/9$}
\label{sec:folded-ccz}

Take $n=3$, $R=\Z/9$, $M=M(3)[-1]\otimes R$, and
$G=C_3=\langle\rho\rangle$, with trivial orientation character.
We use the cyclic resolution $W(3)$ and the tensor action in
\Cref{eq:graded-cyclic-tensor-action}. We abbreviate tensor monomials as
words, for example $uvu=u\otimes v\otimes u$. The degree-three element
\begin{equation}\label{eq:universal-c3-cocycle}
 z=e_0\otimes uuu+3e_1\otimes(2uvu+2vuu)
\end{equation}
is a cochain-divisible cocycle. Indeed, writing
$z=e_0\otimes m_0+e_1\otimes m_1$, direct substitution gives
$dm_0+(\rho-1)m_1=0$ and $dm_1=0$.

Choose $\bar z=e_0\otimes uuu+2e_1\otimes(uvu+vuu)$,
so that $z=\Theta\bar z$. For the coefficientwise
$\{0,1,2\}$-lift $a=\iota x$ of $x\in C^1(X;\Z/3)$,
evaluation gives
\begin{equation}\label{eq:twisted-c3-cochain-formula}
\begin{aligned}
 \widehat\Phi_{\bar z}(a)={}&
 \int_X\bigl[
 \mu_0(a,\rho a,\rho^2a)
 +2\mu_1(a,d(\rho a),\rho^2a)\\
 &\qquad
 +2\mu_1(da,\rho a,\rho^2a)
 \bigr]\quad\in\Z/9.
\end{aligned}
\end{equation}
Here $\mu_0=(1,2,3)$ and $\mu_1=(1,2,3,1)$; in the
conventions of \Cref{sec:cohomology-operations},
$\mu_0(a,b,c)=a\cup b\cup c$ and
$\mu_1(a,b,c)=-a\cup_1(b\cup c)$.
By \Cref{thm:cochain-lift}, this induces
$\Phi_{[z]}\colon H^1(X;\Z/3)\to\Z/9$.
Its application to the qutrit $T$ gate appears in
\Cref{sec:qutrits-and-the-three-torus}.

\subsection{The universal \texorpdfstring{$C_4$}{C4} cocycle}
\label{sec:universal-c4-cocycle}

Take $n=4$, $R=\Z/8$, $M=M(2)[-1]\otimes R$, and
$G=C_4=\langle\rho\rangle$, with $\omega(\rho)=-1$.
We use the cyclic resolution $W(4)$ and the tensor action in
\Cref{eq:graded-cyclic-tensor-action}.

\begin{proposition}\label{prop:universal-c4-cocycle}
The following degree-four element is a cochain-divisible cocycle:
\begin{equation}\label{eq:universal-c4-cocycle}
 z=e_0\otimes uuuu
   +2e_1\otimes(uvuu+2uuvu+3uuuv)
   +4e_2\otimes uuvv.
\end{equation}
\end{proposition}

\begin{proof}
Writing $z=\sum_i e_i\otimes m_i$ and
$N=1+\rho+\rho^2+\rho^3$, the cocycle condition is
\begin{equation}\label{eq:universal-c4-cocycle-identities}
 dm_0+(\rho-1)m_1=0,\qquad
 dm_1=Nm_2,\qquad dm_2=0.
\end{equation}
These identities follow by direct substitution.
The displayed factors $2^i$ give the required divisibility.
\end{proof}

Choose $\bar z$ by removing the displayed factor $2^i$ from each component
of resolution degree $i$. For the coefficientwise
$0/1$-lift $a=\iota x$ of $x\in C^1(X;\Z/2)$,
evaluation gives
\begin{equation}\label{eq:twisted-pontryagin-cochain-formula}
\begin{aligned}
 \widehat\Phi_{\bar z}(a)={}&
 \int_X\bigl[
 \mu_0(a,\rho a,\rho^2a,\rho^3a)
 +\mu_1(a,d(\rho a),\rho^2a,\rho^3a)\\
 &\qquad
 +2\mu_1(a,\rho a,d(\rho^2a),\rho^3a)
 +3\mu_1(a,\rho a,\rho^2a,d(\rho^3a))\\
 &\qquad
 +\mu_2(a,\rho a,d(\rho^2a),d(\rho^3a))
 \bigr]\quad\in\Z/8.
\end{aligned}
\end{equation}
Here $\mu_i=\mu_i^{(4)}$ are the operations of
\Cref{ex:arity-four-operations}.
By \Cref{thm:cochain-lift}, this induces
$\mathfrak P^{(4)}_\rho:=\Phi_{[z]}\colon
H^1(X;\Z/2)\to\Z/8$.
For a cocycle lift $c\in Z^1(X;\Z/8)$ of $x$, \Cref{cor:liftable-classes}
gives
\[
 \mathfrak P^{(4)}_\rho([x])
 =\int_X c\cup\rho c\cup\rho^2c\cup\rho^3c.
\]
It gives the four-torus $T$ gate of \Cref{sec:the-four-torus}.

\subsection{The universal \texorpdfstring{$V_4$}{V4} cocycle}
\label{sec:universal-v4-cocycle}

Keep $n=4$, $R=\Z/8$ and $M=M(2)[-1]\otimes R$, but take
$G=V_4=\langle\sigma,\tau\rangle$ with trivial orientation
character. Order its elements as $(1,\sigma,\tau,\sigma\tau)$
and use the resolution $W=W(2)\otimes W(2)$ and operations
of \Cref{ex:klein-four-operations}.
The tensor action is \Cref{eq:graded-cyclic-tensor-action}.

\begin{proposition}\label{prop:universal-v4-cocycle}
The following degree-four element is a cochain-divisible cocycle:
\begin{equation}\label{eq:universal-v4-cocycle}
\begin{aligned}
 z={}&e_{(0,0)}\otimes uuuu
 +2\bigl(
 2e_{(1,0)}\otimes vuuu
 +e_{(0,1)}\otimes(uvuu+uuvu)
 \bigr)\\
 &+4\bigl(
 e_{(1,1)}\otimes vuvu
 +e_{(0,2)}\otimes vvuu
 \bigr).
\end{aligned}
\end{equation}
\end{proposition}

\begin{proof}
Substitution into the tensor-product differential gives $dz=0$.
The displayed factors give divisibility by $2^{k+l}$
in resolution degree $k+l$.
\end{proof}

A compatible choice for the remaining operation is
\[
 \mu_{(0,2)}
 =\bigl((1,2,1,2)\circ_1(1,2)\bigr)\circ_3(1,2).
\]
Its boundary is $d\mu_{(0,2)}=(1+\tau)\mu_{(0,1)}$, by the binary
resolution identity in \Cref{eq:w-two}.

Choose $\bar z$ by removing the displayed factor $2^{k+l}$ from each
component of resolution degree $k+l$. For the coefficientwise
$0/1$-lift $a=\iota x$ of $x\in C^1(X;\Z/2)$,
evaluation gives
\begin{equation}\label{eq:twisted-v4-cochain-formula}
\begin{aligned}
 \widehat\Phi_{\bar z}(a)={}&
 \int_X\bigl[
 \mu_{(0,0)}(a,\sigma a,\tau a,\sigma\tau a)
 +2\mu_{(1,0)}(da,\sigma a,\tau a,\sigma\tau a)\\
 &\qquad
 +\mu_{(0,1)}(a,d(\sigma a),\tau a,\sigma\tau a)
 +\mu_{(0,1)}(a,\sigma a,d(\tau a),\sigma\tau a)\\
 &\qquad
 +\mu_{(1,1)}(da,\sigma a,d(\tau a),\sigma\tau a)
 +\mu_{(0,2)}(da,d(\sigma a),\tau a,\sigma\tau a)
 \bigr]\quad\in\Z/8.
\end{aligned}
\end{equation}
By \Cref{thm:cochain-lift}, this induces
$\mathfrak P^{(4)}_{\sigma,\tau}:=\Phi_{[z]}\colon
H^1(X;\Z/2)\to\Z/8$.
For a cocycle lift $c\in Z^1(X;\Z/8)$ of $x$, \Cref{cor:liftable-classes}
gives
\[
 \mathfrak P^{(4)}_{\sigma,\tau}([x])
 =\int_X c\cup\sigma c\cup\tau c\cup\sigma\tau c.
\]
We compare the two fourfold operations on surface squares in
\Cref{sec:squares-of-surfaces}.

\section{Addressability by cochain maps}
\label{sec:selecting-cohomology-operations}

In this Section, we show that precomposing a cohomology operation with a suitable cochain map can be used to select out which logical coordinates we would like the corresponding gate to act on.

\subsection{Precomposing an operation}
\label{sec:precomposing-operations}
\label{sec:projections-general-complexes}

Let $C,D$ be based cochain complexes over a field, and let
$\Phi\colon D^q\to R$ induce a cohomology invariant
$\overline\Phi\colon H^q(D)\to R$.

\begin{lemma}[Precomposition]\label{lem:precomposition-by-cochain-maps}
For a cochain map $F\colon C\to D$, the formula $\Phi\circ F$ induces
\begin{equation}\label{eq:precomposition-invariant}
 [a]\longmapsto\overline\Phi([Fa]).
\end{equation}
The same statement holds for several inputs, with a cochain map in each
input, and for maps which commute with $d$ while shifting degree.

Suppose each output coordinate of $F$ uses a bounded number of input
coordinates, and each input occurs in a bounded number of outputs.
If $\Phi$ is a sum of bounded-size phase terms with each coordinate in
boundedly many terms, then $\Phi\circ F$ defines a bounded-depth circuit
with unrestricted connectivity.
\end{lemma}

The circuit hypothesis says that $F$ has bounded weight in both
directions. Fixed finite sums and compositions preserve these bounds.
If $H^q(C)=V\oplus W$ and $H^q(F)$ vanishes on $W$, the gate depends
only on $V$ and acts trivially on the logical qubits belonging to $W$.
We note that the map $F$ need not be invertible or commute with the symmetries defining
$\Phi$.
If $\Phi=\widetilde\Phi\circ\iota$ uses a coefficient lift, precomposition
means $\widetilde\Phi(\iota(Fa))$, with sums taken in the input field
before lifting.

\begin{remark}
The target $D$ may be a relative cochain complex or a subcomplex imposing
coupled boundary conditions. No manifold or product structure on $C$ is
required. Over a field, every projection on cohomology lifts to a cochain
projection by splitting cycles and boundaries, but this does not ensure
bounded weight.
\end{remark}

\subsection{Example: circle folds}
\label{sec:circle-fold-projections}

Circle folds give concrete bounded-weight cochain maps with which to
precompose our operations. Choose a reflection-invariant cellulation of
$S^1$, and let $s:S^1\to S^1$ fix the upper semicircle $I_+$ and reflect
the lower semicircle onto it. Since $s^2=s$, the pullback $s^*$ and its complement $1-s^*$
are complementary cochain projections over any coefficient ring $R$.
On cohomology, $s^*$ retains $H^0(S^1;R)$ and kills $H^1(S^1;R)$,
while $1-s^*$ does the reverse.

Now, consider $X\times S^1$ and extend the map $s$ by the identity on $X$. Using the K\"unneth formula, an arbitrary class in
$H^q(X\times S^1;R)$ can be decomposed as $v+wu$, where $v\in H^q(X;R)$,
$w\in H^{q-1}(X;R)$, and $u$ generates $H^1(S^1;R)$. Since $s^*(u)=0$, we have
\begin{equation}\label{eq:circle-kunneth-projections}
 s^*(v+wu)=v,\text{ and }(1-s^*)(v+wu)=wu.
\end{equation}
Thus, precomposition selects either K\"unneth summand as the input to an
operation. Both maps have bounded weight in both directions, since each
cell has at most two preimages under the fold. Hence
\Cref{lem:precomposition-by-cochain-maps} preserves the bounded-depth
circuit property.

Moreover, $1-s^*$ takes values in the relative cochain complex
\[
 C^*(X\times S^1,X\times I_+;R),
\]
since its output vanishes on $X\times I_+$. We will use this observation
in the relative constructions of \Cref{sec:relative-t-three-torus}.

\paragraph{Choosing a direction on a torus $T^d$.} For $T^d=(S^1)^d$, let $s_i$ fold the $i$-th circle. Let $\gamma_i\in H^1(T^d;\Z)$ be the generating class dual to that circle, and write $\bar\gamma_i$ for its reduction modulo two. Then
\begin{equation}\label{eq:torus-direction-projection}
 (1+s_i^*)\left(\sum_{j=1}^d\epsilon_j\bar\gamma_j\right)=\epsilon_i\bar\gamma_i.
\end{equation}
Thus a single fold selects only the corresponding direction.

For degree-one applications it is often useful to send every selected direction to the same target class. Fold all circles except the $i$-th, then permute coordinates to obtain $F_i$ with $[F_i a]=\epsilon_i\bar\gamma_1$. The sum $\sum_{i\in J}F_i$ then selects the logical parity $\sum_{i\in J}\epsilon_i$.

To apply a relative operation on a cube to one torus direction, let $A_i\subset T^d$ be the subset with $i$-th coordinate in $I_+$, let $Q=[0,L]^d$, and let $\partial_i Q$ be its two faces $x_i=0,L$. The quotient $p\colon Q\to T^d$ is a map of pairs $(Q,\partial_i Q)\to(T^d,A_i)$. Hence
\begin{equation}\label{eq:torus-relative-cube}
 H^1(T^d;\F_2)\xrightarrow{\ 1+s_i^*\ }H^1(T^d,A_i;\F_2)
 \xrightarrow{\ p^*\ }H^1(Q,\partial_i Q;\F_2)\cong\F_2
\end{equation}
sends $\sum_j\epsilon_j\bar\gamma_j$ to $\epsilon_i$, since a coordinate segment between the two faces projects to the $i$-th coordinate circle. Identifying the pairs $(Q,\partial_i Q)$ by coordinate permutations again allows sums selecting logical parities.

\section{Intra-block gates from cup products}
\label{sec:non-clifford-fold-transversal-gates}

In~\cite{CupsAndGates}, the copy--cup construction starts with a
multilinear cohomology invariant
\[
 \int x_1\cup\cdots\cup x_r.
\]
It defines a diagonal circuit on \(r\) code blocks by
\[
 \ket{x_1,\ldots,x_r}
 \longmapsto
 (-1)^{\int x_1\cup\cdots\cup x_r}
 \ket{x_1,\ldots,x_r}.
\]
Because the integral descends to cohomology in every argument, the logical
action of this circuit depends only on the logical cohomology classes
represented by the \(x_i\), and is a product of logical $C^{r-1}Z$ gates in
the chosen logical bases. In this section, we illustrate how automorphisms,
or more generally cochain endomorphisms, can be used to produce versions of the copy--cup gate with
logical action on a single block.

Suppose that we have cochain maps \(F_i\) from a single code complex
to the \(r\) inputs of such a copy--cup gate. Substituting
\(x_i=F_i(x)\) gives the twisted single-block operation
\[
 \ket{x}
 \longmapsto
 (-1)^{\int F_1(x)\cup\cdots\cup F_r(x)}
 \ket{x}.
\]
On cocycles, this phase again depends only on the cohomology class of \(x\).
Under the bounded-weight and incidence conditions of
\Cref{lem:precomposition-by-cochain-maps}, the resulting physical circuit
has bounded depth.

We now discuss topological and non-topological examples using ordinary
cup products. \Cref{sec:non-clifford-gates-via-twisted-cohomology-operations}
uses the coefficient-lifting operations of
\Cref{sec:constructing-twisted-cohomology-operations}, while
\Cref{sec:addressable-relative-cup-gates} uses relative boundary conditions.

\subsection{$CCZ$ on the three-torus}
\label{sec:folded-three-torus}
The simplest example is the three-torus $T^3=(S^1)^3$. Let $\sigma$
cyclically permute its three circle factors. We denote the corresponding
generators by
$\gamma_1,\gamma_2,\gamma_3\in H^1(T^3;\F_2)$ and write
\[
 a=x_1\gamma_1+x_2\gamma_2+x_3\gamma_3.
\]
Then
\begin{equation}\label{eq:folded-ccz-logical-action}
 \int_{T^3}a\cup\sigma(a)\cup\sigma^2(a)
 =x_1+x_2+x_3+x_1x_2x_3.
\end{equation}
Here $\cup$ is the topological cup product and $\int$ denotes evaluation on
the fundamental class. The resulting single-block logical gate is
\[
 \overline Z_1\overline Z_2\overline Z_3\,
 \overline{\op{CCZ}}_{123}.
\]
Thus, up to Pauli gates, cyclic permutation folds the three-copy cup gate
into a single $CCZ$ gate within one three-dimensional toric-code block.
Other automorphisms give addressable Clifford gates. In
\Cref{sec:the-four-torus}, the same folding principle is combined with
twisted cohomology operations to obtain addressable $T$ gates.

\subsection{$CCZ$ gate on the Bolza surface times a circle}
\label{sec:bolza-surface}
\label{sec:circle-control-example}

First let $S$ be any closed connected oriented surface with a cellular
symmetry $g$. The invariant
\begin{equation}\label{eq:surface-twisted-cup-invariant}
 \int_S v\cup g(v)\in\F_2,\qquad v\in H^1(S;\F_2),
\end{equation}
gives a Clifford gate with phase $(-1)^{\int_S v\cup g(v)}$, a variant
of the fold-transversal construction in~\cite{foldtransversalgates}.
Its quadratic terms give $CZ$ gates and its linear terms give $Z$ gates.

Passing to $S\times S^1$, while still encoding in $H^1$, adds a circle
logical qubit. Extend $g$ trivially over the circle, let $s$ be the fold
of \Cref{sec:circle-fold-projections}, and write
$a=v+\lambda u$, where $u$ is the circle generator and $\lambda\in\F_2$.
Then
\begin{equation}\label{eq:surface-circle-control}
 \int_{S\times S^1}(1+s^*)(a)\cup s^*(a)\cup g\bigl(s^*(a)\bigr)
 =\lambda\int_S v\cup g(v).
\end{equation}
The circle qubit therefore controls the surface gate, so $CZ$ terms become
$CCZ$ terms and $Z$ terms become $CZ$ terms. The cochain formula descends
by \Cref{lem:precomposition-by-cochain-maps}, and its logical action
uses only the intersection pairing on $H^1(S;\F_2)$ and the action of $g$.

\paragraph{The Bolza surface.}
Let $\Sigma$ be the Bolza surface, a closed hyperbolic surface of genus two.
Its group of orientation-preserving automorphisms is isomorphic to
$\operatorname{GL}(2,\F_3)$.
Choose integral cohomology classes
$\widetilde\gamma_1,\ldots,\widetilde\gamma_4$ represented by the
Poincar\'e-dual curves in \Cref{fig:bolza-reflections}, and let
$\gamma_1,\ldots,\gamma_4$ denote their reductions modulo two. Their
intersection pairing is
\[
 \int_\Sigma\gamma_i\cup\gamma_j
 =
 \begin{cases}
  0,&i=j,\\
  1,&i\neq j.
 \end{cases}
\]

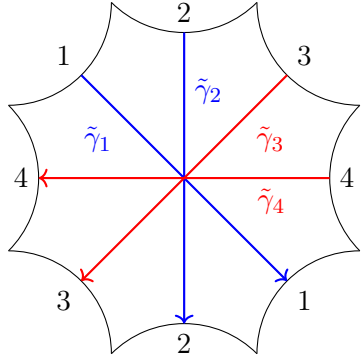
\begin{figure}
 \centering
 \begin{tikzpicture}[scale=0.12]
\draw [fill=none,draw=black] (266.4mm,266.4mm)
-- (426.8mm,266.4mm)
;
\draw [fill=none,draw=black] (266.4mm,266.4mm)
-- (379.8mm,379.8mm)
;
\draw [fill=none,draw=black] (266.4mm,266.4mm)
-- (379.8mm,153.0mm)
;
\draw [fill=none,draw=black] (266.4mm,266.4mm)
-- (266.4mm,426.8mm)
;
\draw [fill=none,draw=black] (266.4mm,266.4mm)
-- (266.4mm,106.0mm)
;
\draw [fill=none,draw=black] (266.4mm,266.4mm)
-- (153.0mm,379.8mm)
;
\draw [fill=none,draw=black] (266.4mm,266.4mm)
-- (153.0mm,153.0mm)
;
\draw [fill=none,draw=black] (266.4mm,266.4mm)
-- (106.0mm,266.4mm)
;

\draw [fill=none,draw=black] (426.8mm,266.4mm)
.. controls ++(0.0mm,14.9mm) and ++(-5.7mm,-13.8mm) .. ++(8.6mm,43.4mm)
.. controls ++(5.7mm,13.8mm) and ++(-10.5mm,-10.5mm) .. ++(24.6mm,36.8mm)
;
\draw [fill=none,draw=black] (460.0mm,186.2mm)
.. controls ++(-10.5mm,10.5mm) and ++(5.7mm,-13.8mm) .. ++(-24.6mm,36.8mm)
.. controls ++(-5.7mm,13.8mm) and ++(0.0mm,-14.9mm) .. ++(-8.6mm,43.4mm)
;
\draw [fill=none,draw=black] (379.8mm,379.8mm)
.. controls ++(-10.5mm,10.5mm) and ++(5.7mm,-13.8mm) .. ++(-24.6mm,36.8mm)
.. controls ++(-5.7mm,13.8mm) and ++(0.0mm,-14.9mm) .. ++(-8.6mm,43.4mm)
;
\draw [fill=none,draw=black] (460.0mm,346.6mm)
.. controls ++(-14.9mm,0.0mm) and ++(13.8mm,-5.7mm) .. ++(-43.4mm,8.6mm)
.. controls ++(-13.8mm,5.7mm) and ++(10.5mm,-10.5mm) .. ++(-36.8mm,24.6mm)
;

\draw [fill=none,draw=black] (346.6mm,72.8mm)
.. controls ++(0.0mm,14.9mm) and ++(-5.7mm,-13.8mm) .. ++(8.6mm,43.4mm)
.. controls ++(5.7mm,13.8mm) and ++(-10.5mm,-10.5mm) .. ++(24.6mm,36.8mm)
;
\draw [fill=none,draw=black] (379.8mm,153.0mm)
.. controls ++(10.5mm,10.5mm) and ++(-13.8mm,-5.7mm) .. ++(36.8mm,24.6mm)
.. controls ++(13.8mm,5.7mm) and ++(-14.9mm,0.0mm) .. ++(43.4mm,8.6mm)
;
\draw [fill=none,draw=black] (346.6mm,460.0mm)
.. controls ++(-10.5mm,-10.5mm) and ++(13.8mm,5.7mm) .. ++(-36.8mm,-24.6mm)
.. controls ++(-13.8mm,-5.7mm) and ++(14.9mm,0.0mm) .. ++(-43.4mm,-8.6mm)
;
\draw [fill=none,draw=black] (266.4mm,426.8mm)
.. controls ++(-14.9mm,0.0mm) and ++(13.8mm,-5.7mm) .. ++(-43.4mm,8.6mm)
.. controls ++(-13.8mm,5.7mm) and ++(10.5mm,-10.5mm) .. ++(-36.8mm,24.6mm)
;
\draw [fill=none,draw=black] (266.4mm,106.0mm)
.. controls ++(14.9mm,0.0mm) and ++(-13.8mm,5.7mm) .. ++(43.4mm,-8.6mm)
.. controls ++(13.8mm,-5.7mm) and ++(-10.5mm,10.5mm) .. ++(36.8mm,-24.6mm)
;
\draw [fill=none,draw=black] (186.2mm,72.8mm)
.. controls ++(10.5mm,10.5mm) and ++(-13.8mm,-5.7mm) .. ++(36.8mm,24.6mm)
.. controls ++(13.8mm,5.7mm) and ++(-14.9mm,0.0mm) .. ++(43.4mm,8.6mm)
;
\draw [fill=none,draw=black] (153.0mm,379.8mm)
.. controls ++(-10.5mm,-10.5mm) and ++(13.8mm,5.7mm) .. ++(-36.8mm,-24.6mm)
.. controls ++(-13.8mm,-5.7mm) and ++(14.9mm,0.0mm) .. ++(-43.4mm,-8.6mm)
;
\draw [fill=none,draw=black] (186.2mm,460.0mm)
.. controls ++(0.0mm,-14.9mm) and ++(5.7mm,13.8mm) .. ++(-8.6mm,-43.4mm)
.. controls ++(-5.7mm,-13.8mm) and ++(10.5mm,10.5mm) .. ++(-24.6mm,-36.8mm)
;
\draw [fill=none,draw=black] (72.8mm,186.2mm)
.. controls ++(14.9mm,0.0mm) and ++(-13.8mm,5.7mm) .. ++(43.4mm,-8.6mm)
.. controls ++(13.8mm,-5.7mm) and ++(-10.5mm,10.5mm) .. ++(36.8mm,-24.6mm)
;
\draw [fill=none,draw=black] (153.0mm,153.0mm)
.. controls ++(10.5mm,-10.5mm) and ++(-5.7mm,13.8mm) .. ++(24.6mm,-36.8mm)
.. controls ++(5.7mm,-13.8mm) and ++(0.0mm,14.9mm) .. ++(8.6mm,-43.4mm)
;
\draw [fill=none,draw=black] (72.8mm,346.6mm)
.. controls (83.3mm,336.1mm) and (91.7mm,323.6mm) .. (97.4mm,309.8mm)
.. controls (103.1mm,296.0mm) and (106.0mm,281.3mm) .. (106.0mm,266.4mm)
;
\draw [fill=none,draw=black] (106.0mm,266.4mm)
.. controls ++(0.0mm,-14.9mm) and ++(5.7mm,13.8mm) .. ++(-8.6mm,-43.4mm)
.. controls ++(-5.7mm,-13.8mm) and ++(10.5mm,10.5mm) .. ++(-24.6mm,-36.8mm)
;

\node[right] at (426.8mm, 266.4mm) {4};
\node[above right] at (379.8mm, 379.8mm) {3};
\node[below right] at (379.8mm, 153.0mm) {1};
\node[above] at (266.4mm, 426.8mm) {2};
\node[below] at (266.4mm, 106.0mm) {2};
\node[above left] at (153.0mm, 379.8mm) {1};
\node[below left] at (153.0mm, 153.0mm) {3};
\node[left] at (106.0mm, 266.4mm) {4};

\draw[blue,thick,->] (153.0mm, 379.8mm) -- node[pos=0.2,below left] {$\tilde{\gamma}_1$} (379.8mm, 153.0mm);
\draw[blue,thick,->] (266.4mm, 426.8mm) -- node[pos=0.2,right] {$\tilde{\gamma}_2$} (266.4mm, 106.0mm);
\draw[red,thick,->] (379.8mm, 379.8mm) -- node[pos=0.2,below right] {$\tilde{\gamma}_3$} (153.0mm, 153.0mm);
\draw[red,thick,<-] (106.0mm, 266.4mm) -- node[pos=0.8,below] {$\tilde{\gamma}_4$} (426.8mm, 266.4mm);
\end{tikzpicture}
 \caption{An octagonal fundamental domain for the Bolza surface. The solid
 oriented curves represent
 $\widetilde\gamma_1,\ldots,\widetilde\gamma_4$. Opposite sides carrying
 the same label are identified.}
 \label{fig:bolza-reflections}
\end{figure}

We first take $g=\sigma$, which acts on the integral classes by
\[
 \widetilde\gamma_1\longmapsto\widetilde\gamma_2,
 \quad
 \widetilde\gamma_2\longmapsto\widetilde\gamma_3,
 \quad
 \widetilde\gamma_3\longmapsto\widetilde\gamma_4,
 \quad
 \widetilde\gamma_4\longmapsto-\widetilde\gamma_1.
\]
Writing
\[
 v=x_1\gamma_1+x_2\gamma_2+x_3\gamma_3+x_4\gamma_4,
\]
the intersection pairing gives
\[
 \int_\Sigma v\cup g(v)
 =
 (x_1+x_3)(x_2+x_4)
 +(x_1+x_3)
 +(x_2+x_4).
\]

Using the circle fold above, for
$a=v+\lambda u\in H^1(\Sigma\times S^1;\F_2)$, we obtain
\begin{align*}
 &\int_{\Sigma\times S^1}
 (1+s^*)(a)\cup s^*(a)\cup g\bigl(s^*(a)\bigr)\\
 &\hspace{2cm}=
 \lambda\bigl((x_1+x_3)(x_2+x_4)
 +(x_1+x_3)+(x_2+x_4)\bigr).
\end{align*}
Using $x_1+x_3$ and $x_2+x_4$ as two surface logical coordinates, the
resulting single-block logical gate is
\begin{equation}\label{eq:bolza-circle-logical-gate}
 \overline{\op{CCZ}}_{\lambda,x_1+x_3,x_2+x_4}
 \overline{\op{CZ}}_{\lambda,x_1+x_3}
 \overline{\op{CZ}}_{\lambda,x_2+x_4}.
\end{equation}
Here the subscripts specify logical coordinates, with $\lambda$ the circle bit.
In particular, its non-Clifford part is a single logical $CCZ$ gate.

The second choice, $g=U$, acts by
\[
 \widetilde\gamma_1\longmapsto
 \widetilde\gamma_2+\widetilde\gamma_3,
 \quad
 \widetilde\gamma_2\longmapsto
 -\widetilde\gamma_1-\widetilde\gamma_2+\widetilde\gamma_4,
 \quad
 \widetilde\gamma_3\longmapsto
 -\widetilde\gamma_3-\widetilde\gamma_4,
 \quad
 \widetilde\gamma_4\longmapsto\widetilde\gamma_3.
\]
Up to $CZ$ factors, it gives
\[
 \overline{\op{CCZ}}_{\lambda,x_1+x_2,x_2+x_3}
 \overline{\op{CCZ}}_{\lambda,x_1+x_2+x_3,x_4}.
\]

\subsection{Beyond topological qLDPC codes: $CCZ$ gates in group algebra codes}
\label{sec:folded-nontopological-codes}
We now give two non-topological qLDPC examples. In~\cite{CupsAndGates}, cup
products were defined on oriented classical code complexes and extended to
their tensor and balanced products.

Our first examples arise as balanced products of group algebra codes over a
finite group $\mathcal{G}$. Two-fold balanced products include bivariate bicycle
codes~\cite{LinPryadko2024,EberhardtSteffan2024}, while the general $n$-fold
construction, including the three-fold complex, was introduced
in~\cite[Appendix~B]{EberhardtSteffan2024}. For
$\mathcal{G}=C_\ell\times C_m\times C_n$, these three-fold codes are now commonly
called trivariate tricycle codes. In~\cite{CupsAndGates}, such balanced
products were equipped with cup products and copy--cup gates; related
constructions were subsequently studied
in~\cite{TiewBreuckmann2026,JacobMcLauchlanBrowne2025,MenonMagicTricycles2026}.
We illustrate how to turn such a gate into an intra-block gate. Let
\[
 R=
 \F_2[x,y,z]/(x^3-1,y^3-1,z^4-1)
 \cong
 \F_2[C_3\times C_3\times C_4].
\]
The magic-tricycle code with parameters $[[108,15,(12,6)]]$ introduced
in~\cite{MenonMagicTricycles2026} is defined by the three elements
\begin{align*}
 a&=y+y^2z+xyz^3+x^2y^2z^2,\\
 b&=z^2+xy+xy^2z+x^2z^3,\\
 c&=yz^3+y^2z+x^2+x^2y^2z^2
\end{align*}
of $R$. Splitting the monomials in these polynomials into incoming and
outgoing terms according to
\begin{align*}
 (a_{\rm in},a_{\rm out})
   &=(y+xyz^3,\ y^2z+x^2y^2z^2),\\
 (b_{\rm in},b_{\rm out})
   &=(z^2+xy,\ xy^2z+x^2z^3),\\
 (c_{\rm in},c_{\rm out})
   &=(yz^3+y^2z,\ x^2+x^2y^2z^2)
\end{align*}
orients the three constituent complexes and thereby equips their balanced
product with a cup product as in~\cite{CupsAndGates}.
A search over bounded-weight $R$-linear endomorphisms of the cochain
complex produces endomorphisms with the following weight-six degree-one components
\[
 F=
 \left(\begin{smallmatrix}
  0&0&0\\
  xyz+x^2z^3&z+xyz^3&x^2y+x^2y^2z^2\\
  0&0&0
 \end{smallmatrix}\right),
 \qquad
F'=
 \left(\begin{smallmatrix}
  xz+xy^2+xy^2z+x^2&x^2z^2&x^2yz^2\\
  0&0&0\\
  0&0&0
 \end{smallmatrix}\right).
\]
For this pair, the cohomology invariant
\[
 \int\alpha\cup F(\alpha)\cup F'(\alpha)
\]
yields a product of three $CCZ$ gates, up to Clifford gates.

\subsection{$\Omega(N^{1-\varepsilon})$ $CCZ$ gates in Golowich--Tamo--Zhu codes}

The three product complexes in~\cite{GolowichTamoZhu2026} are related by
cyclic permutation of their tensor factors. This lets us fold their
addressable three-block gates into one block.

For a graph $\Gamma$ and a local code $\mathcal F$, let
$\mathcal C^*(\Gamma,\mathcal F)$ denote the associated Tanner cochain
complex. Golowich, Tamo, and Zhu~\cite{GolowichTamoZhu2026} begin with the
three Sipser--Spielman code complexes
\[
 \mathcal C=\mathcal C^*(\Gamma,\mathcal F^{(1)}),
 \qquad
 \mathcal E=\mathcal C^*(\Gamma,\mathcal F^{(2)}),
 \qquad
 \mathcal D=\mathcal C^*(\Gamma,{\mathcal F^{(3)}}^\perp)
\]
and form the three cyclically ordered product complexes
\[
 \mathcal D\otimes\mathcal C\otimes\mathcal E,
 \qquad
 \mathcal E\otimes\mathcal D\otimes\mathcal C,
 \qquad
 \mathcal C\otimes\mathcal E\otimes\mathcal D.
\]
Their construction gives addressable three-block copy--cup gates
\[
 \int_{\mathbf a}\alpha_1\cup\alpha_2\cup\alpha_3,
\]
where the weights $\mathbf a$ select the desired logical gates.

Cyclically permuting the three tensor factors identifies the three code
blocks by a bounded-weight chain isomorphism $\sigma$. Folding therefore
gives the single-block phase
\begin{equation}\label{eq:folded-gtz-phase}
 \int_{\mathbf a}
 \alpha\cup\sigma(\alpha)\cup\sigma^2(\alpha).
\end{equation}
For every $\varepsilon>0$, the subsystem qLDPC codes constructed
in~\cite{GolowichTamoZhu2026} have parameters
\[
 [[N,K,D]]_q,
 \qquad
 K\geq N^{1-\varepsilon},
 \qquad
 D\geq N^{(1-\varepsilon)/3},
\]
and constant stabilizer weight. The cyclic fold turns a constant fraction
of their addressable three-block gate slots into
$\Omega(N^{1-\varepsilon})$ addressable, constant-depth
$CCZ$ gates within a single code block.

\section{Non-Clifford Gates via Twisted Cohomology Operations}
\label{sec:non-clifford-gates-via-twisted-cohomology-operations}

In the previous section, ordinary cup products were turned into
single-block gates using cochain endomorphisms. We now apply more general
coefficient-lifting operations constructed in
\Cref{sec:constructing-twisted-cohomology-operations}. We first consider
toric codes and then products of surfaces. As in the previous section, the logical
calculations take place entirely in cohomology, while the physical circuits
use the cochain formulas constructed earlier.

\subsection{Addressable $T$ gates on the four-torus}
\label{sec:the-four-torus}

Let
\[
 T^4=(S^1)^4,
 \qquad
 H^1(T^4;\Z)
 =
 \Z\langle\gamma_1,\gamma_2,\gamma_3,\gamma_4\rangle,
\]
where
\[
 \int_{T^4}
 \gamma_1\cup\gamma_2\cup\gamma_3\cup\gamma_4=1.
\]
Consider the coordinate rotation
\[
 \sigma(x_1,x_2,x_3,x_4)=(x_2,x_3,x_4,x_1).
\]
It has order four and reverses orientation, and therefore it defines the
$C_4$-twisted operation of
\Cref{sec:universal-c4-cocycle}.

Let $\bar\gamma_i$ denote the reduction of $\gamma_i$ modulo two. Write
\[
 [a]=
 \epsilon_1\bar\gamma_1+
 \epsilon_2\bar\gamma_2+
 \epsilon_3\bar\gamma_3+
 \epsilon_4\bar\gamma_4,
 \qquad
 \epsilon_i\in\F_2.
\]
Every class in $H^1(T^4;\F_2)$ lifts integrally. Evaluating on an integral
cocycle lift makes the correction terms vanish. For the integral lift
$\widetilde a=\sum_i\epsilon_i\gamma_i$, reading the $\epsilon_i$ as $0$ or
$1$, expanding the fourfold cup product gives
\begin{align}
 \mathfrak P^{(4)}_\sigma([a])
 &=\int_{T^4}
 \widetilde a\cup\sigma\widetilde a
 \cup\sigma^2\widetilde a\cup\sigma^3\widetilde a\notag\\
 &=\epsilon_1-\epsilon_2+\epsilon_3-\epsilon_4\notag\\
 &\qquad\quad
 -2(\epsilon_1\epsilon_3-\epsilon_2\epsilon_4)\notag\\
 &\qquad\quad
 +4(\epsilon_1\epsilon_2\epsilon_3
    +\epsilon_1\epsilon_2\epsilon_4
    +\epsilon_1\epsilon_3\epsilon_4
    +\epsilon_2\epsilon_3\epsilon_4)
 \pmod8.
 \label{eq:t4-cyclic-phase}
\end{align}
This is the determinant of the corresponding circulant matrix.
Thus the unmodified operation already contains logical $T$, $CS$, and
$CCZ$ gates within a single code block.

We next isolate the individual $T$ gates. Fold all circles except the
$i$-th as in \Cref{sec:circle-fold-projections}, then permute
coordinates to send its surviving degree-one class to $\bar\gamma_1$.
The resulting bounded-weight cochain map $F_i$ induces
\[
 [F_i a]=\epsilon_i\bar\gamma_1.
\]
By \Cref{lem:precomposition-by-cochain-maps}, precomposition
preserves representative independence and bounded depth.

Applying the twisted operation to $F_i a$, using the integral lift
$\epsilon_i\gamma_1$, gives
\begin{align*}
 \mathfrak P^{(4)}_\sigma([F_i a])
 &=\int_{T^4}
 \widetilde{F_i a}\cup\sigma\widetilde{F_i a}
 \cup\sigma^2\widetilde{F_i a}
 \cup\sigma^3\widetilde{F_i a}\\
 &=\int_{T^4}
 (\epsilon_i\gamma_1)\cup(\epsilon_i\gamma_2)
 \cup(\epsilon_i\gamma_3)\cup(\epsilon_i\gamma_4)
 =\epsilon_i\pmod8.
\end{align*}
The corresponding circuit therefore implements
$\overline{\op{T}}_i$ on any of the four logical qubits.

Sums of these cochain maps select logical parities. Applying the operation
to $(F_i+F_j)a$ and $(F_i+F_j+F_k)a$, for distinct indices, gives
respectively
\[
 \overline{\op{T}}_i\,
 \overline{\op{T}}_j\,
 \overline{\op{CS}}_{ij}^\dagger
\]
and
\[
 \overline{\op{T}}_i\,
 \overline{\op{T}}_j\,
 \overline{\op{T}}_k\,
 \overline{\op{CS}}_{ij}^\dagger\,
 \overline{\op{CS}}_{ik}^\dagger\,
 \overline{\op{CS}}_{jk}^\dagger\,
 \overline{\op{CCZ}}_{ijk}.
\]
Cancelling the individual $T$ factors and inverting gives an addressable
$CS$ on any pair; cancelling the $T$ and $CS^\dagger$ factors in the second expression then
isolates $CCZ$ on any triple. Composing these gives individually addressable
$T$, $CS$, and $CCZ$ gates on a four-torus, generating the diagonal third level of the
Clifford hierarchy up to global phase~\cite{CuiGottesmanKrishna2017}.

\subsection{The qutrit $T$ gate on a three-torus}
\label{sec:qutrits-and-the-three-torus}

Here, we use the universal cocycle from \Cref{sec:folded-ccz}. For $\theta=r=3$,
the cyclic resolution and a cochain-divisible universal cocycle give a
cohomology invariant
\[
 \mathfrak P^{(3)}_\sigma:
 H^1(X;\Z/3)\longrightarrow\Z/9
\]
for an oriented three-manifold $X$ with an orientation-preserving
automorphism $\sigma$ of order three.

We now apply this operation to $T^3=(S^1)^3$ with the coordinate rotation
\[
 \sigma(x_1,x_2,x_3)=(x_2,x_3,x_1).
\]
Choose generators
\[
 H^1(T^3;\Z)
 =
 \Z\langle\gamma_1,\gamma_2,\gamma_3\rangle
\]
such that
\[
 \int_{T^3}\gamma_1\cup\gamma_2\cup\gamma_3=1.
\]
Let $\bar\gamma_i$ denote the reduction of $\gamma_i$ modulo three. For
\[
 [a]
 =
 \epsilon_1\bar\gamma_1+
 \epsilon_2\bar\gamma_2+
 \epsilon_3\bar\gamma_3,
 \qquad
 \epsilon_i\in\Z/3,
\]
choose the integral lift $\widetilde a=\sum_i\epsilon_i\gamma_i$, reading
the $\epsilon_i$ in $\{0,1,2\}$. Evaluating on an integral cocycle lift makes
the correction terms in \Cref{eq:twisted-c3-cochain-formula} vanish. The
logical phase is
\begin{align}
 \mathfrak P^{(3)}_\sigma([a])
 &=\int_{T^3}
 \widetilde a\cup\sigma\widetilde a\cup\sigma^2\widetilde a
 =\det\operatorname{circ}
   (\epsilon_1,\epsilon_2,\epsilon_3)\notag\\
 &=\epsilon_1^3+\epsilon_2^3+\epsilon_3^3
   -3\epsilon_1\epsilon_2\epsilon_3
 \pmod9.
 \label{eq:qutrit-logical-action}
\end{align}

For qutrits, cube roots of unity suffice for $Z$ and $S$,
whereas the $T$ gate uses ninth roots. Let
\[
 \zeta_9=e^{2\pi i/9},
 \qquad
 \zeta_3=\zeta_9^3,
 \qquad
 \overline{\op{T}}
 =
 \operatorname{diag}(1,\zeta_9,\zeta_9^{-1}).
\]
The first three terms in \Cref{eq:qutrit-logical-action} give
$\overline{\op{T}}_1$,
$\overline{\op{T}}_2$, and
$\overline{\op{T}}_3$. The final term contributes
\[
 \zeta_9^{-3\epsilon_1\epsilon_2\epsilon_3}
 =
 \zeta_3^{-\epsilon_1\epsilon_2\epsilon_3},
\]
which is the inverse qutrit $CCZ$ gate. The complete logical action is therefore
\[
 \overline{\op{T}}_1\,
 \overline{\op{T}}_2\,
 \overline{\op{T}}_3\,
 \overline{\op{CCZ}}_{123}^{-1}.
\]

As for the four-torus, fold the two circles other than the $i$-th and
permute coordinates to send its surviving degree-one class to
$\bar\gamma_1$. Over $\F_3$, the resulting cochain map $F_i$ satisfies
\[
 [F_i a]=\epsilon_i\bar\gamma_1
\]
and
\[
 \mathfrak P^{(3)}_\sigma([F_i a])
 =\int_{T^3}
 \widetilde{F_i a}\cup\sigma\widetilde{F_i a}
 \cup\sigma^2\widetilde{F_i a}
 =\epsilon_i^3\pmod9.
\]
Thus the resulting circuit implements exactly
$\overline{\op{T}}_i$, giving an addressable qutrit
$T$ gate on any of the three logical qutrits.

\subsection{Squares of surfaces}
\label{sec:squares-of-surfaces}

Let $\Sigma$ be a closed oriented surface with an orientation-reversing
cellular involution
\[
 \tau:\Sigma\longrightarrow \Sigma.
\]
Every class in $H^1(\Sigma;\F_2)$ lifts integrally. Consequently, the logical
actions of the twisted operations on $\Sigma\times \Sigma$ can be calculated using
ordinary cup products of integral lifts.

First consider the order-four map
\[
 \sigma:\Sigma\times \Sigma\longrightarrow \Sigma\times \Sigma,
 \qquad
 \sigma(p,q)=(\tau q,p).
\]
It reverses the orientation of $\Sigma\times \Sigma$ and therefore defines the
$C_4$-twisted operation. Under the K\"unneth decomposition
\[
 H^1(\Sigma\times \Sigma;\Z)
 =
 H^1(\Sigma;\Z)\oplus H^1(\Sigma;\Z),
\]
write an integral lift as $(u,v)$, and denote its reduction modulo two by
$(\bar u,\bar v)$. The pullback action is $\sigma(u,v)=(v,\tau u)$.
Expanding the fourfold cup product gives
\begin{align}
 \mathfrak P^{(4)}_\sigma(\bar u,\bar v)
 ={}&
 \int_{\Sigma\times \Sigma}
 (u,v)\cup\sigma(u,v)
 \cup\sigma^2(u,v)\cup\sigma^3(u,v)\notag\\
 ={}&
 \left(\int_\Sigma u\cup\tau u\right)^2
 -\left(\int_\Sigma v\cup\tau v\right)^2\notag\\
 &+
 4\left(\int_\Sigma u\cup v\right)
  \left(\int_\Sigma u\cup\tau v\right)
 \pmod8.
 \label{eq:surface-square-c4}
\end{align}

The surface square also carries a Klein-four action. Define
\[
 s(p,q)=(\tau p,\tau q),
 \qquad
 \kappa(p,q)=(q,p).
\]
The maps $s$ and $\kappa$ are commuting orientation-preserving
involutions. They therefore define the $V_4$-twisted operation, whose
logical action is
\begin{align}
 \mathfrak P^{(4)}_{s,\kappa}(\bar u,\bar v)
 ={}&
 \int_{\Sigma \times \Sigma }
 (u,v)\cup s(u,v)\cup\kappa(u,v)\cup s\kappa(u,v)\notag\\
 ={}&
 \left(\int_\Sigma u\cup\tau u\right)^2
 +\left(\int_\Sigma v\cup\tau v\right)^2\notag\\
 &-
 2\left(\int_\Sigma u\cup v\right)^2
 -2\left(\int_\Sigma u\cup\tau v\right)^2
 \pmod8.
 \label{eq:surface-square-v4}
\end{align}

\Cref{eq:surface-square-c4,eq:surface-square-v4} determine the complete logical gates. After
choosing a basis of $H^1(\Sigma ;\Z)$, their coefficients give the corresponding
$T$, $CS$, $CZ$, $CCZ$, and $CCCZ$ factors. In particular, the
calculation uses only $H^1(\Sigma ;\Z)$, its intersection pairing, and the
induced action of $\tau$.

More generally, replacing $\tau$ by another orientation-reversing cellular
involution simply replaces its action on $H^1(\Sigma ;\Z)$. Precomposing the operation with a
bounded-weight cochain endomorphism likewise amounts to precomposing
\Cref{eq:surface-square-c4} or
\Cref{eq:surface-square-v4} with the induced linear map on
$H^1(\Sigma \times \Sigma ;\F_2)$.

\paragraph{The Bolza surface.}
\label{sec:bolza-square}

Return to the Bolza surface $\Sigma $ and the integral classes
\[
 \widetilde\gamma_1,\ldots,\widetilde\gamma_4
 \in H^1(\Sigma ;\Z)
\]
from \Cref{sec:bolza-surface}. Consider the orientation-reversing
involution acting by
\[
 \widetilde\gamma_1\longmapsto\widetilde\gamma_2,
 \qquad
 \widetilde\gamma_2\longmapsto\widetilde\gamma_1,
 \qquad
 \widetilde\gamma_3\longmapsto-\widetilde\gamma_4,
 \qquad
 \widetilde\gamma_4\longmapsto-\widetilde\gamma_3.
\]
For
\[
 u=
 x_1\widetilde\gamma_1+
 x_2\widetilde\gamma_2+
 x_3\widetilde\gamma_3+
 x_4\widetilde\gamma_4,
 \qquad
 x_i\in\{0,1\},
\]
the intersection pairing gives
\begin{equation}
 \int_\Sigma u\cup\tau u
 =
 x_1-x_2-x_3+x_4
 +2(x_1+x_2)(x_3+x_4)
 \pmod4.
 \label{eq:bolza-twisted-intersection}
\end{equation}
Together with the ordinary intersection pairing from
\Cref{sec:bolza-surface}, this determines every term in
\Cref{eq:surface-square-c4} and
\Cref{eq:surface-square-v4}.

For the cyclic action, evaluating on a single basis element gives
\[
 \mathfrak P^{(4)}_\sigma(\gamma_i,0)=1,
 \qquad
 \mathfrak P^{(4)}_\sigma(0,\gamma_i)=-1
 \pmod8.
\]
Thus the linear part of the logical gate is
\[
 \overline{\op{T}}_1
 \overline{\op{T}}_2
 \overline{\op{T}}_3
 \overline{\op{T}}_4
 \overline{\op{T}}_5^\dagger
 \overline{\op{T}}_6^\dagger
 \overline{\op{T}}_7^\dagger
 \overline{\op{T}}_8^\dagger.
\]
The remaining controlled phases are given by
\Cref{eq:surface-square-c4}.

For the Klein-four action,
\[
 \mathfrak P^{(4)}_{s,\kappa}(\gamma_i,0)
 =
 \mathfrak P^{(4)}_{s,\kappa}(0,\gamma_i)
 =1
 \pmod8.
\]
Its linear part is therefore
\[
 \overline{\op{T}}_1
 \overline{\op{T}}_2
 \overline{\op{T}}_3
 \overline{\op{T}}_4
 \overline{\op{T}}_5
 \overline{\op{T}}_6
 \overline{\op{T}}_7
 \overline{\op{T}}_8,
\]
with the remaining controlled phases specified by
\Cref{eq:surface-square-v4}.

These $T$ factors occur together with controlled phases. Unlike the
torus selectors, this calculation alone does not isolate the individual
$T$ gates.

\section{Addressable Gates using Relative Cohomology}
\label{sec:addressable-relative-cup-gates}

The primary example in this section constructs an addressable $T$ gate on a 3-torus by folding, which requires extending our formalism to operations in relative cohomology. To build up to this example, we first review how the $T$ gate on three glued copies of the toric code on a tetrahedron is obtained from a relative cohomology operation. This is an operation known in the literature~\cite{HsinKobayashiZhu2024,Bauer_2025,Davydova_2026} which was originally understood through the transversal $T$-gate on a color code tetrahedron~\cite{Bombin_2007}. We then show how to use coordinate reflections and relative cochains to obtain a $T$-gate on a cube, which is, to the best of our knowledge, a new operation. This can be extended to an addressable $T$ on a 3-torus by using the circle folds of \Cref{sec:circle-fold-projections}.

\Cref{sec:relative-t-further-examples} extends these ideas to higher dimensions and other manifolds, while \Cref{sec:relative-t-beyond-manifolds} gives examples beyond manifolds.

\subsection{A \texorpdfstring{$T$}{T} gate for three copies of the toric code on a tetrahedron}
\label{sec:relative-t-tetrahedron}

Let ${\mathcal D}=[p_0p_1p_2p_3]$ be an ordered, subdivided tetrahedron, let $\mathcal F_i$ be the face opposite $p_i$, and put ${\mathcal F}={\mathcal F}_0$, $\mathcal E=[p_1p_2]$, and $\eta=[p_0p_1]$. These denote oriented fundamental chains of the subdivided faces and edges, with $\partial {\mathcal D}={\mathcal F}_0-{\mathcal F}_1+{\mathcal F}_2-{\mathcal F}_3$. Consider three copies of the toric code described by a cochain complex on the tetrahedron. For cochains $b=(b_1,b_2,b_3)$, we impose the boundary conditions
\begin{equation}\label{eq:tetrahedral-relative-conditions}
 b_i|_{{\mathcal F}_i}=0\quad(i=1,2,3),\qquad (b_1+b_2+b_3)|_{\mathcal F}=0.
\end{equation}
In every degree, these conditions define a subcomplex $\mathcal A^*({\mathcal D})\subset C^*({\mathcal D};\Z/2)^3$. 

With these boundary conditions, the relative complex $\mathcal A^*({\mathcal D})$ has one-dimensional first cohomology (as we show below), detected by $\lambda(b)=\int_\eta b_1$ for $b\in\mathcal A^1({\mathcal D})$.

Write $\widetilde b_i=\iota(b_i)$ for the coefficientwise $0/1$-lift. We introduce the following cochain operation:
\begin{equation}\label{eq:tetrahedral-t-phase}
 \Phi_{\mathcal D}(b)=4\int_{\mathcal D}\widetilde b_1\cup\widetilde b_3\cup\widetilde b_2+2\int_{\mathcal F}\widetilde b_1\cup\widetilde b_2-\int_{\mathcal E}\widetilde b_1\pmod8.
\end{equation}
Expressed as a physical circuit, the three terms above give products of physical $CCZ$, $CS^{\pm1}$, and $T^{\pm1}$ gates on tetrahedra, triangles, and edges, respectively.

The operation $\Phi_{\mathcal D}(b)$ is a cohomology operation:

\begin{proposition}\label{prop:tetrahedral-relative-t}
    Consider the relative cochain complex $\mathcal A^*({\mathcal D})$ as discussed above. Then:
    \begin{enumerate}
        \item The map $\lambda$ identifies $H^1(\mathcal A^*({\mathcal D})) \cong \Z/2$, and thus, there is one logical qubit.
        \item The operation $\Phi_{\mathcal D}(b)$ is a cohomology operation corresponding to the logical $T$-gate implemented by $U_{\Phi_{\mathcal D}}$ on this logical qubit. Namely, on cocycles we have
        \begin{equation}\label{eq:tetrahedral-t-action}
        \Phi_{\mathcal D}(b)=\iota(\lambda(b))\in\Z/8.
        \end{equation}
    \end{enumerate}
\end{proposition}
\begin{proof}
Write $b_i=du_i$ with $u_i(p_0)=0$. Then $u_i|_{\mathcal F_i}=0$, and $u_1+u_2+u_3$ is constant on $\mathcal F$, with value $u_1(p_1)=\lambda(b)$. Thus $b$ is a coboundary in $\mathcal A^*({\mathcal D})$ precisely when this constant is zero. Both constants occur, by assigning the value one at each vertex of $\mathcal F$ to a component $u_i$ whose face $\mathcal F_i$ does not contain that vertex, and extending by zero elsewhere.

For a binary one-cocycle $b=du$, the canonical integer lifts satisfy $d\widetilde b=2\widetilde b\cup\widetilde b$ and $d\widetilde u=\widetilde b-2\widetilde u\cup\widetilde b$. The Leibniz rule therefore gives the following identities modulo eight
\begin{align*}
 d(4\widetilde u_1\cup\widetilde b_3\cup\widetilde b_2)&=4\widetilde b_1\cup\widetilde b_3\cup\widetilde b_2 &&\text{on }{\mathcal D}       ,\\
 d(2\widetilde u_1\cup\widetilde b_2)&=4\widetilde u_1\cup\widetilde b_3\cup\widetilde b_2+2\widetilde b_1\cup\widetilde b_2 &&\text{on }{\mathcal F},\\
 -d\widetilde u_1&=2\widetilde u_1\cup\widetilde b_2-\widetilde b_1 &&\text{on }{\mathcal E}.
\end{align*}
Here $b_3=b_1+b_2$ on ${\mathcal F}$ and $b_2=b_1$ on ${\mathcal E}$. Stokes' formula moves the three terms successively from ${\mathcal D}$ to ${\mathcal F}$ to ${\mathcal E}$ and all other faces vanish by~\Cref{eq:tetrahedral-relative-conditions}. The result is $\widetilde u_1(p_1)-\widetilde u_1(p_2)=\iota(\lambda(b))$, since $p_2\in {\mathcal F}_1$.
\end{proof}

\subsection{A \texorpdfstring{$T$}{T} gate on a cube}
\label{sec:relative-t-cube}

In this section, we extend the $T$ gate on the tetrahedron to the cube, and subsequently, we also extend this to a 3-torus. Before proceeding, let us outline the core of the main idea behind this construction. 
We start with a 3-torus $T_L^3$ with the usual cubic cellulation and even $L$. Let $Q=[0,L]^3$, subdivided into unit cubes. Consider the following chain of maps:
\begin{equation}\label{eq:relative-t-roadmap}
 C^1(T^3_L; \Z/2)
 \xrightarrow{\ F_x\ }
 C^1(Q,\partial_xQ; \Z/2)
 \xrightarrow{\ c\,\mapsto\,b(c)\ }
 \mathcal A^1({\mathcal D}; \Z/2)
 \xrightarrow{\ \Phi_{\mathcal D}\ }
 \Z/8,
\end{equation}
which contains the following steps.
\begin{enumerate}
 \item \emph{Selecting one direction on the torus.} More specifically, the map $F_x$ is a composition of two maps, which together take us from a torus to a cube with boundary conditions giving one-dimensional first cohomology,
 \begin{equation} \label{eq:pair-maps}
 C^1(T^3_L; \Z/2)
 \xrightarrow{\ 1+s_x^*\ }
 C^1(T^3_L,A_x; \Z/2)
 \xrightarrow{\ p^*\ }
 C^1(Q,\partial_xQ; \Z/2)
\end{equation}
A circle fold $s_x$ (introduced in \Cref{sec:circle-fold-projections})
 turns a cochain $a$ on the 3-torus into the cochain $a+s_x^*a$, which
 vanishes on the fixed set $A_x = I_+ \times T_L^2$. Pulling back along the
 quotient map $p\colon Q\to T^3_L$ gives a cochain on the cube
 which vanishes on the two faces $x=0$ and $x=L$.  On cohomology this step
 retains exactly the logical qubit belonging to the $x$-direction. The total cochain map $F_x = p^* (1+s_x^*)$ is the map we use for the precomposition in order to make a gate on a cube addressable.
 \item \emph{$T$ gate on the cube}. The cube is the union of six tetrahedra, which are
 permuted by the coordinate permutations.  One only needs to consider half of them. Intuitively, it is because half of the cube is topologically still a cube, so restricting to half a cube does not lose any homological information. Then, one can fold three tetrahedra belonging to a half of the cube into a single one. The action of the coordinate permutation can be then used to show that there exists a basis for the one-cocycles in which the boundary conditions coincide with the tetrahedral boundary conditions we have seen earlier. Finally, we can apply the relative cohomology operation $\Phi_{\mathcal D}$.
\end{enumerate}

\begin{figure}[t]
\centering
\begin{tikzpicture}[scale=3.4,>=Stealth]
 \coordinate (c000) at (0,0);
 \coordinate (c100) at (1,0);
 \coordinate (c010) at (0.45,0.3);
 \coordinate (c110) at (1.45,0.3);
 \coordinate (c001) at (0,1);
 \coordinate (c101) at (1,1);
 \coordinate (c011) at (0.45,1.3);
 \coordinate (c111) at (1.45,1.3);
 \fill[pattern=north east lines,pattern color=gray!60] (c000)--(c010)--(c011)--(c001)--cycle;
 \fill[pattern=north east lines,pattern color=gray!60] (c100)--(c110)--(c111)--(c101)--cycle;
 \draw[gray!70,dashed] (c000)--(c010);
 \draw[gray!70,dashed] (c010)--(c110);
 \draw[gray!70,dashed] (c010)--(c011);
 \draw[gray!70] (c000)--(c100)--(c101)--(c001)--cycle;
 \draw[gray!70] (c100)--(c110)--(c111)--(c101);
 \draw[gray!70] (c011)--(c111);
 \draw[gray!70] (c001)--(c011);
 \draw[red!60!black,thick] (c000)--(c010)--(c011)--(c001)--cycle;
 \draw[red!60!black,thick] (c100)--(c110)--(c111)--(c101)--cycle;
 \fill[blue,fill opacity=0.4] (c000)--(c001)--(c011)--cycle;
 \fill[blue,fill opacity=0.3] (c000)--(c001)--(c111)--cycle;
 \fill[blue,fill opacity=0.1] (c001)--(c011)--(c111)--cycle;
 \draw[dashed] (c000)--(c011);                        %
 \draw (c001)--(c111);                                %
 \draw (c000)--(c111);                                %
 \draw (c011)--(c111);                                %
 \draw[very thick,red!75!black] (c001)--(c011);       %
 \draw[very thick,blue!70!black] (c000)--(c001);      %
 \draw[very thick,blue!70!black,->] (c000)--(c100);   %
 \fill (c000) circle (0.45pt) node[below left=2pt,font=\small] {$p_0$};
 \fill (c001) circle (0.45pt) node[left=2pt,font=\small] {$p_1$};
 \fill (c011) circle (0.45pt) node[above=2pt,font=\small] {$p_2$};
 \fill (c111) circle (0.45pt) node[above right=2pt,font=\small] {$p_3$};
 \node[font=\small,blue!70!black,anchor=east] at (-0.06,0.35) {$\eta$};
 \node[font=\small,blue!70!black,anchor=north] at (0.5,-0.06) {$\eta_x=\rho(\eta)$};
 \node[font=\small,red!75!black] at (0.14,1.26) {$\mathcal E$};
 \node[font=\footnotesize,align=left,anchor=west] (faceTwoLabel)
   at (1.73,1.07) {$\mathcal F_2\colon x{=}y$\\$b_2=0$};
 \draw[gray!75,thin,->] (faceTwoLabel.west)--(1.57,1.07)--(0.80,0.88);
 \node[font=\footnotesize,anchor=south] (topFaceLabel)
   at (0.93,1.48) {$\mathcal F\colon z{=}L$};
 \draw[gray!75,thin,->] (topFaceLabel.south)--(0.93,1.24);
 \node[font=\footnotesize,red!60!black,align=left,anchor=west] (rightBoundaryLabel)
   at (1.73,0.60) {$\partial_x Q$\\$x{=}L$};
 \draw[red!60!black,thin,->] (rightBoundaryLabel.west)--(1.57,0.60)--(1.27,0.66);
 \node[font=\footnotesize,red!60!black,align=right,anchor=east] (leftBoundaryLabel)
   at (-0.34,0.85) {$\partial_x Q$\\$x{=}0$};
 \draw[red!60!black,thin,->] (leftBoundaryLabel.east)--(-0.16,0.85)--(0.08,0.90);
 \node[font=\footnotesize,anchor=east] (faceThreeLabel) at (-0.34,0.58) {$\mathcal F_3$};
 \draw[gray!75,thin,->] (faceThreeLabel.east)--(-0.16,0.58)--(0.13,0.62);
 \draw[->,gray] (-0.35,-0.3)--(-0.1,-0.3) node[right,font=\scriptsize] {$x$};
 \draw[->,gray] (-0.35,-0.3)--(-0.2,-0.2) node[above right=-2pt,font=\scriptsize] {$y$};
 \draw[->,gray] (-0.35,-0.3)--(-0.35,-0.05) node[above,font=\scriptsize] {$z$};
\end{tikzpicture}
\caption{The tetrahedron ${\mathcal D}=\{0\le x\le y\le z\le L\}$ of the cube
$Q=[0,L]^3$, with the faces $\mathcal F=\{z=L\}$, $\mathcal F_2=\{x=y\}$ and $\mathcal F_3=\{x=0\}$
shaded and the faces $x=0,L$ forming $\partial_x Q$
hatched. The coordinate permutation $\rho$ maps $\eta=[p_0p_1]$ to $\eta_x$.}
\label{fig:cube-chamber}
\end{figure}
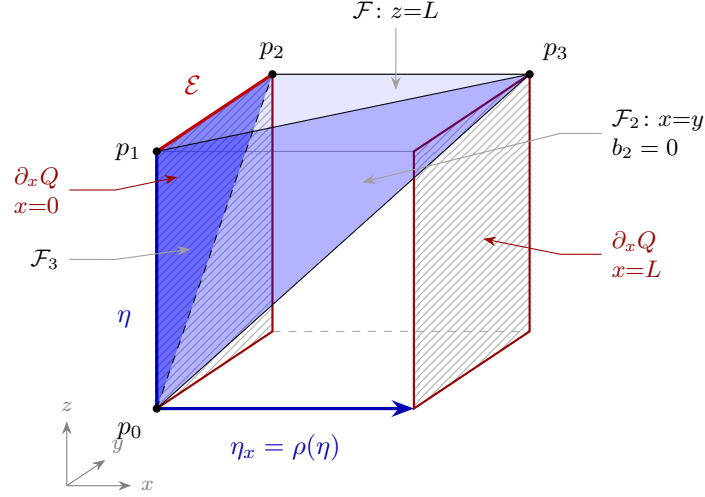

Now, assume that we start with $C(Q,\partial_xQ; \Z/2)$ with $\partial_x Q = \{x=0\}\cup\{x=L\}$. While the example on the tetrahedron has three input cochains, we will obtain these by folding the cube. Subdivide each unit cube into its six coordinate-order tetrahedra and order vertices lexicographically by $(x+y+z,x,y,z)$. The coordinate permutations
\begin{equation*}
 r_{12}(x,y,z)=(y,x,z),\qquad r_{23}(x,y,z)=(x,z,y),\qquad \rho=r_{12}r_{23}
\end{equation*}
preserve this triangulation and its simplex orders. Use the tetrahedron ${\mathcal D}=\{0\leq x\leq y\leq z\leq L\}$, with vertices
\begin{equation*}
 p_0=(0,0,0),\quad p_1=(0,0,L),\quad p_2=(0,L,L),\quad p_3=(L,L,L),
\end{equation*}
oriented in this order. In particular, $\mathcal F=\{z=L\}\cap {\mathcal D}$, $\mathcal F_1=\{y=z\}\cap {\mathcal D}$, $\mathcal F_2=\{x=y\}\cap {\mathcal D}$, and $\mathcal F_3=\{x=0\}\cap {\mathcal D}$.

Consider the three tetrahedra $\mathcal D$, $r_{12}\mathcal D$ and $\rho \mathcal D$ that form half of the cube: $\{(x,y,z) \in Q:  \ y \leq z\}$. For a cochain $c \in C^1(Q, \partial_x Q; \Z/2)$, folding $r_{12}\mathcal D$ and $\rho \mathcal D$ onto $\mathcal D$ gives the pullbacks $r_{12}^*c$ and $\rho^*c$, respectively. Now consider the following change of basis, restricted to ${\mathcal D}$:
\begin{equation}\label{eq:cube-reflection-inputs}
 b_1=r_{12}^*c+\rho^*c,\qquad b_2=c+r_{12}^*c,\qquad b_3=c.
\end{equation}
The first two of them vanish on the mirror planes $\mathcal F_1$ and $\mathcal F_2$, and the third vanishes on $\mathcal F_3\subset {\partial_x Q}$. Their sum is $\rho^*c$, which vanishes on $\mathcal F$ because $\rho(\mathcal F)\subset\{x=L\}$. Namely, for the new cochain $b = (b_1,b_2,b_3)$ defined above, the tetrahedral boundary conditions
\begin{equation}\label{eq:cube-tetrahedral-conditions}
 b_i|_{{\mathcal F}_i}=0\quad(i=1,2,3),\qquad (b_1+b_2+b_3)|_{\mathcal F}=0.
\end{equation}
are satisfied.
Thus~\Cref{eq:cube-reflection-inputs} is a cochain map into $\mathcal A^*({\mathcal D})$ that maps $c \mapsto b(c)$. These are pullbacks of absolute cochains followed by restriction: the individual reflections need not preserve ${\partial_x Q}$. In particular, these cochains do not need to be extended to the entire cube because they already satisfy the correct boundary conditions and vanish where appropriate.

Define 
\begin{equation}
\Phi_Q(c)=\Phi_{\mathcal D}(b(c)).
\end{equation}
We provide the full expression for completeness, with binary sums taken
before the coefficientwise $0/1$-lift:
\begin{align}
  \Phi_Q(c) &= 4\int_{\mathcal D}
  \iota(r_{12}^*c+\rho^*c)\cup\widetilde c\cup\iota(c+r_{12}^*c)\notag\\
  &\quad+2\int_{\mathcal F}
  \iota(r_{12}^*c+\rho^*c)\cup\iota(c+r_{12}^*c)
  -\int_{\mathcal E}\iota(r_{12}^*c+\rho^*c)\pmod8.
\end{align}

Since $r_{12}(\eta)\subset {\partial_x Q}$ and $\rho(\eta)=\eta_x$, where $\eta_x$ runs from $(0,0,0)$ to $(L,0,0)$, \Cref{prop:tetrahedral-relative-t} gives
\begin{equation}\label{eq:relative-cube-t-action}
 \Phi_Q(c)=\iota\!\left(\int_{\eta_x}c\right)\quad\text{for }c\in Z^1(Q,{\partial_x Q};\Z/2).
\end{equation}
The integral identifies $H^1(Q,{\partial_x Q};\Z/2)$ with $\Z/2$. A cocycle is $dt$, with $t$ constant on each of the two faces, and the integral is the difference of those constants. Hence, $U_{\Phi_Q}$ gives a $T$ gate on the relative cube code.

\subsection{An addressable \texorpdfstring{$T$}{T} gate on the three-torus}
\label{sec:relative-t-three-torus}

For even $L$, fold the first circle as in
\Cref{sec:circle-fold-projections} and put
$A_x=I_+\times T_L^2$.\footnote{For $L=2h+1$, use $I_+=[0,h]$ and set
$s(t)=t$ on $[0,h]$, $s(t)=h$ on $[h,h+1]$, and $s(t)=L-t$ on
$[h+1,L]$. This collapses one edge and gives the same bounded-weight
retraction and cohomological projection. The circle constructions below
therefore also work for odd $L$.}
The map $1+s_x^*$ takes values in $C^*(T^3_L,A_x;\Z/2)$. The
quotient $p:Q\to T^3_L$ sends both faces of $\partial_x Q$ into $A_x$, and it is a map that sends $(Q,\partial_x Q)\to(T^3_L,A_x)$ in \Cref{eq:pair-maps}.

On cochains, set
\begin{equation}\label{eq:relative-t-torus-formula}
 F_x=p^*(1+s_x^*),\qquad \Phi_{T^3}(a)=\Phi_Q(F_xa).
\end{equation}
The sum defining $F_xa$ is binary, before the lift.
\Cref{eq:relative-cube-t-action} gives
\begin{equation}\label{eq:relative-t-torus-action}
 \Phi_{T^3}(a)=\iota(\epsilon_x),\qquad U_{\Phi_{T^3}}\big|_{\mathcal C}
 =\overline{T}_x\otimes\mathrm{id}_y\otimes\mathrm{id}_z.
\end{equation}
Permuting coordinates gives $\overline{T}_y$ and
$\overline{T}_z$. \Cref{lem:precomposition-by-cochain-maps}
supplies representative independence and bounded depth. This is a gate
on one block of the cubic toric code. The cube $Q$ and the tetrahedron $\mathcal D$ with its boundary conditions are just the domain on which
the phase is evaluated.

\paragraph{Remark: other gates at 3rd level of Clifford hierarchy on a 3-torus.} We can also use the expression $\Phi_{Q_d}(\sum_{i} F_i a)$ is $\iota(\sum_{i}\epsilon_i)$, giving $\op{R}(2\pi/2^d)$ on the total parity. For $d=2$, selecting both directions gives $\overline{\op{S}}_1\overline{\op{S}}_2\overline{\op{CZ}}_{12}$. For $d=3$, selecting all three gives
\begin{equation}\label{eq:relative-t-parity-gate}
 \overline{T}_1\overline{T}_2\overline{T}_3\,
 \overline{CS}_{12}^{\dagger}\overline{CS}_{13}^{\dagger}\overline{CS}_{23}^{\dagger}\overline{CCZ}_{123}.
\end{equation}
This follows by expanding the binary parity as in \Cref{lem:phase-polynomial}. This expression can be made invariant under the cyclic coordinate permutation $\sigma$ by using $F_1(1+\sigma^*+\cdots+(\sigma^*)^{d-1})$; this sum is unchanged by $\sigma^*$ on all cochains. 

Similarly, the combinations of cochain maps $F_i$ discussed in \Cref{sec:selecting-cohomology-operations} give individually addressable $CS$ and $CCZ$ gates here as well.

\subsection{Higher-dimensional cubes, tori, and further examples}
\label{sec:relative-t-further-examples}

This approach can be extended to every dimension $d\geq2$. 
Given a cohomology operation $\Phi_{\mathcal D_d}$ on a $d$-simplex, combining it with the coordinate permutation and the circle fold, one obtains the logical phase gate
\begin{equation}\label{eq:relative-higher-root}
 \op{R}(2\pi/2^d),
\end{equation}
where $\op{R}(\varphi) = \op{diag}(1,e^{i\varphi})$ is the single-qubit phase gate.
Explicit operations for $d=2,3,4,5$ are given in~\cite{HsinKobayashiZhu2024}. Let us now discuss how this extends to addressable gates on $d$-dimensional tori.

\paragraph{The cohomology invariant on the simplex.} Let ${\mathcal D}_d=[p_0\cdots p_d]$ be a subdivided $d$-simplex, with $\mathcal F_i$ opposite $p_i$ and $\eta=[p_0p_1]$. The $d$ inputs satisfy
\begin{equation}\label{eq:simplex-relative-conditions}
 b_i|_{\mathcal F_i}=0\quad(1\leq i\leq d),\qquad \left.\sum_{i=1}^d b_i\right|_{\mathcal F_0}=0.
\end{equation}
Write $\mathcal A_d^*(\mathcal D_d)$ for this subcomplex of $C^*(\mathcal D_d;\Z/2)^d$. The argument with $b_i=du_i$ in \Cref{prop:tetrahedral-relative-t} identifies $H^1(\mathcal A_d^*(\mathcal D_d))$ with $\Z/2$ via $\lambda(b)=\int_\eta b_1$. The required phase is a formula on all constrained one-cochains with
\begin{equation}\label{eq:simplex-relative-root-phase}
 \Phi_{\mathcal D_d}:\mathcal A_d^1(\mathcal D_d)\longrightarrow\Z/2^d,
 \qquad \Phi_{\mathcal D_d}(b)=\iota(\lambda(b))\quad\text{on cocycles}.
\end{equation}

The simplest example of this is the $d=2$ gate on a folded triangle that gives the $S$-gate:
\begin{equation}\label{eq:relative-s-triangle}
 \Phi_{\mathcal D_2}(b)=2\int_{\mathcal D_2}\widetilde b_1\cup\widetilde b_2-\int_{[p_1p_2]}\widetilde b_1\pmod4.
\end{equation}

\paragraph{$d$-dimensional cube.} Give $Q_d=[0,L]^d$ the coordinate-order triangulation used in \Cref{sec:relative-t-cube}, put $\partial_x Q_d=\{x_1=0,L\}$, and take the simplex
\begin{equation*}
 \mathcal D_d=\{0\leq x_1\leq\cdots\leq x_d\leq L\},\qquad
 p_i=(\underbrace{0,\ldots,0}_{d-i},\underbrace{L,\ldots,L}_{i}).
\end{equation*}
Let $r_{k,k+1}$ interchange coordinates $k$ and $k+1$, and set $\tau_1=1$ and $\tau_k=r_{12}r_{23}\cdots r_{k-1,k}$, with the rightmost map applied first. For $c\in C^*(Q_d,\partial_x Q_d;\Z/2)$, define the following cochains on $\mathcal D_d$, taking all sums in $\Z/2$:
\begin{equation}\label{eq:general-chamber-inputs}
 a_k=\tau_k^*c,\qquad b_d=a_1,\qquad b_{d-k}=a_k+a_{k+1}\quad(1\leq k<d).
\end{equation}
Each pair $a_k,a_{k+1}$ agrees on the mirror $x_k=x_{k+1}$, so its binary sum vanishes on $\mathcal F_{d-k}$. Also $b_d=0$ on $\mathcal F_d=\{x_1=0\}$, and $\sum_i b_i=a_d$ vanishes on $\mathcal F_0=\{x_d=L\}$ because $\tau_d(\mathcal F_0)\subset\{x_1=L\}$. Thus~\Cref{eq:general-chamber-inputs} is a cochain map into $\mathcal A_d^*(\mathcal D_d)$.

\begin{proposition}\label{prop:relative-roots-cubes-tori}
Suppose the simplex phase in \Cref{eq:simplex-relative-root-phase} is a sum of terms of bounded size, with each cochain coordinate occurring in boundedly many terms, uniformly in $L$ for fixed $d$. Then
\begin{equation}\label{eq:general-cube-phase}
 \Phi_{Q_d}(c)=\Phi_{{\mathcal D}_d}(b(c))
 =\iota\!\left(\int_{\eta_1}c\right)\quad\text{on relative cocycles},
\end{equation}
where $\eta_1$ is the first coordinate segment from $0$ to $L$. It gives a bounded-depth logical phase gate $\op{R}(2\pi/2^d)$ on the relative cube code and, after a circle fold, an individually addressable $\op{R}(2\pi/2^d)$ on each of the $d$ logical qubits of the toric code on $T_L^d$ with qubits on edges.
\end{proposition}
\begin{proof}
The map $\tau_d$ sends $\eta$ to $\eta_1$, while $\tau_{d-1}(\eta)\subset \partial_x Q_d$. Consequently $\lambda(b(c))=\int_{\eta_1}c$, proving~\Cref{eq:general-cube-phase}. On the torus, fold circle $i$ and relabel coordinates to obtain $F_i:C^*(T_L^d;\Z/2)\to C^*(Q_d,\partial_x Q_d;\Z/2)$ as earlier. For $[a]=\sum_j\epsilon_j\bar\gamma_j$,
\begin{equation}\label{eq:general-torus-root}
 Q_i^{(d)}(a)=\Phi_{Q_d}(F_i a)=\iota(\epsilon_i),\qquad
 U_{Q_i^{(d)}}\big|_{\mathcal C}
 =\overline{\op{R}}_i(2\pi/2^d)\otimes\bigotimes_{j\ne i}\mathrm{id}_j.
\end{equation}
Every map used here has bounded weight in both directions for fixed $d$.
\end{proof}
Thus, this construction gives an addressable $\op{R}(2\pi/2^d)$ in $d$ dimensions.

\paragraph{Products with a circle.} Let $K$ be a finite connected regular cubical complex containing a square, and let $K_L$ be its subdivision into cubes of side $1/L$, for $L\geq3$. A chosen square gives a cellular inclusion $j:S_L^1\times[0,L]^2\hookrightarrow S_L^1\times K_L$. Restrict to this product subcomplex, fold its circle, and pull back to $Q$:
\begin{equation}\label{eq:relative-t-circle-product}
 F=p^*(1+s^*)j^*,\qquad \Phi_Q(Fa)=\iota\!\left(\int_{S^1\times\{v\}}a\right)\quad\text{on cocycles}.
\end{equation}
Here $p:Q\to S_L^1\times[0,L]^2$ identifies the $x$-faces, and $v$ is any vertex of the chosen square. By the K\"unneth splitting in \Cref{sec:circle-fold-projections}, this is $T$ on the circle qubit and the identity on the qubits from $K$. It applies to $S^1\times S^2$, $S^1\times\Sigma_g$, and $S^1\times S^3$, with respectively $1$, $1+2g$, and $1$ logical qubits. The depth is independent of $L$ and, for the surface example, of $g$, because only one square is used.

\paragraph{More general manifolds.} The same construction applies around an embedded circle in a manifold.
\begin{theorem}\label{thm:relative-t-manifold}
Let $M$ be a closed connected oriented smooth manifold of dimension $n\geq3$, and let $0\neq\zeta\in H_1(M;\Z/2)$. For $2\leq d\leq n$, assume a simplex phase as in \Cref{prop:relative-roots-cubes-tori}. There is a compatible cellulation for which the code with qubits on edges admits a bounded-depth logical $\op{R}(2\pi/2^d)$ on the bit $\langle[a],\zeta\rangle$ and the identity on the remaining bits in an adapted basis.
\end{theorem}
\begin{proof}
Represent $\zeta$ by an embedded circle. Its oriented normal bundle is trivial, so a tubular neighborhood contains $S^1\times {\mathcal D}^{d-1}$. Choose a compatible product cellulation and extend it over $M$. Restriction to this subspace, followed by its circle fold, gives a relative input to $\Phi_{Q_d}$ with logical value $\iota(\langle[a],\zeta\rangle)$. Product subdivisions preserve the required bounds.
\end{proof}
For example, this will give $T$ gates on $S^1\times S^2$, $S^1\times\Sigma_g$, and $S^1\times S^3$.

\subsection{Tricycle and tile codes}
\label{sec:relative-t-beyond-manifolds}

We now discuss more general three-dimensional codes with
translation-invariant bulk checks, including trivariate tricycle (TT) codes
with periodic boundary conditions and three-dimensional tile codes with
open boundary conditions. Both can be described using Koszul complexes.

Let $R$ be a commutative ring and $f,g,h\in R$. Their Koszul complex is
\begin{equation}\label{eq:relative-koszul-complex}
 K^\bullet(f,g,h)=
 \left[
 R\xrightarrow{\left(\begin{smallmatrix}f\\g\\h\end{smallmatrix}\right)}
 R^3\xrightarrow{\left(\begin{smallmatrix}
 0&h&-g\\
 -h&0&f\\
 g&-f&0
 \end{smallmatrix}\right)}
 R^3\xrightarrow{(f\;\;g\;\;h)}R
 \right],
\end{equation}
placed in degrees $-3,-2,-1,0$. Let $X$ be the common zero locus of
$f,g,h$. If they form a regular sequence, the associated complex of
sheaves resolves the structure sheaf $\mathcal O_X$. For the corresponding
infinite-lattice codes, this exactness gives topological order, as discussed
in~\cite{Haah2013Modules,BreuckmannChoeEberhardtPereiraSteffan2026}.
More generally, the Koszul complex models the functions on the derived intersection of the zero sets of $f$, $g$ and $h$.

For polynomials in $x,y,z$, a point $P=(p_x,p_y,p_z)\in X$ gives
\begin{equation*}
 (f,g,h)\subseteq(x-p_x,y-p_y,z-p_z)
\end{equation*}
after extending the coefficient field if necessary. We can therefore
choose a polynomial matrix $M_P$ with
\begin{equation}\label{eq:relative-koszul-comparison}
 (f,g,h)=(x-p_x,y-p_y,z-p_z)M_P.
\end{equation}
Its exterior powers define a map
\begin{equation*}
 K^\bullet(f,g,h)\longrightarrow
 K^\bullet(x-p_x,y-p_y,z-p_z),
\end{equation*}
whose component in degree $-j$ is $\bigwedge^j M_P$. On $H^0$, this is
restriction of functions to $P$,
\begin{equation*}
 R/(f,g,h)\longrightarrow R/(x-p_x,y-p_y,z-p_z).
\end{equation*}
At $P=(1,1,1)$, the target is $K^\bullet(x-1,y-1,z-1)$. With the periodic
or relative boundary conditions used below, this gives the reference
torus or relative cube complex, respectively. Realizing these comparisons
as bounded-weight cochain maps allows us to pull back the gates
constructed above.

In particular, suppose that a binary cochain map
\begin{equation*}
 F:C^\bullet\longrightarrow C^\bullet(Q,\partial_xQ;\F_2)
\end{equation*}
satisfies $\int_{\eta_x}Fa=\epsilon([a])$ on cocycles, where $\epsilon$
is a logical coordinate. Then \Cref{eq:relative-cube-t-action} gives
\begin{equation*}
 \Phi_Q(Fa)=\iota\bigl(\epsilon([a])\bigr),
\end{equation*}
which implements $T$ on that logical qubit. Bounded row and column
weights of $F$ give bounded circuit depth by
\Cref{lem:precomposition-by-cochain-maps}.

\paragraph{Trivariate tricycle codes.}
\label{sec:addressable-tt-example}
We return to the $[[108,15,(12,6)]]$ magic-tricycle code from
\Cref{sec:folded-nontopological-codes}, with defining polynomials $a,b,c$ over
\begin{equation*}
 R=\F_2[x,y,z]/(x^3-1,y^3-1,z^4-1)
 =\F_2[\mathcal G],
 \qquad
 \mathcal G=C_3\times C_3\times C_4.
\end{equation*}
Here we place the Koszul complex in degrees $0,1,2,3$, with qubits in
degree one, and write $u=(u_a,u_b,u_c)\in C^1$. The maximum $X$- and
$Z$-check weights are $12$ and $8$.

We construct five binary cochain maps
\begin{equation*}
 \Psi_1,\dots,\Psi_5:C^\bullet\longrightarrow
 C^\bullet(\mathcal G;\F_2)
\end{equation*}
to the cubic toric complex on the same group. Each map reads three
logical coordinates through the coordinate-circle integrals. Together
these give all fifteen logical coordinates.

Let $\omega^2+\omega+1=0$ and $t=1+z$. The Chinese remainder decomposition
in $x,y$ splits $C$ into five summands. Only the summands obtained by
evaluating at
\begin{equation*}
 (x,y)=(1,1),\qquad (x,y)=(\omega,\omega)
\end{equation*}
have nonzero cohomology. Denote these evaluations by $\rho_1,\rho_\omega$,
and put
\begin{equation*}
 S_1=\F_2[t]/(t^4),\qquad S_\omega=\F_4[t]/(t^4).
\end{equation*}
In the first summand, all three defining polynomials become $t^3$, and
replacing $(a,b,c)$ by $(a+b,a+c,a)$ gives $(0,0,t^3)$. In the second,
replacing them by $(a+Ac,b+Bc,Uc)$ gives $(0,0,t)$, where
\begin{equation}\label{eq:tt-field-change}
 \begin{aligned}
 U&=1+\omega t^2,\\
 A&=t+\omega^2t^2+\omega t^3,\\
 B&=\omega+\omega t+\omega^2t^3.
 \end{aligned}
\end{equation}
Consequently,
\begin{equation}\label{eq:tt-logical-sectors}
 H^1(C)\cong\bigl[(t)\oplus(t)\oplus S_1/(t^3)\bigr]\oplus
 \bigl[\F_4t^3\oplus\F_4t^3\oplus\F_4\bigr],
\end{equation}
of binary dimension $9+6=15$.

To extract the first nine coordinates, write $w=\sum_{j=0}^3w_jz^j$ and
$|w|=\sum_jw_j$. For $k=1,2,3$, define
\begin{equation}\label{eq:tt-layer-maps}
 \begin{aligned}
 \beta_k^0(w)&=|w|(1+z+\cdots+z^{k-1})+w_kt^3,\\
 \beta_k^1(w)&=w_0+w_kz^k.
 \end{aligned}
\end{equation}
The identity
\begin{equation*}
 t\beta_k^0(w)=\beta_k^1(t^3w)
\end{equation*}
makes these cochain maps from the two-term complex
$S_1\xrightarrow{t^3}S_1$ to $S_1\xrightarrow{t}S_1$. On the kernel $(t)$,
they read the coordinates $w_k$. On the cokernel $S_1/(t^3)$, their outputs $w_0+w_k$. These maps are binary-linear but not
$S_1$-linear, allowing them to distinguish the nilpotent components.

Let $
 e=(1+x+x^2)(1+y+y^2)$
and let $\theta_0,\theta_1:\F_4\to\F_2$ be the coordinate functionals
for the basis $1,\omega$, applied coefficientwise. Since
$(1+x)e=(1+y)e=0$, the required maps in degree one are
\begin{equation}\label{eq:tt-five-maps}
\begin{aligned}
 \Psi_k(u)&=e\bigl(\beta_k^0(v_a+v_b),\ \beta_k^0(v_a+v_c),\ \beta_k^1(v_a)\bigr),
 && v=\rho_1(u),\quad k=1,2,3,\\
 \Psi_{4+s}(u)&=e\,\theta_s\bigl(v_a+Av_c,\ v_b+Bv_c,\ Uv_c\bigr),
 && v=\rho_\omega(u),\quad s=0,1,
\end{aligned}
\end{equation}
The three output components correspond to the $x$-, $y$- and $z$-edges
of the torus.

Writing $\zeta_r$ for its coordinate circles, the fifteen integrals
\begin{equation}\label{eq:tt-logical-coordinates}
 \epsilon_{k,r}([u])=\int_{\zeta_r}\Psi_ku,
 \qquad
 1\leq k\leq5,\quad r\in\{x,y,z\},
\end{equation}
form a basis of $\operatorname{Hom}(H^1(C),\F_2)$. For the first three
maps, this follows from the kernel and cokernel coordinates above. The
remaining two maps extract the two binary coordinates of each of the
three $\F_4$-components.

To apply the torus phase, let $q:T_4^3\to\mathcal G$ collapse one edge in
each of the first two circles. Its pullback preserves the coordinate-circle
integrals. Let $\kappa_r$ exchange the $x$- and $r$-circles of $T_4^3$.
Then \Cref{eq:relative-t-torus-action} gives
\begin{equation}\label{eq:tt-addressable-t}
 \Phi_{k,r}(u)=\Phi_{T^3}\bigl(\kappa_r^*q^*\Psi_ku\bigr)
 =\iota\bigl(\epsilon_{k,r}([u])\bigr)
 \qquad(u\in Z^1(C)).
\end{equation}
Thus the code admits individually addressable $T$ gates on all fifteen
logical qubits within one block. Summing the binary maps
$\kappa_r^*q^*\Psi_k$ before applying $\Phi_{T^3}$ gives $T$ on the
corresponding logical parity, and therefore on each qubit of any
prescribed logical basis.

\paragraph{Three-dimensional tile codes.}
For open boundary conditions, consider the three polynomials
\begin{equation}\label{eq:relative-t-tile-polynomials}
 \begin{aligned}
 f&=y+yz+xz+x^2y,\\
 g&=1+xz+y^2z+xy^2z,\\
 h&=x+z^2+xy+xyz^2.
 \end{aligned}
\end{equation}
Homogenize them to multidegrees $(2,1,1)$, $(1,2,1)$, and $(1,1,2)$ on
$Y=(\mathbb P^1_{\F_2})^3$, and set
\begin{equation}\label{eq:relative-t-tile-complex}
 C^q=H^2\!\left(Y,K^{q-3}(f,g,h)\otimes\mathcal O(10,-3,-3)\right).
\end{equation}
In the monomial bases,
\begin{equation*}
 \dim C^1=545,\qquad
 \operatorname{rank}d^0=252,\qquad
 \operatorname{rank}d^1=277.
\end{equation*}
This gives a $[[545,16]]$ code with maximum $X$- and $Z$-check weights
$12$ and $8$. A numerical search found logical operators of weight $19$
on both sides, giving $d_X,d_Z\leq19$.

With the same twist, the tile complex for $x-1,y-1,z-1$ is the cochain
complex of a cuboid relative to its two opposite $x$-faces. Its $x$-factor
is an eleven-edge interval relative to both endpoints, while the other
factors are ordinary two-edge intervals. Their tensor product therefore
has the required relative boundary conditions.

Let $E$ be a finite extension of $\F_2$. For a common zero
$P\in(E^\times)^3$, the comparison defined by $M_P$
extends to the homogenized Koszul complexes. Applying the same twist
and taking $H^2$ gives a cochain map respecting these boundary
conditions. Normalizing the nonzero coordinates of $P$ and pulling
back from the cuboid to the reference cube gives
\begin{equation*}
 F_P:C^\bullet\longrightarrow C^\bullet(Q,\partial_xQ;E).
\end{equation*}
On cohomology, this reads the component associated with $P$.

For this example,
\begin{equation}\label{eq:relative-t-tile-algebra}
 \F_2[x,y,z]/(f,g,h)
 \cong\F_2\times\F_{16}\times\F_{2048}.
\end{equation}
The factors correspond to three Frobenius orbits of common zeros, of
sizes $1,4,11$. Choose one point in each orbit and a binary basis of its
residue field. Applying the dual coordinate functionals coefficientwise
to $F_P$ gives sixteen binary cochain maps
\begin{equation*}
 F_i:C^\bullet\longrightarrow C^\bullet(Q,\partial_xQ;\F_2).
\end{equation*}
Their integrals define logical coordinates $\epsilon_i$, and
\begin{equation}\label{eq:relative-t-tile-action}
 \Phi_Q(F_i a)
 =\iota\!\left(\int_{\eta_x}F_i a\right)
 =\iota\bigl(\epsilon_i([a])\bigr)
 \qquad(a\in Z^1(C)).
\end{equation}
Each $\Phi_Q\circ F_i$ therefore implements $T$ on one logical qubit and
the identity on the other fifteen. The maps have bounded weight in
both directions for fixed defining polynomials and bounded cuboid
aspect ratios, giving bounded-depth circuits.

\bibliographystyle{alphaurl}
\bibliography{references}

@article{Steenrod1947,
  author  = {Steenrod, Norman E.},
  title   = {Products of Cocycles and Extensions of Mappings},
  journal = {Annals of Mathematics},
  volume  = {48},
  number  = {2},
  pages   = {290--320},
  year    = {1947},
  doi     = {10.2307/1969172}
}

@article{Steenrod1953,
  author  = {Steenrod, Norman E.},
  title   = {Cyclic Reduced Powers of Cohomology Classes},
  journal = {Proceedings of the National Academy of Sciences of the United States of America},
  volume  = {39},
  number  = {3},
  pages   = {217--223},
  year    = {1953},
  doi     = {10.1073/pnas.39.3.217}
}

@article{Steenrod1957,
  author  = {Steenrod, Norman E.},
  title   = {Cohomology Operations Derived from the Symmetric Group},
  journal = {Commentarii Mathematici Helvetici},
  volume  = {31},
  pages   = {195--218},
  year    = {1957},
  doi     = {10.1007/BF02564359}
}

@article{McClureSmith2003,
  author  = {McClure, James E. and Smith, Jeffrey H.},
  title   = {Multivariable Cochain Operations and Little {$n$}-Cubes},
  journal = {Journal of the American Mathematical Society},
  volume  = {16},
  number  = {3},
  pages   = {681--704},
  year    = {2003},
  doi     = {10.1090/S0894-0347-03-00419-3}
}

@incollection{May1970,
  author    = {May, J. Peter},
  title     = {A General Algebraic Approach to {S}teenrod Operations},
  editor    = {Peterson, F. P.},
  booktitle = {The Steenrod Algebra and Its Applications},
  series    = {Lecture Notes in Mathematics},
  volume    = {168},
  pages     = {153--231},
  publisher = {Springer},
  address   = {Berlin},
  year      = {1970},
  doi       = {10.1007/BFb0058524}
}

@article{KMM2021,
  author  = {Kaufmann, Ralph M. and Medina-Mardones, Anibal M.},
  title   = {Cochain Level {M}ay--{S}teenrod Operations},
  journal = {Forum Mathematicum},
  volume  = {33},
  number  = {6},
  pages   = {1507--1526},
  year    = {2021},
  doi     = {10.1515/forum-2020-0296}
}

@book{May1972,
  author    = {May, J. Peter},
  title     = {The Geometry of Iterated Loop Spaces},
  series    = {Lecture Notes in Mathematics},
  volume    = {271},
  publisher = {Springer},
  address   = {Berlin},
  year      = {1972},
  doi       = {10.1007/BFb0067491}
}

@article{BergerFresse2004,
  author  = {Berger, Clemens and Fresse, Benoit},
  title   = {Combinatorial Operad Actions on Cochains},
  journal = {Mathematical Proceedings of the Cambridge Philosophical Society},
  volume  = {137},
  number  = {1},
  pages   = {135--174},
  year    = {2004},
  doi     = {10.1017/S0305004103007138}
}

@book{Thomas1957,
  author    = {Thomas, Emery},
  title     = {The Generalized {Pontrjagin} Cohomology Operations and Rings with Divided Powers},
  series    = {Memoirs of the American Mathematical Society},
  number    = {27},
  publisher = {American Mathematical Society},
  address   = {Providence, RI},
  year      = {1957},
  doi       = {10.1090/memo/0027}
}

@article{GottesmanChuang99,
  author  = {Gottesman, Daniel and Chuang, Isaac L.},
  title   = {Demonstrating the Viability of Universal Quantum Computation Using Teleportation and Single-Qubit Operations},
  journal = {Nature},
  volume  = {402},
  pages   = {390--393},
  year    = {1999},
  doi     = {10.1038/46503}
}

@article{CuiGottesmanKrishna2017,
  author  = {Cui, Shawn X. and Gottesman, Daniel and Krishna, Anirudh},
  title   = {Diagonal Gates in the {Clifford} Hierarchy},
  journal = {Physical Review A},
  volume  = {95},
  number  = {1},
  pages   = {012329},
  year    = {2017},
  doi     = {10.1103/PhysRevA.95.012329}
}

@article{CupsAndGates,
  author        = {Breuckmann, Nikolas P. and Davydova, Margarita and Eberhardt, Jens N. and Tantivasadakarn, Nathanan},
  title         = {{Cups and Gates I}: Cohomology Invariants and Logical Quantum Operations},
  journal       = {Communications in Mathematical Physics},
  volume        = {407},
  pages         = {86},
  year          = {2026},
  doi           = {10.1007/s00220-026-05570-z},
  eprint        = {2410.16250},
  archivePrefix = {arXiv},
  primaryClass  = {quant-ph}
}

@article{balancedproductcodes,
  author  = {Breuckmann, Nikolas P. and Eberhardt, Jens N.},
  title   = {Balanced Product Quantum Codes},
  journal = {IEEE Transactions on Information Theory},
  volume  = {67},
  number  = {10},
  pages   = {6653--6674},
  year    = {2021},
  doi     = {10.1109/TIT.2021.3097347}
}

@article{BreuckmannEberhardt2021,
  author        = {Breuckmann, Nikolas P. and Eberhardt, Jens Niklas},
  title         = {Quantum Low-Density Parity-Check Codes},
  journal       = {PRX Quantum},
  volume        = {2},
  number        = {4},
  pages         = {040101},
  year          = {2021},
  doi           = {10.1103/PRXQuantum.2.040101},
  eprint        = {2103.06309},
  archivePrefix = {arXiv},
  primaryClass  = {quant-ph}
}

@article{foldtransversalgates,
  author    = {Breuckmann, Nikolas P. and Burton, Simon},
  title     = {Fold-Transversal {Clifford} Gates for Quantum Codes},
  journal   = {Quantum},
  volume    = {8},
  pages     = {1372},
  year      = {2024},
  doi       = {10.22331/q-2024-06-13-1372}
}

@article{LinPryadko2024,
  author  = {Lin, Hsiang-Ku and Pryadko, Leonid P.},
  title   = {Quantum Two-Block Group Algebra Codes},
  journal = {Physical Review A},
  volume  = {109},
  number  = {2},
  pages   = {022407},
  year    = {2024},
  doi     = {10.1103/PhysRevA.109.022407}
}

@article{EberhardtSteffan2024,
  author        = {Eberhardt, Jens Niklas and Steffan, Vincent},
  title         = {Logical Operators and Fold-Transversal Gates of Bivariate Bicycle Codes},
  journal       = {IEEE Transactions on Information Theory},
  volume        = {71},
  number        = {2},
  pages         = {1140--1152},
  year          = {2025},
  doi           = {10.1109/TIT.2024.3521638},
  eprint        = {2407.03973},
  archivePrefix = {arXiv},
  primaryClass  = {quant-ph}
}

@article{Haah2013Modules,
  author        = {Haah, Jeongwan},
  title         = {Commuting {Pauli} Hamiltonians as Maps Between Free Modules},
  journal       = {Communications in Mathematical Physics},
  volume        = {324},
  number        = {2},
  pages         = {351--399},
  year          = {2013},
  doi           = {10.1007/s00220-013-1810-2},
  eprint        = {1204.1063},
  archivePrefix = {arXiv},
  primaryClass  = {quant-ph}
}

@article{BreuckmannChoeEberhardtPereiraSteffan2026,
  author        = {Breuckmann, Nikolas P. and Choe, Shin Ho and Eberhardt, Jens Niklas and Pereira, Francisco Revson Fernandes and Steffan, Vincent},
  title         = {Logical Operators and Derived Automorphisms of Tile Codes},
  journal       = {Communications in Mathematical Physics},
  volume        = {407},
  pages         = {205},
  year          = {2026},
  doi           = {10.1007/s00220-026-05672-8},
  eprint        = {2511.14589},
  archivePrefix = {arXiv},
  primaryClass  = {quant-ph}
}

@misc{TiewBreuckmann2026,
  author        = {Tiew, Ryan and Breuckmann, Nikolas P.},
  title         = {Copy-Cup Gates in Tensor Products of Group Algebra Codes},
  year          = {2026},
  eprint        = {2602.23307},
  archivePrefix = {arXiv},
  primaryClass  = {quant-ph}
}

@misc{JacobMcLauchlanBrowne2025,
  author        = {Jacob, Abraham and McLauchlan, Campbell and Browne, Dan E.},
  title         = {Single-Shot Decoding and Fault-Tolerant Gates with Trivariate Tricycle Codes},
  year          = {2025},
  eprint        = {2508.08191},
  archivePrefix = {arXiv},
  primaryClass  = {quant-ph}
}

@article{MenonMagicTricycles2026,
  author  = {Menon, Varun and Bonilla Ataides, J. Pablo and Mehta, Rohan and Gu, Andi and Tan, Daniel Bochen and Lukin, Mikhail D.},
  title   = {Magic Tricycles: Efficient Magic-State Generation with Finite Block-Length Quantum {LDPC} Codes},
  journal = {Physical Review X},
  volume  = {16},
  number  = {2},
  pages   = {021014},
  year    = {2026},
  doi     = {10.1103/ghhp-cytl},
  eprint  = {2508.10714},
  archivePrefix = {arXiv},
  primaryClass  = {quant-ph}
}

@misc{GolowichTamoZhu2026,
  author        = {Golowich, Louis and Tamo, Itzhak and Zhu, Guanyu},
  title         = {Improved Transversal Non-{Clifford} Gates from Cup Products},
  year          = {2026},
  eprint        = {2609.27801},
  archivePrefix = {arXiv},
  primaryClass  = {quant-ph},
  note          = {ECCC Report TR26-160}
}

@article{BravyiKoenig2013,
  author  = {Bravyi, Sergey and K{\"o}nig, Robert},
  title   = {Classification of Topologically Protected Gates for Local Stabilizer Codes},
  journal = {Physical Review Letters},
  volume  = {110},
  number  = {17},
  pages   = {170503},
  year    = {2013},
  doi     = {10.1103/PhysRevLett.110.170503}
}

@misc{HsinKobayashiZhu2024,
  author        = {Hsin, Po-Shen and Kobayashi, Ryohei and Zhu, Guanyu},
  title         = {Classifying Logical Gates in Quantum Codes via Cohomology Operations and Symmetry},
  year          = {2024},
  eprint        = {2411.15848},
  archivePrefix = {arXiv},
  primaryClass  = {quant-ph}
}

@article{Davydova_2026,
   title={Universal Fault-Tolerant Quantum Computation in 2D without Getting Tied in Knots},
   volume={16},
   ISSN={2160-3308},
   url={http://dx.doi.org/10.1103/vrty-qs5h},
   DOI={10.1103/vrty-qs5h},
   number={3},
   journal={Physical Review X},
   publisher={American Physical Society (APS)},
   author={Davydova, Margarita and Bauer, Andreas and Magdalena de la Fuente, Julio C. and Webster, Mark and Williamson, Dominic J. and Brown, Benjamin J.},
   year={2026},
   month=July }

@article{Bauer_2025,
   title={Low-overhead non-Clifford fault-tolerant circuits for all non-chiral abelian topological phases},
   volume={9},
   ISSN={2521-327X},
   url={http://dx.doi.org/10.22331/q-2025-03-25-1673},
   DOI={10.22331/q-2025-03-25-1673},
   journal={Quantum},
   publisher={Verein zur Forderung des Open Access Publizierens in den Quantenwissenschaften},
   author={Bauer, Andreas},
   year={2025},
   month=Mar, pages={1673} }

@article{Bombin_2007,
   title={Topological Computation without Braiding},
   volume={98},
   ISSN={1079-7114},
   url={http://dx.doi.org/10.1103/PhysRevLett.98.160502},
   DOI={10.1103/physrevlett.98.160502},
   number={16},
   journal={Physical Review Letters},
   publisher={American Physical Society (APS)},
   author={Bombin, H. and Martin-Delgado, M. A.},
   year={2007},
   month=Apr }

@article{BravyiKitaev05,
  author  = {Bravyi, Sergey and Kitaev, Alexei},
  title   = {Universal quantum computation with ideal {Clifford} gates and noisy ancillas},
  journal = {Phys. Rev. A},
  volume  = {71},
  pages   = {022316},
  year    = {2005},
  doi     = {10.1103/PhysRevA.71.022316},
  eprint  = {quant-ph/0403025},
  archivePrefix = {arXiv}
}

@misc{KapustinThorngren14anomalies,
  author  = {Kapustin, Anton and Thorngren, Ryan},
  title   = {Anomalies of discrete symmetries in various dimensions and group cohomology},
  year    = {2014},
  eprint  = {1404.3230},
  archivePrefix = {arXiv},
  primaryClass  = {hep-th}
}

@article{GaiottoKapustin16,
  author  = {Gaiotto, Davide and Kapustin, Anton},
  title   = {Spin {TQFTs} and fermionic phases of matter},
  journal = {Int. J. Mod. Phys. A},
  volume  = {31},
  pages   = {1645044},
  year    = {2016},
  doi     = {10.1142/S0217751X16450445},
  eprint  = {1505.05856},
  archivePrefix = {arXiv}
}

@article{ChenKapustinRadicevic18,
  author  = {Chen, Yu-An and Kapustin, Anton and Radi{\v{c}}evi{\'c}, {\DJ}or{\dj}e},
  title   = {Exact bosonization in two spatial dimensions and a new class of lattice gauge theories},
  journal = {Ann. Phys.},
  volume  = {393},
  pages   = {234--253},
  year    = {2018},
  doi     = {10.1016/j.aop.2018.03.024},
  eprint  = {1711.00515},
  archivePrefix = {arXiv}
}

@article{KubicaYoshidaPastawski15,
  author  = {Kubica, Aleksander and Yoshida, Beni and Pastawski, Fernando},
  title   = {Unfolding the color code},
  journal = {New J. Phys.},
  volume  = {17},
  pages   = {083026},
  year    = {2015},
  doi     = {10.1088/1367-2630/17/8/083026},
  eprint  = {1503.02065},
  archivePrefix = {arXiv}
}

@misc{Zhu25topological,
  author  = {Zhu, Guanyu},
  title   = {A topological theory for {qLDPC}: non-{Clifford} gates and magic state fountain on homological product codes with constant rate and beyond the {$N^{1/3}$} distance barrier},
  year    = {2025},
  eprint  = {2501.19375},
  archivePrefix = {arXiv},
  primaryClass  = {quant-ph}
}

@misc{LiLiLiuNguyen25,
  author  = {Li, Yiming and Li, Zimu and Liu, Zi-Wen and Nguyen, Quynh T.},
  title   = {Poincar{\'e} duality and multiplicative structures on quantum codes},
  year    = {2025},
  eprint  = {2512.21922},
  archivePrefix = {arXiv},
  primaryClass  = {quant-ph}
}

@misc{LiShaoWeiLiLiu26,
  author  = {Li, Zimu and Shao, Yuguo and Wei, Fuchuan and Li, Yiming and Liu, Zi-Wen},
  title   = {Theory of (co)homological invariants on quantum {LDPC} codes},
  year    = {2026},
  eprint  = {2603.25831},
  archivePrefix = {arXiv},
  primaryClass  = {quant-ph}
}

@misc{SahayLongKhemani26,
  author  = {Sahay, Rahul and Long, David M. and Khemani, Vedika},
  title   = {Computing with {qLDPC} codes by climbing the chain map hierarchy},
  year    = {2026},
  eprint  = {2609.02999},
  archivePrefix = {arXiv},
  primaryClass  = {quant-ph}
}

@misc{FuZhengLiLiu25,
  author  = {Fu, Esther Xiaozhen and Zheng, Han and Li, Zimu and Liu, Zi-Wen},
  title   = {No-go theorems for logical gates on product quantum codes},
  year    = {2025},
  eprint  = {2507.16797},
  archivePrefix = {arXiv},
  primaryClass  = {quant-ph}
}

@misc{Zhu25systolic,
  author  = {Zhu, Guanyu},
  title   = {Transversal non-{Clifford} gates on {qLDPC} codes breaking the {$\sqrt{N}$} distance barrier and quantum-inspired geometry with {$\mathbb{Z}_2$} systolic freedom},
  year    = {2025},
  eprint  = {2507.15056},
  archivePrefix = {arXiv},
  primaryClass  = {quant-ph}
}

@misc{LiPreskillXu25,
  author  = {Li, Christine and Preskill, John and Xu, Qian},
  title   = {Transversal dimension jump for product {qLDPC} codes},
  year    = {2025},
  eprint  = {2510.07269},
  archivePrefix = {arXiv},
  primaryClass  = {quant-ph}
}

@misc{LiLiLiu26almostgood,
  author  = {Li, Yiming and Li, Zimu and Liu, Zi-Wen},
  title   = {Transversal non-{Clifford} gates on almost-good quantum {LDPC} and quantum locally testable codes},
  year    = {2026},
  eprint  = {2604.01874},
  archivePrefix = {arXiv},
  primaryClass  = {quant-ph}
}

@misc{LiLiLiu26good,
  author  = {Li, Yiming and Li, Zimu and Liu, Zi-Wen},
  title   = {Transversal non-{Clifford} gates on good quantum locally testable codes},
  year    = {2026},
  eprint  = {2609.26691},
  archivePrefix = {arXiv},
  primaryClass  = {quant-ph}
}

@article{DijkgraafWitten,
  author  = {Dijkgraaf, Robbert and Witten, Edward},
  title   = {Topological Gauge Theories and Group Cohomology},
  journal = {Communications in Mathematical Physics},
  volume  = {129},
  number  = {2},
  pages   = {393--429},
  year    = {1990},
  doi     = {10.1007/BF02096988}
}

@article{Chen13,
  author  = {Chen, Xie and Gu, Zheng-Cheng and Liu, Zheng-Xin and Wen, Xiao-Gang},
  title   = {Symmetry Protected Topological Orders and the Group Cohomology of Their Symmetry Group},
  journal = {Physical Review B},
  volume  = {87},
  pages   = {155114},
  year    = {2013},
  doi     = {10.1103/PhysRevB.87.155114},
  eprint  = {1106.4772},
  archivePrefix = {arXiv}
}

@incollection{kapustin2015higher,
  author    = {Kapustin, Anton and Thorngren, Ryan},
  title     = {Higher Symmetry and Gapped Phases of Gauge Theories},
  booktitle = {Algebra, Geometry, and Physics in the 21st Century: Kontsevich Festschrift},
  series    = {Progress in Mathematics},
  volume    = {324},
  pages     = {177--202},
  publisher = {Birkh{\"a}user},
  year      = {2017},
  doi       = {10.1007/978-3-319-59939-7_5},
  eprint    = {1309.4721},
  archivePrefix = {arXiv}
}

@article{Gaiottoetal2015,
  author  = {Gaiotto, Davide and Kapustin, Anton and Seiberg, Nathan and Willett, Brian},
  title   = {Generalized Global Symmetries},
  journal = {Journal of High Energy Physics},
  volume  = {2015},
  number  = {2},
  pages   = {172},
  year    = {2015},
  doi     = {10.1007/JHEP02(2015)172},
  eprint  = {1412.5148},
  archivePrefix = {arXiv}
}

@article{Yoshida15,
  author  = {Yoshida, Beni},
  title   = {Topological Color Code and Symmetry-Protected Topological Phases},
  journal = {Physical Review B},
  volume  = {91},
  pages   = {245131},
  year    = {2015},
  doi     = {10.1103/PhysRevB.91.245131}
}

@article{Yoshida16,
  author  = {Yoshida, Beni},
  title   = {Topological Phases with Generalized Global Symmetries},
  journal = {Physical Review B},
  volume  = {93},
  pages   = {155131},
  year    = {2016},
  doi     = {10.1103/PhysRevB.93.155131}
}

@article{Yoshida17,
  author  = {Yoshida, Beni},
  title   = {Gapped Boundaries, Group Cohomology and Fault-Tolerant Logical Gates},
  journal = {Annals of Physics},
  volume  = {377},
  pages   = {387--413},
  year    = {2017},
  doi     = {10.1016/j.aop.2016.12.014}
}

@article{Barkeshli23,
  author  = {Barkeshli, Maissam and Chen, Yu-An and Huang, Sheng-Jie and Kobayashi, Ryohei and Tantivasadakarn, Nathanan and Zhu, Guanyu},
  title   = {Codimension-2 Defects and Higher Symmetries in (3+1){D} Topological Phases},
  journal = {SciPost Physics},
  volume  = {14},
  pages   = {065},
  year    = {2023},
  doi     = {10.21468/SciPostPhys.14.4.065}
}

@article{Barkeshli24_1,
  author  = {Barkeshli, Maissam and Chen, Yu-An and Hsin, Po-Shen and Kobayashi, Ryohei},
  title   = {Higher-Group Symmetry in Finite Gauge Theory and Stabilizer Codes},
  journal = {SciPost Physics},
  volume  = {16},
  pages   = {089},
  year    = {2024},
  doi     = {10.21468/SciPostPhys.16.4.089}
}

@article{Barkeshli24_2,
  author  = {Barkeshli, Maissam and Hsin, Po-Shen and Kobayashi, Ryohei},
  title   = {Higher-Group Symmetry of (3+1){D} Fermionic {$\mathbb{Z}_2$} Gauge Theory: Logical {CCZ}, {CS}, and {T} Gates from Higher Symmetry},
  journal = {SciPost Physics},
  volume  = {16},
  pages   = {122},
  year    = {2024},
  doi     = {10.21468/SciPostPhys.16.5.122}
}

@article{Chen2023higher,
  author  = {Chen, Yu-An and Tata, Sri},
  title   = {Higher Cup Products on Hypercubic Lattices: Application to Lattice Models of Topological Phases},
  journal = {Journal of Mathematical Physics},
  volume  = {64},
  number  = {9},
  pages   = {091902},
  year    = {2023},
  doi     = {10.1063/5.0095189}
}

@article{Chen23,
  author  = {Chen, Xie and Dua, Arpit and Hsin, Po-Shen and Jian, Chao-Ming and Shirley, Wilbur and Xu, Cenke},
  title   = {Loops in 4+1d Topological Phases},
  journal = {SciPost Physics},
  volume  = {15},
  pages   = {001},
  year    = {2023},
  doi     = {10.21468/SciPostPhys.15.1.001}
}

@misc{Wang2023,
  author  = {Wang, Yi-Fei and Wang, Yixu and Chen, Yu-An and Zhang, Wenjun and Zhang, Tao and Hu, Jiazhong and Chen, Wenlan and Gu, Yingfei and Liu, Zi-Wen},
  title   = {Efficient Fault-Tolerant Implementations of Non-{Clifford} Gates with Reconfigurable Atom Arrays},
  year    = {2023},
  eprint  = {2312.09111},
  archivePrefix = {arXiv},
  primaryClass  = {quant-ph}
}

@misc{Zhu2023,
  author  = {Zhu, Guanyu and Sikander, Shehryar and Portnoy, Elia and Cross, Andrew W. and Brown, Benjamin J.},
  title   = {Non-{Clifford} and Parallelizable Fault-Tolerant Logical Gates on Constant and Almost-Constant Rate Homological Quantum {LDPC} Codes via Higher Symmetries},
  year    = {2023},
  eprint  = {2310.16982},
  archivePrefix = {arXiv},
  primaryClass  = {quant-ph}
}

@article{PaetznickReichardt2013,
  author  = {Paetznick, Adam and Reichardt, Ben W.},
  title   = {Universal Fault-Tolerant Quantum Computation with Only Transversal Gates and Error Correction},
  journal = {Physical Review Letters},
  volume  = {111},
  pages   = {090505},
  year    = {2013},
  doi     = {10.1103/PhysRevLett.111.090505},
  eprint  = {1304.3709},
  archivePrefix = {arXiv}
}

@article{AndersonDuclosCianciPoulin2014,
  author  = {Anderson, Jonas T. and Duclos-Cianci, Guillaume and Poulin, David},
  title   = {Fault-Tolerant Conversion between the {Steane} and {Reed--Muller} Quantum Codes},
  journal = {Physical Review Letters},
  volume  = {113},
  pages   = {080501},
  year    = {2014},
  doi     = {10.1103/PhysRevLett.113.080501},
  eprint  = {1403.2734},
  archivePrefix = {arXiv}
}

@article{BombinGaugeCC,
  author  = {Bomb{\'i}n, H{\'e}ctor},
  title   = {Gauge Color Codes: Optimal Transversal Gates and Gauge Fixing in Topological Stabilizer Codes},
  journal = {New Journal of Physics},
  volume  = {17},
  number  = {8},
  pages   = {083002},
  year    = {2015},
  doi     = {10.1088/1367-2630/17/8/083002}
}

@article{HorsmanSurgery,
  author  = {Horsman, Dominic and Fowler, Austin G. and Devitt, Simon and Van Meter, Rodney},
  title   = {Surface Code Quantum Computing by Lattice Surgery},
  journal = {New Journal of Physics},
  volume  = {14},
  number  = {12},
  pages   = {123011},
  year    = {2012},
  doi     = {10.1088/1367-2630/14/12/123011}
}

@article{CohenSurgery,
  author  = {Cohen, Lawrence Z. and Kim, Isaac H. and Bartlett, Stephen D. and Brown, Benjamin J.},
  title   = {Low-Overhead Fault-Tolerant Quantum Computing Using Long-Range Connectivity},
  journal = {Science Advances},
  volume  = {8},
  number  = {20},
  pages   = {eabn1717},
  year    = {2022},
  doi     = {10.1126/sciadv.abn1717}
}

@article{CowtanSurgery,
  author  = {Cowtan, Alexander and Burton, Simon},
  title   = {{CSS} Code Surgery as a Universal Construction},
  journal = {Quantum},
  volume  = {8},
  pages   = {1344},
  year    = {2024},
  doi     = {10.22331/q-2024-05-14-1344}
}

@article{YoderPartitioning,
  author  = {Yoder, Theodore J. and Takagi, Ryuji and Chuang, Isaac L.},
  title   = {Universal Fault-Tolerant Gates on Concatenated Stabilizer Codes},
  journal = {Physical Review X},
  volume  = {6},
  pages   = {031039},
  year    = {2016},
  doi     = {10.1103/PhysRevX.6.031039}
}

@inproceedings{GrasslRoetteler2013,
  author    = {Grassl, Markus and Roetteler, Martin},
  title     = {Leveraging Automorphisms of Quantum Codes for Fault-Tolerant Quantum Computation},
  booktitle = {2013 IEEE International Symposium on Information Theory (ISIT)},
  pages     = {534--538},
  year      = {2013},
  doi       = {10.1109/ISIT.2013.6620283},
  eprint    = {1302.1035},
  archivePrefix = {arXiv}
}

@article{Moussa2016,
  author  = {Moussa, Jonathan E.},
  title   = {Transversal {Clifford} Gates on Folded Surface Codes},
  journal = {Physical Review A},
  volume  = {94},
  pages   = {042316},
  year    = {2016},
  doi     = {10.1103/PhysRevA.94.042316},
  eprint  = {1603.02286},
  archivePrefix = {arXiv}
}

@article{Quintavalle_2023,
  author  = {Quintavalle, Armanda O. and Webster, Paul and Vasmer, Michael},
  title   = {Partitioning Qubits in Hypergraph Product Codes to Implement Logical Gates},
  journal = {Quantum},
  volume  = {7},
  pages   = {1153},
  year    = {2023},
  doi     = {10.22331/q-2023-10-24-1153}
}

@misc{SayginelEtAl2024,
  author  = {Sayginel, Hasan and Koutsioumpas, Stergios and Webster, Mark and Rajput, Abhishek and Browne, Dan E.},
  title   = {Fault-Tolerant Logical {Clifford} Gates from Code Automorphisms},
  year    = {2024},
  eprint  = {2409.18175},
  archivePrefix = {arXiv},
  primaryClass  = {quant-ph}
}

@article{WebsterSearch,
  author  = {Webster, Mark A. and Quintavalle, Armanda O. and Bartlett, Stephen D.},
  title   = {Transversal Diagonal Logical Operators for Stabiliser Codes},
  journal = {New Journal of Physics},
  volume  = {25},
  number  = {10},
  pages   = {103018},
  year    = {2023},
  doi     = {10.1088/1367-2630/acfc5f}
}

@article{VasmerBrown3D,
  author  = {Vasmer, Michael and Browne, Dan E.},
  title   = {Three-Dimensional Surface Codes: Transversal Gates and Fault-Tolerant Architectures},
  journal = {Physical Review A},
  volume  = {100},
  pages   = {012312},
  year    = {2019},
  doi     = {10.1103/PhysRevA.100.012312}
}

@article{4DTCnonClifford,
  author  = {Jochym-O'Connor, Tomas and Yoder, Theodore J.},
  title   = {Four-Dimensional Toric Code with Non-{Clifford} Transversal Gates},
  journal = {Physical Review Research},
  volume  = {3},
  pages   = {013118},
  year    = {2021},
  doi     = {10.1103/PhysRevResearch.3.013118}
}

@misc{burton2020limitations,
  author  = {Burton, Simon and Browne, Dan},
  title   = {Limitations on Transversal Gates for Hypergraph Product Codes},
  year    = {2020},
  eprint  = {2012.05842},
  archivePrefix = {arXiv},
  primaryClass  = {quant-ph}
}

@article{BravyiHaah2012,
  author  = {Bravyi, Sergey and Haah, Jeongwan},
  title   = {Magic-State Distillation with Low Overhead},
  journal = {Physical Review A},
  volume  = {86},
  pages   = {052329},
  year    = {2012},
  doi     = {10.1103/PhysRevA.86.052329},
  eprint  = {1209.2426},
  archivePrefix = {arXiv}
}

@misc{lin2024transversal,
  author  = {Lin, Ting-Chun},
  title   = {Transversal Non-{Clifford} Gates for Quantum {LDPC} Codes on Sheaves},
  year    = {2024},
  eprint  = {2410.14631},
  archivePrefix = {arXiv},
  primaryClass  = {quant-ph}
}

@misc{golowich2024quantum,
  author  = {Golowich, Louis and Lin, Ting-Chun},
  title   = {Quantum {LDPC} Codes with Transversal Non-{Clifford} Gates via Products of Algebraic Codes},
  year    = {2024},
  eprint  = {2410.14662},
  archivePrefix = {arXiv},
  primaryClass  = {quant-ph}
}

@article{HsinKobayashiZhu25,
  title = {Non-Abelian Self-Correcting Quantum Memory and Transversal Non-Clifford Gate Beyond the $n^{1/3}$ Distance Barrier},
  author = {Hsin, Po-Shen and Kobayashi, Ryohei and Zhu, Guanyu},
  journal = {PRX Quantum},
  volume = {6},
  number = {4},
  pages = {040360},
  year = {2025},
  month = dec,
  publisher = {American Physical Society},
  doi = {10.1103/hkgn-hhqx}
}

@article{KobayashiZhuHsin26,
  title = {Clifford Hierarchy Stabilizer Codes: Transversal Non-Clifford Gates and Magic States},
  author = {Kobayashi, Ryohei and Zhu, Guanyu and Hsin, Po-Shen},
  journal = {Phys. Rev. Lett.},
  volume = {136},
  number = {25},
  pages = {250802},
  year = {2026},
  month = jun,
  publisher = {American Physical Society},
  doi = {10.1103/j2q4-bnc6}
}

@article{WarmanSchaferNameki26,
  title = {Constant-Depth Clifford-Hierarchy Gates via Non-Abelian Surface Codes},
  author = {Warman, Alison and Sch{\"a}fer-Nameki, Sakura},
  journal = {PRX Quantum},
  volume = {7},
  number = {3},
  pages = {033029},
  year = {2026},
  month = aug,
  publisher = {American Physical Society},
  doi = {10.1103/kfxl-rv7c}
}

@misc{HuangWarmanSchaferNamekiChen25,
  title        = {Hybrid Lattice Surgery: Non-Clifford Gates via Non-Abelian Surface Codes},
  author       = {Huang, Sheng-Jie and Warman, Alison and Sch{\"a}fer-Nameki, Sakura and Chen, Yanzhu},
  year         = {2025},
  eprint       = {2510.20890},
  archivePrefix = {arXiv},
  primaryClass = {quant-ph}
}

\end{document}